\documentclass[reqno,oneside]{amsart}
\usepackage{graphicx} 

\usepackage[utf8]{inputenc}
\usepackage[super]{nth}
\usepackage{xcolor}
\usepackage{comment,fancyhdr,tikz,amsmath,amssymb,amsthm,mathtools,centernot,caption,float,graphicx,makecell,array,float,changepage,verbatim,textcomp,parallel,geometry,ragged2e,xcolor,lastpage,bbm,mathrsfs,pgfkeys,mathabx}
\newcommand\join\vee
\newcommand\meet\wedge

\usepackage[
    cal=boondoxo,
    scr=boondox,
    frak=euler,
    bb=ncmbbr]{mathalpha}
\usepackage{bm}

    \usepackage[shortlabels]{enumitem}

\usepackage[export]{adjustbox}
\usepackage[
    colorlinks=false,
    citebordercolor=green,
    linkbordercolor=red
    ]{hyperref}

    \usepackage[capitalise]{cleveref}
    \crefformat{equation}{(#2#1#3)}
    \crefformat{enumi}{(#2#1#3)}
    \crefformat{section}{\S#2#1#3}
    \crefformat{subsection}{\S#2#1#3}
    \crefformat{subsubsection}{\S#2#1#3}
    \numberwithin{equation}{section}
    \newcommand{\crefdefpart}[2]{%
        \hyperref[#2]{\namecref{#1}~\labelcref*{#1}~\ref*{#2}}}
    
    \let\etoolboxforlistloop\forlistloop 
    \usepackage{autonum}
    \let\forlistloop\etoolboxforlistloop 
    \makeatletter
    \newcommand{\blx@noerroretextools}{}
    \makeatother
    
    \usepackage[
        backend=biber,
        style=alphabetic,
        maxbibnames=50,
        maxalphanames=50,
        maxcitenames=50
        ]{biblatex}
    \newcommand{\R}{\mathbb{R}}
    \newcommand{\C}{\mathbb{C}}
    
    \def\sph^#1{\mathbb S^{#1}}
    \newcommand\p[1]{\left(#1\right)}
    \newcommand\abs[1]{\left|#1\right|}

    \newcommand\jp[1]{\left\langle#1\right\rangle}
    
    \newcommand\eps\varepsilon
    \newcommand\pji\varphi
    \newcommand\sig\varsigma
    
    \newcommand\ui{[0,1]}
    \newcommand\beq{\begin{equation}}
    \newcommand\eeq{\end{equation}}
    \newcommand\inv{^{-1}{}}

    \def\XXint#1#2#3{{\setbox0=\hbox{$#1{#2#3}{\int}$ }
    \vcenter{\hbox{$#2#3$ }}\kern-.6\wd0}}

    \newtheoremstyle{thmst}
        {5pt}
        {5pt}
        {}
        {}
        {\bfseries}
        {.}
        {.5em}
        {}%
    \newtheoremstyle{rmkst}
        {5pt}
        {5pt}
        {}
        {}
        {\bfseries}
        {.}
        {.5em}
        {}%
    \newtheoremstyle{lmst}
        {5pt}
        {5pt}
        {}
        {}
        {\bfseries}
        {.}
        {5pt plus 1pt minus 1pt}
        {}%
    \newtheoremstyle{prpst}
        {5pt}
        {5pt}
        {}
        {}
        {\bfseries}
        {.}
        {.5em}
        {}%
    \newtheoremstyle{defst}
        {5pt}
        {5pt}
        {}
        {}
        {\bfseries}
        {.}
        {.5em}
        {}%
    \newtheoremstyle{wlst}
        {5pt}
        {5pt}
        {}
        {}
        {\bfseries}
        {.}
        {.5em}
        {}%
    \newtheoremstyle{asmpst}
        {5pt}
        {5pt}
        {}
        {}
        {\bfseries}
        {.}
        {.5em}
        {}%
        
    \theoremstyle{thmst}
        \newtheorem{theorem}{Theorem}[section]
        \newtheorem{conjecture}[theorem]{Conjecture}

        \crefname{conjecture}{Conjecture}{Conjectures}
    \theoremstyle{rmkst}
        \newtheorem{remark}[theorem]{Remark}
    \theoremstyle{defst}
        \newtheorem{definition}[theorem]{Definition}
    \theoremstyle{defst}
        
    \theoremstyle{lmst}
        \newtheorem{lemma}[theorem]{Lemma}
        
    \theoremstyle{prpst}
        \newtheorem{proposition}[theorem]{Proposition}
    \theoremstyle{wlst}
        
    \theoremstyle{asmpst}
        
    \theoremstyle{thmst}
        \newtheorem{corollary}[theorem]{Corollary}

        \newlist{defenum}{enumerate}{1} 
        \setlist[defenum]{label=(\roman*),ref=\thedefinition\,(\roman*)}
        \crefname{defenumi}{Definition}{Definitions}
        
        \newlist{lemenum}{enumerate}{1} 
        \setlist[lemenum]{label=(\roman*),ref=\thelemma\,(\roman*)}
        \crefname{lemenumi}{Lemma}{Lemmas}
        
        \newlist{corenum}{enumerate}{1} 
        \setlist[corenum]{label=(\roman*),ref=\thecorollary\,(\roman*)}
        \crefname{corenumi}{Corollary}{Corollaries}
        
        \newlist{thmenum}{enumerate}{1} 
        \setlist[thmenum]{label=(\roman*),ref=\thetheorem\,(\roman*)}
        \crefname{thmenumi}{Theorem}{Theorems}
        
        \newlist{rmkenum}{enumerate}{1} 
        \setlist[rmkenum]{label=(\roman*),ref=\theremark\,(\roman*)}
        \crefname{rmkenumi}{Remark}{Remarks}
        
        \newlist{prpenum}{enumerate}{1} 
        \setlist[prpenum]{label=(\roman*),ref=\theproposition\,(\roman*)}
        \crefname{prpenumi}{Proposition}{Propositions}
        
        \newlist{axenum}{enumerate}{1} 
        \setlist[axenum]{label=(\roman*),ref=\theaxiom\,(\roman*)}
        \crefname{axenumi}{Axiom}{Axioms}
        
        \newlist{defcrit}{enumerate}{1} 
        \setlist[defcrit]{label=(\roman*),ref=\thedefinition\,(\roman*)}
        \crefname{defcriti}{Definition}{Definitions}
        
        \newlist{lemcrit}{enumerate}{1} 
        \setlist[lemcrit]{label=(\alph*),ref=\thelemma\,(\alph*)}
        \crefname{lemcriti}{Lemma}{Lemmas}
        
        \newlist{thmcrit}{enumerate}{1} 
        \setlist[thmcrit]{label=(\alph*),ref=\thetheorem\,(\alph*)}
        \crefname{thmcriti}{theorem}{theorems}
        
        \newlist{rmkcrit}{enumerate}{1} 
        \setlist[rmkcrit]{label=(\alph*),ref=\theremark\,(\alph*)}
        \crefname{rmkcriti}{remark}{remarks}
        
        \newlist{prpcrit}{enumerate}{1} 
        \setlist[prpcrit]{label=(\alph*),ref=\theproposition\,(\alph*)}
        \crefname{prpcriti}{proposition}{propositions}

        \makeatletter
        \let\ifnc\@ifnextchar
        \makeatother

            \reversemarginpar
            \def\redit {\marginpar{\raggedleft{See edit $\implies$}}\color{red}}
            \def\rs#1.{\redit #1.\color{black}}
            \def\rsm#1.{{\color{red} #1.}}
            
            \def\bedit {\marginpar{\raggedleft{See edit $\implies$}}\color{blue}}
            \def\bs#1.{\bedit #1.\color{black}}
            \def\bsm#1.{{\color{blue} #1.}}

            \def\ind#1_#2{\left\{#1_{#2}\right\}}
            \def\<#1>{\jp{#1}}
            \def\-#1/{{}_{#1}}
            \makeatletter
                \newcommand\pl@write[3]{%
                    \left\|#1\right\|_{L^{#2}#3}%
                    }
                \def\pl #1__#2{%
                    \def\pl@arg@i{#1}%
                    \def\pl@arg@ii{#2}%
                    \def\pl@arg@iii{\alpha}%
                    \futurelet\next\pl@eval%
                    }
                \def\pl@eval{%
                    \ifx\next\bgroup%
                            \expandafter\pl@eval@iii%
                        \else%
                            \expandafter\pl@eval@ii%
                        \fi%
                    }
                \def\pl@eval@iii#1{%
                    \pl@write\pl@arg@i\pl@arg@ii{\p{#1}}%
                    }
                \def\pl@eval@ii{%
                    \pl@write\pl@arg@i\pl@arg@ii{}%
                    }
            \makeatother
            \makeatletter
                \newcommand\ef@write[3]{%
                    ^{#1\frac{#2}{#3}}{}%
                    }
                \newcommand\Ef@write[3]{%
                    ^{#1#2/#3}{}%
                    }
                \def\ef #1{%
                    \def\ef@arg@i{}%
                    \def\ef@arg@ii{#1}%
                    \futurelet\next\ef@eval%
                    }
                \def\ief #1{%
                    \def\ef@arg@i{-}%
                    \def\ef@arg@ii{#1}%
                    \futurelet\next\ef@eval%
                    }
                \def\ef@eval{%
                    \ifx\next/%
                            \expandafter\ef@eval@iii%
                        \else%
                            \expandafter\ef@eval@ii%
                        \fi%
                    }
                \def\ef@eval@iii/{%
                    \expandafter\ef@eval@v%
                    }
                \def\ef@eval@v#1{%
                    \Ef@write\ef@arg@i\ef@arg@ii{#1}}
                \def\ef@eval@ii{%
                    \expandafter\ef@eval@iv%
                    }
                \def\ef@eval@iv#1{%
                    \ef@write\ef@arg@i\ef@arg@ii{#1}%
                    }
            \def\e@writep #1{%
                ^{+#1}{}%
                }
            \def\e@writen #1{%
                ^{-#1}{}%
                }
            \def\e #1{%
                \ifx#1-%
                        \expandafter\e@writen%
                    \else%
                        \ifx#1+%
                                \expandafter\e@writep%
                            \else%
                                ^{#1}{}%
                            \fi%
                    \fi%
                }
            \def\e@2 {%
                }
            \newcommand\ie[1]{^{-#1}{}}

            \makeatother
            \makeatletter
            \newcommand{\undersim}[1]{\mathrel{\mathpalette\@undersim{#1}}}
            \newcommand{\@undersim}[2]{%
              \vcenter{%
                \ialign{%
                  ##\cr
                  $\m@th#1#2$\cr
                  \noalign{\nointerlineskip\kern.2ex}
                  $\m@th#1\sim$\cr
                  \noalign{\kern-.4ex}
                }%
              }%
            }
            \makeatother

    \makeatletter
    
    \def\col #1{%
        \def\col@arg{#1}%
        \begin{bmatrix}\col@i\end{bmatrix}%
        }
    \def\col@i{%
        \expandafter\col@ii\col@arg,\col@end%
        }
    \def\col@ii #1,{%
        #1%
        \futurelet\next\col@iii%
        }
    \def\col@iii{%
        \ifx\next\col@end%
                {}%
            \else%
                \\\col@ii%
            \fi%
        }
    \def\col@end{%
        {}%
        }
    \makeatother
\makeatletter\@ifundefined{remark}{\theoremstyle{remark}\newtheorem*{remark}{Remark}}{}\makeatother

\newcommand{\Qsq}{\mathcal Q} 
\newcommand{\Hsp}{\mathfrak H}          
\newcommand{\Haut}{\Hsp_{\mathrm{aut}}} 
\newcommand{\Hinf}{\Hsp^{\infty}}       
\newcommand{\Csp}{\mathfrak C}          
\newcommand{\Fsp}{\mathfrak F}          
\newcommand{\dom}{\operatorname{dom}}
\newenvironment{restatedlemma}[1]{\renewcommand\thetheorem{\ref*{#1}}\renewcommand\theHtheorem{restated.#1}\begin{lemma}}{\end{lemma}\addtocounter{theorem}{-1}}
\title[A counterexample to the {Ehlers--Kundt} conjecture]{On the completeness of gravitational pp-wave spacetimes: A counterexample to the {Ehlers--Kundt} conjecture}

\author{Hannah Cairo}
\date{September 2026}
\begin{document}
    \let\startmathold\[
    \let\endmathold\]
    \def\[{\futurelet\next\testingforaline}
    \def\aline{hello world}
    \def\testingforaline{\ifx\next\aline\expandafter\alining\else\startmathold\fi}
    \def\alining\aline#1\]{\begin{align}#1\end{align}}
\begin{abstract}
    This paper concerns the geometry of gravitational pp-wave spacetimes. We give a counterexample to the Ehlers--Kundt conjecture. In fact, we show that geodesic completeness is generic, in the sense of Baire category, among gravitational pp-waves. We also define a concept of ``universal pp-wave'', i.e. a pp-wave whose metric approximates that of every other gravitational pp-wave arbitrarily closely on arbitrarily large compact sets, and we show that there are geodesically complete universal pp-waves. Our machinery also suffices to give a new proof of the polynomial case of the Ehlers--Kundt conjecture (previously proved by Flores and S\'anchez).
\end{abstract}
\maketitle
\setcounter{tocdepth}{1}
\tableofcontents
\setcounter{tocdepth}{2}
\section{Introduction}\label{s-introduction}
This paper concerns the Ehlers--Kundt conjecture, a hypothetical classification of the geodesic completeness of gravitational pp-wave spacetimes:
    \begin{conjecture}\label{ek-conjecture}
        Any gravitational pp-wave that is geodesically complete is isometric to a plane wave spacetime.
    \end{conjecture}
A pp-wave, i.e. a ``plane-fronted wave with parallel rays'', is a spacetime $M=(\R^4,g)$ 
so that in some coordinates $(u,v,x,y)$, we have
    \[g=2dudv-H(u,x,y)du^2+dx^2+dy^2\]
Such spacetimes were first studied by Brinkmann (\cite{Brinkmann_1925}) and model radiation traveling at the speed of light. When $H(u,x,y)$ is harmonic in $(x,y)$, the spacetime is Ricci-flat and satisfies the vacuum Einstein field equations; such spacetimes are known as \textit{gravitational pp-waves} since they model gravitational waves propagating through spacetime mostly devoid of matter. When $H(u,x,y)$ is quadratic in $(x,y)$, it is called a plane-wave spacetime and admits many symmetries.
 
The reader is encouraged to read the article \cite{Roche_2023} for a thorough discussion of the history behind pp-waves and the Ehlers--Kundt conjecture. The completeness of gravitational pp-waves has important physical applications. As Flores and S\'anchez put it (\cite[Section~1.1]{Flores_2020}, summarizing the discussion in \cite{Ehlers_Kundt_1962}), a counterexample to \cref{ek-conjecture} would represent a graviton field ``independent of any matter by which it would be generated''.
 This is because completeness prevents the existence of any singularities (which would lie outside the spacetime), and the Ricci-flat condition implies that it has no internal material source either.

\subsection{Main results}\label{ss-main-results}
The primary result of this paper is a family of counterexamples to \cref{ek-conjecture}. That is, there are geodesically complete gravitational pp-waves which are not plane waves:
    \begin{theorem}[Counterexample to the Ehlers--Kundt conjecture]\label{t-counterexample}~\\
        There exists a geodesically complete gravitational pp-wave which is not a plane wave. More precisely, there exists an entire function $\mathcal f:\C\to\C$ which is not a polynomial of degree $\leq2$, so that the gravitational pp-wave
            \[g=2dudv+2\Re\p{\mathcal f(x+iy)}du^2+dx^2+dy^2\]
        (i.e. $H(u,x,y)=-2\Re\p{\mathcal f(x+iy)}$) is geodesically complete. In particular, \cref{ek-conjecture} is false.
        \end{theorem}
The examples are not isolated: in fact, completeness is generic. Among the gravitational pp-waves with $H$ independent of $u$, topologized by the uniform convergence of $H$ on compact sets, the geodesically complete ones form a dense $G_\delta$ set, and the same holds among all gravitational pp-waves with the topology of $C^\infty$ convergence in $u$, uniformly on compact sets (\cref{t-baire}). Combining this with a classical theorem of Birkhoff, we obtain a single complete gravitational pp-wave which approximates all gravitational pp-waves, after a translation. We call such pp-waves ``universal'':
    \begin{theorem}[Universal pp-waves]\label{t-universal-intro}
        \begin{enumerate}
            \item[(a)] There exists a geodesically complete gravitational pp-wave $(\R^4,g)$ which is \emph{$u$-universal} in the following sense: for every gravitational pp-wave $(\R^4,g')$, every compact set $K\subset\R^4$, every integer $k\geq0$ and every $\eps>0$, there is a translation $\tau(u,v,x,y)=(u+u_0,v,x+x_0,y+y_0)$ so that $\|\tau^*g-g'\|_{C^k(K)}\leq\eps$.
            \item[(b)] There exists a geodesically complete gravitational pp-wave with $H$ independent of $u$ which is \emph{universal} in the following sense: the statement of (a) holds for every gravitational pp-wave $(\R^4,g')$ for which $H'$ is independent of $u$, and the translation may be taken with $u_0=0$.
            \end{enumerate}
        \end{theorem}
See \cref{ss-universal} for details and definitions about universal pp-waves.

Our machinery also suffices to give a new proof of the polynomial case of \cref{ek-conjecture}, which was previously proved in \cite{Flores_2020}:
    \begin{theorem}[Polynomial case; cf.\ {\cite{Flores_2020}}]\label{t-polynomial-intro}
        Let $n\geq3$, and let $(M,g)$ be a gravitational pp-wave such that $H(u,\cdot)$ is a polynomial of degree $\leq n$ in $(x,y)$ for every $u\in\R$, whose coefficients are $C^1$ functions of $u$, and such that $H(u,\cdot)$ has degree exactly $n$ for at least one $u$. Then $(M,g)$ is geodesically incomplete.
        \end{theorem}
The polynomial EK conjecture was proved in 2020 by Flores and S\'anchez (\cite{Flores_2020}), so this is not a novel result; we include it since the machinery we develop allows us to give a rather short proof (see \cref{s-polynomial-conjecture}).

\subsection{Relation to previous work}\label{ss-related-work}
The conjecture goes back to Ehlers and Kundt in 1962 (\cite[Section 2-5.7]{Ehlers_Kundt_1962}), where it is posed with the addendum ``no matter which topology one chooses''. Ehlers and Kundt proved that plane waves are complete; more generally, a pp-wave is complete whenever $-H(u,\cdot)$ grows at most quadratically, locally uniformly in $u$ (\cite{Candela_2003,Candela_Romero_Sanchez_2013}; see also the survey \cite{Sanchez_2015}). Since a harmonic function on $\R^2$ which is bounded on one side by a polynomial of degree $m$ is itself a polynomial of degree at most $m$, such a bound forces a gravitational pp-wave to be a plane wave, so the conjecture is only interesting beyond quadratic growth. The strongest positive result is due to Flores and S\'anchez \cite{Flores_2020}, who proved the conjecture whenever $H$ is \textit{polynomially $u$-bounded}, i.e. whenever (in our sign convention) $-H(u,\cdot)$ is bounded above by a polynomial in $(x,y)$, locally uniformly in $u$; they reduce completeness to that of the planar dynamical system $\ddot{\bm z}=-\nabla_zV(\bm z,t)$ with $V$ harmonic in $z$, and show that a harmonic polynomial potential of degree $\geq3$ always has trajectories which escape to infinity in finite time inside a narrow sector. We give another proof of this result (\cref{t-polynomial-intro}). In particular, any counterexample must have $H(u,\cdot)$ non-polynomial for some $u$, and Flores and S\'anchez (\cite{Flores_2020}) raise the possibility that such a counterexample, if it existed, ``might admit some interpretation as a sort of exotic state of vacuum''. The examples of \cref{t-counterexample} are of this type: $H=-2\Re\p{\mathcal f}$ does not depend on $u$, and $\mathcal f$ is a non-polynomial entire function, which can be obtained iteratively as the limit of $\mathcal f_i=A_i-A_ie^{-A_i\inv p_{i-1}(z)}$ with $p_{i-1}$ a polynomial approximation of $\mathcal f_{i-1}$ (see \cref{ss-construction}). Such examples can be chosen to approximate every gravitational pp-wave uniformly on compact sets (translate the universal examples of \cref{t-universal-intro}); thus the counterexamples do not constitute some kind of isolated examples.
 
Other partial results assume additional structure. Costa e Silva, Flores and Herrera \cite{silva2016rigiditygeodesiccompletenessbrinkmann} proved that a geodesically complete, \textit{strongly causal}, autonomous gravitational pp-wave is a plane wave (indeed a Cahen--Wallach space), as a special case of a rigidity theorem for Ricci-flat Brinkmann spacetimes which admit a Killing field conjugate to the parallel null field; see also \cite{silva2016splittingproblemlorentzianmanifolds}. In the companion paper \cite{Cairo_causality} we prove that all distinguishing gravitational pp-waves are plane waves, so the assumption on completeness in \cite{silva2016rigiditygeodesiccompletenessbrinkmann} is unnecessary. The counterexamples of \cref{t-counterexample} are autonomous, so they are consistent with this result only because they fail to be strongly causal; indeed, by \cite{Cairo_causality}, they are not even distinguishing. Regarding the topological addendum, Leistner and Schliebner \cite{Leistner_Schliebner_2016} showed that compact pp-waves are geodesically complete and that every Ricci-flat compact pp-wave is a plane wave, which settles the compact version of the problem; the completeness statement was extended to all compact Brinkmann spacetimes by Mehidi and Zeghib \cite{Mehidi_Zeghib_2022}.
Finally, the conjecture is known to fail at low regularity: impulsive pp-waves, where $H(u,x,y)=f(x,y)\delta(u)$, are complete for every smooth profile $f$ \cite{frauenberger2025failureehlerskundtconjectureimpulsive} (see also \cite{Samann_Steinbauer_Svarc_2016}), so that completeness is no longer a rigidity criterion for distributional profiles. The examples in this paper show that the same is true for smooth (indeed real-analytic) metrics, i.e. in the setting in which the conjecture was originally posed.
 
Nevertheless our argument uses the fact that $\mathcal f$ is holomorphic in several places, both in the change of variables of \cref{ss-coordinate-transformation} and in the passage from local to global completeness in \cref{s-localization}, which relies on the density of the functions $A-Ae^{-A\inv p(z)}$ among entire functions and on the Cauchy estimates. We are not aware of previous results on the completeness of $\ddot{\bm z}=-\nabla V(\bm z)$ for a non-polynomial harmonic potential $V$; this was raised as an open problem in \cite{Flores_2020}, together with its analogue in higher dimensions.
 
Our use of the word ``universal'' follows complex analysis rather than relativity. Birkhoff \cite{Birkhoff_1929} constructed an entire function whose translates are dense in the space of all entire functions (see the survey \cite{Grosse_Erdmann_1999}), and \cref{t-universal} is the analogous statement for the profile of a \textit{complete} gravitational pp-wave (in the $u$-independent and in the general setting); the proof is Birkhoff's, in the Baire category form in which it is presented in \cite{Grosse_Erdmann_1999}. This should not be confused with the \textit{universal spacetimes} of \cite{Coley_Gibbons_Hervik_Pope_2008,Hervik_Pravda_Pravdova_2014}, i.e. metrics which solve the vacuum field equations of every higher-order theory of gravity. (Gravitational pp-waves are the prototypical examples of universal spacetimes in that sense, which is an additional reason to keep the two notions apart.)

\subsection{Outline of the paper}\label{ss-outline}
In \cref{s-geodesics}, we recall how the completeness problem for a pp-wave reduces to the global existence of solutions to the complex ODE $\ddot{\bm z}=\overline{\mathcal f'(t;\bm z)}$, where $\mathcal f$ is entire in $z$ and $H=-2\Re\p{\mathcal f}$ (\cref{t-reduction,c-ek-false}); the computations are classical, and we only cite them. In \cref{ss-appendix-continuity}, we set up the function spaces in which we work and record a single elementary lemma (\cref{l-continuity}): the solutions, the flow, and the first crossings of a hypersurface depend continuously on the (possibly time-dependent) function $\mathcal f$ and on the initial data.

In \cref{s-localization}, we state the main result in its ODE form (\cref{main-theorem}). We then introduce the notion of a \textit{locally complete} function (\cref{d-locally-complete}), for which global existence is only required for initial data in a compact set, and we show by a Baire category argument that \cref{main-theorem} (and, in fact, the genericity of completeness, \cref{t-baire}) follows from the existence of a dense family of locally complete functions (\cref{l-locally-complete-family}). The same argument produces complete pp-waves which are universal, i.e. whose translates approximate every other gravitational pp-wave (\cref{ss-universal}). In \cref{ss-construction}, we describe what the resulting functions look like, by spelling out the iteration which the Baire category theorem hides, and we illustrate one step of it in \cref{fig-stage2}. In \cref{ss-locally-complete-candidate}, we use the scaling symmetries of the ODE to reduce \cref{l-locally-complete-family} first to the local completeness of $-Ae^{A\inv p(z)}$ for polynomials $p$ of degree $\geq3$ (\cref{l-poly-exp}), and then to the local completeness of $1-e^{-z^n}$ and of its small (possibly time-dependent) perturbations (\cref{main-lemma}).
 
In \cref{ss-coordinate-transformation}, we introduce a complex change of coordinates, which simplifies the technical aspects of estimating solutions greatly. In the model case $\mathcal f(z)=z^n$ at zero energy, the angles satisfy an autonomous system for which we construct an exactly decaying Lyapunov function (\cref{c-model}); we also record a uniform bound on exit times in the model case (\cref{c-model}(4)). These results allow us to rederive the proof of the Ehlers--Kundt conjecture in the case where $\mathcal f$ is a polynomial. This is a short argument in \cref{s-polynomial-conjecture}, which combines the model case with the scaling symmetry of the equation and the continuity of the flow. The polynomial case of the Ehlers--Kundt conjecture was first proved by \cite{Flores_2020}. Our proof requires no new substantial insight; for us, the polynomial case serves as an example of how much the change of variables in \cref{ss-coordinate-transformation} can simplify the proof.
 
In \cref{s-logic}, we discuss a notion of angular confinement for trajectories of $\ddot{\bm z}=\overline{\mathcal f'(t;\bm z(t))}$ (\cref{d-aperture,d-S-region,defn-S-statement}), which quantifies the extent to which trajectories are confined to a neighborhood of the radial directions $\arg(z^n)=0$, near which the derivative of $1-e^{-z^n}$ decays exponentially. The central statement is the focusing theorem (\cref{t-focusing}): trajectories with small initial data are confined to an arbitrarily small aperture once $|\bm z|$ is large (see \cref{fig-focusing}). We deduce \cref{main-lemma} from \cref{t-focusing} in \cref{ss-main-lemma-proof}. In \cref{ss-focusing-theorem-logic}, we state five focusing lemmas, each of which addresses a different range of scales (or, in the case of \cref{l-stable-focusing}, the passage from zero-energy solutions for $1-e^{-z^n}$ to general solutions for its time-dependent perturbations), and we show how they combine to give \cref{t-focusing}; the main point is the order in which the various constants are chosen. The five focusing lemmas are then proved in \cref{s-focusing-lemmas}, one per subsection.
 
\Cref{s-polynomial-conjecture} can be read independently of \cref{s-localization}, \cref{s-logic} and \cref{s-focusing-lemmas}: there we use the model case of \cref{ss-coordinate-transformation} (\cref{c-model}) to give a new proof of the polynomial case of \cref{ek-conjecture} (\cref{t-polynomial}), which is due to \cite{Flores_2020}. The causality of gravitational pp-waves (in particular, the fact that the counterexamples of \cref{t-counterexample} are not distinguishing) is the subject of the companion paper \cite{Cairo_causality}, which is independent of the present paper except for \cref{t-reduction}(d).

\subsection{Statement on AI usage}
All ideas in this paper are of human origin. The author used Claude Fable 5.1 and Claude Sonnet 5 to generate figures, check proofs, search the literature, edit, and write parts of this paper.

\subsection{Notation}\label{ss-notation}
Throughout the proof, we will make use of many different constants that depend on each other in different ways.
    \begin{definition}~\\
        \begin{itemize}
        \item[$(\lesssim)$] Suppose that $A,B>0$ are two positive quantities. If there exists some constant $C$ so that $A\leq CB$, then we write $A\lesssim B$ or $B\gtrsim A$. If $A\lesssim B$ and $B\lesssim A$, then we write $A\sim B$. If the constant in $A\leq CB$ depends on other parameters, e.g. $D$, then we will display this dependence as $A\lesssim_DB$.
        \item[$(\ll)$] Let $P(x,y)$ be a proposition about two positive quantities $x,y$. Suppose that there exists some constant $c>0$ such that $x\leq cy\implies P(x,y)$. Then we write this statement as follows:
            \begin{enumerate}
                \item[$(\bullet)$] If $x\ll y$, then $P(x,y)$.
                \end{enumerate}
            We define $\gg$ similarly (i.e. $x\gg y$ means $y\ll x$). When we are dealing with an arbitrarily small positive number $x$, we often write $x\ll1$. We also write $x\ll0$ to mean that $x\leq-C$ for some sufficiently large constant $C$; that is, $x$ is a large negative number, and similarly $x\gg0$ to mean that $x\geq C$ for some sufficiently large constant $C$. (These are not covered by the definition above, which requires $y>0$.)
        \item[$(\lll)$] This notation is non-standard. Let $\mathfrak W$ denote the set of all functions $g:(0,\infty)\to(0,\infty)$, partially ordered pointwise by $\leq$. When we write ``if $x\lll y$, then $P(x,y)$'', we mean that there exists some $\mathcal C\in\mathfrak W$ such that $x\leq\mathcal C(y)\implies P(x,y)$; the statement then also holds for every $\mathcal C'\leq\mathcal C$ in $\mathfrak W$, so $\mathcal C$ should be thought of as ``sufficiently small'' in the partial order. Similarly, we write $x\ggg y$ to mean that $x\inv\lll y\inv$. Throughout this paper, $\mathcal C$ can be taken to be a fixed function, depending only on a possible parameter $n$. Note that $x\lll y$ is not equivalent to $y\lll x$. We can loosely interpret $x\lll y$ to mean ``$x$ is sufficiently small with respect to $y$''.

            We will use two extensions of this notation. First, ``$x\lll y_1,\dots,y_k$'' means that $x\leq\mathcal C(y_1,\dots,y_k)$ for some function $\mathcal C:(0,\infty)^k\to(0,\infty)$, and ``$x_1,\dots,x_j\lll y_1,\dots,y_k$'' means that this holds for each $x_i$. Second, a chain ``$x_1\lll x_2\lll\cdots\lll x_k$'' of consecutive $\lll$ signs means that $x_i\lll x_{i+1},\dots,x_k$ for every $i<k$. That is, the constants in a chain are chosen from right to left, and each of them is required to be sufficiently small with respect to all of the constants which have already been chosen.
        \item[$(D_R)$] For $z_0\in\C$ and $R>0$ we write $D_R(z_0):=\{z\in\C:|z-z_0|<R\}$ for the open disc of radius $R$ about $z_0$, and $D_R:=D_R(0)$.
        \end{itemize}
    \end{definition}

\subsection{Acknowledgements}

The author would like to thank Amir Aazami for introducing her to the Ehlers--Kundt conjecture.

\section{The geodesic equations and the reduction to a complex ODE}\label{s-geodesics}
In this section we recall how the completeness of a pp-wave reduces to a statement about a complex ODE. None of this is new (see \cite{Candela_2003,Flores_2020}), but we need to fix our normalizations, and \cref{t-reduction}(d) is used in the companion paper \cite{Cairo_causality}. The reader familiar with geodesics on pp-waves can comfortably skip this section.

Throughout this section, $H:\R\times\R^2\to\R$, $(u,x,y)\mapsto H(u,x,y)$, is a smooth function, and $M=(\R^4,g)$ is the pp-wave with metric
    \[\label{e-brinkmann}
        g=2dudv-H(u,x,y)du^2+dx^2+dy^2
        \]
in the coordinates $(u,v,x,y)$. We write $(x^1,x^2)=(x,y)$, and the indices $i,j$ always range over $\{1,2\}$. We write $H_u=\partial_uH$, $H_i=\partial_{x^i}H$ and $H_{ij}=\partial_{x^i}\partial_{x^j}H$, and we let $\nabla H=(H_1,H_2)$ and $\Delta H=H_{11}+H_{22}$ denote the gradient and the Laplacian of $H$ in the variables $(x,y)$ only. We do not assume that $H$ is harmonic until \cref{ss-reduction}.

With respect to the coordinate frame $(\partial_u,\partial_v,\partial_x,\partial_y)$, the only non-constant component of $g$ is $g_{uu}=-H$, and $g_{uv}=g_{xx}=g_{yy}=1$ are the only other nonzero components. So $g$ is a Lorentzian metric for every choice of $H$, and $\partial_v$ is a null vector field. The Christoffel symbols of $g$ which do not vanish identically are
    \[\label{e-christoffel}
        \Gamma^v_{uu}=-\tfrac12H_u,\qquad\Gamma^v_{ui}=\Gamma^v_{iu}=-\tfrac12H_i,\qquad\Gamma^i_{uu}=\tfrac12H_i\qquad(i=1,2)
        \]
(see \cite[Section 2]{Candela_2003}; recall that their $H$ is our $-H$). In particular, $\Gamma^\lambda_{\mu v}=0$ for all $\lambda,\mu$, i.e. $\nabla\partial_v=0$: the null vector field $\partial_v$ is parallel. (This is the origin of the phrase ``parallel rays''.) The same is then true of the one-form $du=g(\partial_v,\cdot)$, and by \cref{e-christoffel}, for any vector $X$ we have $\nabla_Xdx^i=-\Gamma^i_{\mu\nu}X^\mu dx^\nu=-\tfrac12H_i\,du(X)\,du$, so that
    \[\label{e-nabla-dx}
        \nabla_X\p{du\wedge dx^i}=du\wedge\nabla_Xdx^i=0 .
        \]
That is, the two-forms $du\wedge dx$ and $du\wedge dy$ are parallel as well. This makes the curvature of $g$ particularly easy to describe. We use the conventions $R(X,Y)Z=\nabla_X\nabla_YZ-\nabla_Y\nabla_XZ-\nabla_{[X,Y]}Z$ and $(\alpha\wedge\beta)(X,Y)=\alpha(X)\beta(Y)-\alpha(Y)\beta(X)$; none of the statements that we use below depends on these sign conventions.

\begin{lemma}[Curvature]\label{l-curvature}
    The Riemann curvature tensor of $g$ (with all indices lowered) and its covariant derivative in the direction of a vector $X$ are
        \[\label{e-riemann}
            R=\tfrac12\sum_{i,j}H_{ij}\,(du\wedge dx^i)\otimes(du\wedge dx^j),\qquad\nabla_XR=\tfrac12\sum_{i,j}X(H_{ij})\,(du\wedge dx^i)\otimes(du\wedge dx^j),
            \]
    and the Ricci tensor of $g$ is $\mathrm{Ric}=\tfrac12\Delta H\,du^2$. In particular, $g$ is Ricci-flat if and only if $H(u,\cdot)$ is harmonic for every $u$.
    \end{lemma}

\begin{proof}
    In coordinates, $R^\rho{}_{\sigma\mu\nu}=\partial_\mu\Gamma^\rho_{\nu\sigma}-\partial_\nu\Gamma^\rho_{\mu\sigma}+\Gamma^\rho_{\mu\lambda}\Gamma^\lambda_{\nu\sigma}-\Gamma^\rho_{\nu\lambda}\Gamma^\lambda_{\mu\sigma}$. By \cref{e-christoffel}, a product $\Gamma^\rho_{\mu\lambda}\Gamma^\lambda_{\nu\sigma}$ can only be nonzero if $\lambda=i$, $\rho=v$ and $\mu=\nu=\sigma=u$, and these terms cancel in $R^v{}_{uuu}$. So the curvature is given by the derivative terms alone, and \cref{e-christoffel} gives $R_{iuju}=R^i{}_{uju}=\partial_j\Gamma^i_{uu}=\tfrac12H_{ij}$, while all components of $R$ which are not obtained from $R_{iuju}$ by the symmetries of the curvature tensor vanish (for instance, $R^v{}_{uju}=\partial_j\Gamma^v_{uu}-\partial_u\Gamma^v_{ju}=0$). Since $(du\wedge dx^k)(\partial_i,\partial_u)=-\delta_{ki}$, this is the first formula in \cref{e-riemann}; the second formula follows from the first one and \cref{e-nabla-dx}. Finally, the only nonzero component of the Ricci tensor is $R_{uu}=\sum_iR^i{}_{uiu}=\tfrac12\Delta H$ (cf. \cite[Proposition 2.1]{Candela_2003}).
\end{proof}

The geodesic equations of $g$ follow from \cref{e-christoffel}: for a curve $\bm\sigma=(\bm u,\bm v,\bm\gamma)$ with $\bm\gamma=(\bm x,\bm y)$ and parameter $s$, they read $\ddot{\bm u}=0$, $\ddot{\bm\gamma}=-\tfrac12\dot{\bm u}^2\nabla H$ and $\ddot{\bm v}=\tfrac12\dot{\bm u}^2H_u+\dot{\bm u}\,\nabla H\cdot\dot{\bm\gamma}$, where $H$ and its derivatives are evaluated at $(\bm u(s),\bm\gamma(s))$. Note that $\bm u$ is affine, the equation for $\bm\gamma$ does not involve $\bm v$, and $\bm v$ is obtained from $\bm u$ and $\bm\gamma$ by integrating the conservation law for the energy $E=g(\dot{\bm\sigma},\dot{\bm\sigma})$. So the only equation which can be responsible for incompleteness is the one for $\bm\gamma$, which is Newton's equation in the plane for the time-dependent potential $\tfrac12H$. This gives the following classical criterion.

\begin{proposition}[Completeness criterion; {\cite[Proposition 3.1 and Theorem 3.2]{Candela_2003}}]\label{p-completeness}
    The geodesics of $(M,g)$ are the following curves.
        \begin{enumerate}
            \item[(1)] The affinely parameterized straight lines $s\mapsto p+sX$ with $p\in\R^4$ and $du(X)=0$. These are defined for all $s\in\R$.
            \item[(2)] Up to an affine change of the parameter, the curves
                \[\label{e-geodesic-form}
                    t\mapsto\p{t,\ \bm v(t),\ \bm z(t)},\qquad\bm v(t)=v_0+\tfrac12\int_{t_0}^t\p{E+H(\tau,\bm z(\tau))-|\dot{\bm z}(\tau)|^2}d\tau,\qquad t\in I,
                    \]
                where $t_0,v_0,E\in\R$ and $\bm z:I\to\R^2$ is a solution of
                \[\label{e-newton}
                    \ddot{\bm z}(t)=-\tfrac12\nabla H(t,\bm z(t)),
                    \]
                defined on an open interval $I\ni t_0$. The geodesic \cref{e-geodesic-form} is inextendible if and only if $\bm z:I\to\R^2$ is a maximal solution of \cref{e-newton}.
            \end{enumerate}
    Consequently, $(M,g)$ is geodesically complete if and only if, for every $t_0\in\R$ and all $a,b\in\R^2$, the maximal solution of \cref{e-newton} with $\bm z(t_0)=a$ and $\dot{\bm z}(t_0)=b$ is defined on all of $\R$.
    \end{proposition}

The curves in (2) are the geodesics with $\dot{\bm u}\neq0$, reparameterized so that $\bm u(t)=t$, and the constant $E$ in \cref{e-geodesic-form} is $g(\dot{\bm\sigma},\dot{\bm\sigma})$. Since the integrand in \cref{e-geodesic-form} is continuous on $I$, the function $\bm v$ is defined on all of $I$, which gives the last two statements.
\begin{remark}\label{r-causal-character}
    \begin{enumerate}
        \item[(i)] The constant $E$ in \cref{e-geodesic-form} determines the causal character of the geodesic (timelike for $E<0$, null for $E=0$, spacelike for $E>0$), and it does not enter \cref{e-newton}. Hence every maximal solution of \cref{e-newton} which is not defined on all of $\R$ gives rise to incomplete timelike, null \textit{and} spacelike geodesics: for pp-waves, the three notions of geodesic completeness coincide.
        \item[(ii)] Smoothness of $H$ is assumed for convenience only. \Cref{p-completeness} only requires $H$ and $\nabla H$ to be of class $C^1$: then the right-hand side of \cref{e-newton} is $C^1$, so that its solutions are unique, and $\bm v$ is determined by \cref{e-geodesic-form}, even though $\Gamma^v_{uu}=-\tfrac12H_u$ is merely continuous. This is the regularity of \cref{t-polynomial-intro}; the ODE statement \cref{t-polynomial} behind it only requires continuity in $t$.
    \end{enumerate}
\end{remark}

\subsection{The reduction to \texorpdfstring{\cref{main-theorem}}{Theorem \ref{main-theorem}}}\label{ss-reduction}
We now identify $\R^2$ with $\C$ via $z=x+iy$, and we assume that $H(u,\cdot)$ is harmonic for every $u$, i.e. that $(M,g)$ is a gravitational pp-wave (\cref{l-curvature}). Then $\tfrac12(H_x-iH_y)$ satisfies the Cauchy--Riemann equations in $z$, which read $H_{xx}=-H_{yy}$ and $H_{xy}=H_{yx}$. So
    \[\label{e-f-from-H}
        \mathcal f(u;z):=-\tfrac12H(u,0,0)-\tfrac12\int_0^z\p{H_x-iH_y}(u,\xi)\,d\xi
        \]
(the integral being taken along any path from $0$ to $z$) is smooth in $(u,z)$ and entire in $z$, with $\mathcal f'=\partial_z\mathcal f=-\tfrac12(H_x-iH_y)$. Since $\partial_x\Re\p{\mathcal f}=\Re\p{\mathcal f'}=-\tfrac12H_x$ and $\partial_y\Re\p{\mathcal f}=-\Im\p{\mathcal f'}=-\tfrac12H_y$, we conclude that
    \[\label{e-H-f}
        H(u,x,y)=-2\Re\p{\mathcal f(u;x+iy)}\qquad\text{and}\qquad-\tfrac12\p{H_x+iH_y}=\overline{\mathcal f'},
        \]
and $\mathcal f$ is determined by $H$ up to the addition of a purely imaginary function of $u$, which does not affect $\mathcal f'$. Conversely, for any $\mathcal f:\R\times\C\to\C$ which is smooth and entire in $z$, the function $H=-2\Re\p{\mathcal f}$ is smooth and harmonic in $(x,y)$, and \cref{e-H-f} holds. Differentiating \cref{e-H-f} once more gives
    \[\label{e-hessian-f}
        H_{xx}=-H_{yy}=-2\Re\p{\mathcal f''},\qquad H_{xy}=2\Im\p{\mathcal f''}.
        \]
By a \textit{plane wave} we mean, as in \cref{s-introduction}, a pp-wave $(\R^4,g')$ with $g'=2dudv-H'(u,x,y)du^2+dx^2+dy^2$ for which $H'(u,\cdot)$ is a polynomial of degree $\leq2$ in $(x,y)$ for every $u$.

\begin{theorem}[Reduction to the complex ODE]\label{t-reduction}
    Let $\mathcal f:\R\times\C\to\C$, $(u,z)\mapsto\mathcal f(u;z)$, be smooth, and entire in $z$. Let $M_{\mathcal f}=(\R^4,g_{\mathcal f})$ be the pp-wave with
        \[\label{e-g-f}
            g_{\mathcal f}=2dudv+2\Re\p{\mathcal f(u;x+iy)}du^2+dx^2+dy^2,
            \]
    i.e. the pp-wave \cref{e-brinkmann} with $H=-2\Re\p{\mathcal f}$. Then the following hold.
        \begin{enumerate}
            \item[(a)] $M_{\mathcal f}$ is a gravitational pp-wave, and every gravitational pp-wave is of the form $M_{\mathcal f}$ for some such $\mathcal f$.
            \item[(b)] A curve is a geodesic of $M_{\mathcal f}$ which is not one of the straight lines of \cref{p-completeness}(1) if and only if, up to an affine change of the parameter, it is of the form $t\mapsto(t,\bm v(t),\bm z(t))$, where $\bm z=\bm x+i\bm y$ solves
                \[\label{e-complex-ode}
                    \ddot{\bm z}(t)=\overline{\mathcal f'(t;\bm z(t))}
                    \]
                and $\bm v$ is given by \cref{e-geodesic-form}. Equivalently, $(\bm z,\bm w)$ with $\bm w:=\dot{\bm z}/\sqrt2$ solves the system \cref{e-ode}.
            \item[(c)] $M_{\mathcal f}$ is geodesically complete if and only if, for every $t_0\in\R$ and all $a,b\in\C$, the maximal solution of \cref{e-complex-ode} with $\bm z(t_0)=a$ and $\dot{\bm z}(t_0)=b$ is defined for all $t\in\R$. In particular, if $\mathcal f(u;z)=\mathcal f(z)$ does not depend on $u$, then $M_{\mathcal f}$ is geodesically complete if and only if all solutions of $\ddot{\bm z}(t)=\overline{\mathcal f'(\bm z(t))}$ exist for all time $t\in\R$.
            \item[(d)] Suppose that $\mathcal f(u_0;\cdot)$ is not a polynomial of degree $\leq2$ for some $u_0\in\R$. Then $M_{\mathcal f}$ is not isometric to a plane wave. More precisely, there is some $z_0=x_0+iy_0$ with $\mathcal f'''(u_0;z_0)\neq0$, and for any such $z_0$ and any $v_0\in\R$, no neighborhood of the point $(u_0,v_0,x_0,y_0)$ in $M_{\mathcal f}$ is isometric to an open subset of a plane wave. Conversely, if $\mathcal f(u;\cdot)$ is a polynomial of degree $\leq2$ for every $u$, then $M_{\mathcal f}$ is a plane wave.
            \end{enumerate}
    \end{theorem}

\begin{proof}
    (a) follows from \cref{l-curvature} and the discussion preceding the theorem. For (b) and (c), note that by \cref{e-H-f} the equation \cref{e-newton} for $\bm z=(\bm x,\bm y)$ is precisely the equation \cref{e-complex-ode}, once we identify $\R^2$ with $\C$ by writing $\bm z=\bm x+i\bm y$:
        \[\ddot{\bm z}=\ddot{\bm x}+i\ddot{\bm y}=-\tfrac12\p{H_x+iH_y}(t,\bm z)=\overline{\mathcal f'(t;\bm z)}.\]
    So (b) and the first statement of (c) are a restatement of \cref{p-completeness}. If $\mathcal f$ does not depend on $u$, then \cref{e-complex-ode} is autonomous, so the initial time $t_0$ plays no role, and we obtain the second statement of (c).

    (d) For a pp-wave $(M,g)$ as in \cref{e-brinkmann} and a point $p\in M$, consider the subspace
        \[\mathcal N_p(M,g):=\{X\in T_pM:\nabla_XR=0\text{ at }p\}\subset T_pM .\]
    If $\Phi:U\to U'$ is an isometry between open subsets of two pseudo-Riemannian manifolds, then $\Phi$ intertwines the curvature tensors and their covariant derivatives, so $d\Phi_p$ maps $\mathcal N_p$ isomorphically onto $\mathcal N_{\Phi(p)}$ for every $p\in U$. In particular, the function $p\mapsto\dim\mathcal N_p$ is preserved by (local) isometries. The two-forms $du\wedge dx$ and $du\wedge dy$ are linearly independent at every point, and hence so are the four tensors $(du\wedge dx^i)\otimes(du\wedge dx^j)$; so by \cref{l-curvature}
        \[\label{e-N-p}
            \mathcal N_p(M,g)=\{X\in T_pM:X(H_{ij})=0\text{ at $p$, for all }i,j\} .
            \]
    If $(M',g')$ is a plane wave, then the functions $H'_{ij}$ depend only on $u$, so $\mathcal N_{p'}(M',g')$ contains $\partial_v,\partial_x,\partial_y$, and $\dim\mathcal N_{p'}(M',g')\geq3$ for every $p'\in M'$. On the other hand, for $M_{\mathcal f}$, \cref{e-hessian-f} and \cref{e-N-p} show that $X=X^u\partial_u+X^v\partial_v+X^x\partial_x+X^y\partial_y$ belongs to $\mathcal N_p(M_{\mathcal f})$ if and only if $X(\mathcal f'')=0$ at $p$, i.e. if and only if
        \[\label{e-N-p-f}
            X^u\,\partial_u\mathcal f''+\p{X^x+iX^y}\mathcal f'''=0\qquad\text{at }p,
            \]
    since $\mathcal f''$ is holomorphic in $z$ and does not depend on $v$. Now, if $\mathcal f(u_0;\cdot)$ is not a polynomial of degree $\leq2$, then the entire function $\mathcal f'''(u_0;\cdot)$ is not identically zero, so there is some $z_0$ with $\mathcal f'''(u_0;z_0)\neq0$. At $p=(u_0,v_0,x_0,y_0)$, the left-hand side of \cref{e-N-p-f} is then a surjective real-linear map $\R^4\to\C$ of $(X^u,X^v,X^x,X^y)$ (it is already surjective on the span of $\partial_x,\partial_y$), so $\dim\mathcal N_p(M_{\mathcal f})=2<3$. Hence no neighborhood of $p$ is isometric to an open subset of a plane wave. The converse statement is immediate, since $H(u,\cdot)=-2\Re\p{\mathcal f(u;\cdot)}$ is a polynomial of degree $\leq2$ in $(x,y)$ whenever $\mathcal f(u;\cdot)$ is a polynomial of degree $\leq2$ in $z$.
\end{proof}

\begin{corollary}\label{c-ek-false}
    Let $\mathcal f:\C\to\C$ be a function as in \cref{main-theorem}, i.e. an entire function which is not a polynomial of degree $\leq2$, such that all solutions of $\ddot{\bm z}(t)=\overline{\mathcal f'(\bm z(t))}$ exist for all time $t\in\R$. Then
        \[g_{\mathcal f}=2dudv+2\Re\p{\mathcal f(x+iy)}du^2+dx^2+dy^2\]
    is a geodesically complete, Ricci-flat pp-wave metric on $\R^4$ which is not isometric (not even locally, near any point $(u,v,x,y)$ with $\mathcal f'''(x+iy)\neq0$) to a plane wave. In particular, \cref{main-theorem} implies that \cref{ek-conjecture} is false.
    \end{corollary}

\begin{proof}
    This is \cref{t-reduction}(a), (c) and (d), applied to $\mathcal f(u;z):=\mathcal f(z)$.
\end{proof}

\begin{remark}[Energy]\label{r-energy}
    Let $t\mapsto(t,\bm v(t),\bm z(t))$ be a geodesic of $M_{\mathcal f}$ as in \cref{t-reduction}(b), and let $\bm w=\dot{\bm z}/\sqrt2$. By \cref{e-geodesic-form} with $H=-2\Re\p{\mathcal f}$, we have $\dot{\bm v}=\tfrac12E-\Re\p{\mathcal f(t;\bm z)}-|\bm w|^2$, and $g(\dot{\bm\sigma},\partial_u)=\dot{\bm v}-H=\dot{\bm v}+2\Re\p{\mathcal f(t;\bm z)}$. Hence the quantity $\bm\kappa=|\bm w|^2-\Re\p{\mathcal f(t;\bm z)}$ of \cref{ss-coordinate-transformation} is
        \[\bm\kappa=\tfrac12E-\p{\dot{\bm v}+2\Re\p{\mathcal f(t;\bm z)}}=\tfrac12g(\dot{\bm\sigma},\dot{\bm\sigma})-g(\dot{\bm\sigma},\partial_u).\]
    So the energy $\bm\kappa$, which plays a central role in \cref{s-logic,s-focusing-lemmas}, is (up to the constant $\tfrac12E$) the momentum of the geodesic conjugate to $u$. When $\mathcal f$ does not depend on $u$, the vector field $\partial_u$ is Killing and $\bm\kappa$ is conserved; this is \cref{e-conservation-of-energy-focusing-section}. In general, $\frac d{dt}g(\dot{\bm\sigma},\partial_u)=\tfrac12(\mathcal L_{\partial_u}g)(\dot{\bm\sigma},\dot{\bm\sigma})=\Re\p{\mathcal f_t(t;\bm z)}$, in accordance with the identity $\dot{\bm\kappa}=-\Re\p{\bm f_t}$ of \cref{l-cov}.
\end{remark}

\section{Continuous dependence on the potential, and first crossing times}\label{ss-appendix}\label{ss-appendix-continuity}
Several proofs in this paper use that the solutions of our equations depend continuously on the initial data \emph{and on the function $\mathcal f$}, and that the same is true of the first time at which a solution reaches a given hypersurface (in practice, a given radius), provided that it crosses this hypersurface transversally. We record these facts here, together with the function spaces in which the Baire category argument of \cref{s-localization} takes place. They are statements about the continuity of certain maps, and the only point which requires some care is the choice of the topologies.

\emph{The spaces of functions.} Let $\Hsp$ denote the set of all continuous functions $\mathcal f:\R\times\C\to\C$, $(t,z)\mapsto\mathcal f(t;z)$, which are entire in $z$, with the topology of uniform convergence on compact subsets of $\R\times\C$. Thus the sets
    \[\label{e-Hsp-neighborhood}
        \mathcal N(\mathcal g,L,\eta):=\left\{\tilde{\mathcal f}\in\Hsp:|\tilde{\mathcal f}-\mathcal g|<\eta\text{ on }[-L,L]\times\overline{D_L}\right\},\qquad L\geq1,\ \eta>0,\]
form a neighborhood basis of $\mathcal g\in\Hsp$. Note that a neighborhood of $\mathcal g$ only constrains the functions on a compact subset of $\R\times\C$. By the Cauchy integral formula, $\mathcal f'$ and $\mathcal f''$ are continuous on $\R\times\C$ for every $\mathcal f\in\Hsp$, and if $\mathcal f_k\to\mathcal f$ in $\Hsp$, then $\mathcal f_k'\to\mathcal f'$ and $\mathcal f_k''\to\mathcal f''$ uniformly on compact sets. Since a locally uniform limit of functions which are entire in $z$ is again entire in $z$, the space $\Hsp$ is a closed subspace of the Fr\'echet space of all continuous functions on $\R\times\C$; in particular, it is completely metrizable, for instance by the metric $d(\mathcal f,\mathcal g)=\sum_{L\geq1}2^{-L}\min\p{1,\sup_{[-L,L]\times\overline{D_L}}|\mathcal f-\mathcal g|}$. We will also use the following two spaces.
    \begin{itemize}
        \item $\Haut\subset\Hsp$ denotes the closed subspace of functions which do not depend on $t$. We identify it with the space of entire functions $\mathcal f:\C\to\C$, with the topology of uniform convergence on compact sets; in this way, an entire function is regarded as an element of $\Hsp$.
        \item $\Hinf\subset\Hsp$ denotes the set of functions $\mathcal f\in\Hsp$ for which all partial derivatives $\partial_t^m\mathcal f$, $m\geq0$, exist and are continuous on $\R\times\C$, with the finer topology defined by the seminorms
            \[\label{e-Hinf-seminorms}
                \|\mathcal f\|_{k,S}:=\max_{0\leq m\leq k}\ \sup_{[-S,S]\times\overline{D_S}}\abs{\partial_t^m\mathcal f},\qquad k\geq0,\ S>0.\]
            Thus $\mathcal f_j\to\mathcal f$ in $\Hinf$ if and only if $\partial_t^m\mathcal f_j\to\partial_t^m\mathcal f$ uniformly on compact sets for every $m\geq0$. Note that $\partial_t^m\mathcal f(t;\cdot)$ is again entire for $\mathcal f\in\Hinf$ (it is a locally uniform limit of difference quotients of $\partial_t^{m-1}\mathcal f$), so that, by the Cauchy estimates, $\|\mathcal f\|_{k,S+1}$ controls all mixed derivatives $\partial_t^m\partial_z^l\mathcal f$, $m\leq k$, on $[-S,S]\times\overline{D_S}$; in particular, $\mathcal f$ is smooth on $\R\times\C$.
    \end{itemize}
Both spaces are completely metrizable: $\Haut$ is a closed subspace of $\Hsp$, and $\Hinf$ is a Fr\'echet space (a sequence which is Cauchy for all the seminorms \cref{e-Hinf-seminorms} converges, together with all its $t$-derivatives, uniformly on compact sets, and the limit is smooth in $t$ and entire in $z$). The inclusions $\Haut\hookrightarrow\Hsp$ and $\Hinf\hookrightarrow\Hsp$ are continuous. By \cref{t-reduction}(a), the gravitational pp-waves are precisely the $M_{\mathcal f}$ with $\mathcal f\in\Hinf$, and those with $H$ independent of $u$ are the $M_{\mathcal f}$ with $\mathcal f\in\Haut$.

\emph{The space of curves.} Let $\Csp$ denote the set of all \emph{inextendable} $C^1$ curves $\gamma:(\alpha,\beta)\to\C^2$, where $-\infty\leq\alpha<0<\beta\leq\infty$ may depend on $\gamma$; here inextendable means that $\gamma$ is not the restriction of a $C^1$ curve which is defined on a strictly larger open interval. We write $\dom\gamma=(\alpha,\beta)$. For $\gamma\in\Csp$, a compact interval $J\subset\dom\gamma$ and $\eps>0$, let
    \[\mathcal U(\gamma,J,\eps):=\left\{\tilde\gamma\in\Csp:J\subset\dom\tilde\gamma\text{ and }\|\tilde\gamma-\gamma\|_{C^1(J)}<\eps\right\}.\]
We equip $\Csp$ with the topology in which a set is open if and only if it contains such a set $\mathcal U(\gamma,J,\eps)$ about each of its points $\gamma$. The sets $\mathcal U(\gamma,J,\eps)$ are themselves open, and $\Csp$ is first countable. It is also Hausdorff: two curves in $\Csp$ which agree on the intersection of their domains are equal, since otherwise one of them would extend the other. A sequence $\gamma_k$ converges to $\gamma$ if and only if, for every compact interval $J\subset\dom\gamma$, we have $J\subset\dom\gamma_k$ for all large $k$ and $\gamma_k\to\gamma$ in $C^1(J)$. In particular, for every $T>0$, the set
    \[\label{e-Csp-T}
        \Csp_T:=\{\gamma\in\Csp:[-T,T]\subset\dom\gamma\}\qquad\text{is open in }\Csp,\]
and the evaluation $\gamma\mapsto\gamma(t)$ is continuous on the open set $\{\gamma:t\in\dom\gamma\}$.

\emph{The solution map and the flow.} For $(\mathcal f,z,w)\in\Hsp\times\C^2$, let
    \[\mathcal F(\mathcal f,z,w):=(\bm z,\bm w)\in\Csp\]
denote the maximal solution to
    \[\label{e-appendix-first-order}
        \dot{\bm z}=\sqrt2\,\bm w,\qquad\dot{\bm w}=\frac1{\sqrt2}\,\overline{\mathcal f'(t;\bm z)}\]
with $(\bm z(0),\bm w(0))=(z,w)$. Here maximal means that the domain $(\alpha,\beta)$ of the solution is as large as possible (it may be all of $\R$). The maximal solution is inextendable as a $C^1$ curve, so that $\mathcal F(\mathcal f,z,w)\in\Csp$: if it were the restriction of a $C^1$ curve defined beyond, say, $\beta<\infty$, then $(\bm z,\bm w)$ would have a limit as $t\to\beta$, and the local existence theorem, applied with this limit as data at the time $\beta$, would continue the solution beyond $\beta$, contradicting maximality. This is the system \cref{e-ode} (and \cref{z-w-ode,z-w-tilde-f-ode} below), and $\bm z$ is the solution to $\ddot{\bm z}=\overline{\mathcal f'(t;\bm z)}$ with $\bm z(0)=z$ and $\dot{\bm z}(0)=\sqrt2w$. The solution with data $(z,w)$ at an arbitrary initial time $t_0$ is obtained by translating $\mathcal f$ in time: it is $t\mapsto\mathcal F(\mathcal f(t_0+\cdot\,;\cdot),z,w)(t-t_0)$. We collect all of these solutions into the \emph{flow} of $\mathcal f$,
    \[\label{e-flow}
        \mathcal G(\mathcal f):\R\times\C^2\to\Csp,\qquad\mathcal G(\mathcal f)(t_0,z,w):=\mathcal F\p{\mathcal f(t_0+\cdot\,;\cdot),z,w},\]
and we let $\Fsp:=C(\R\times\C^2,\Csp)$ denote the space of continuous maps $\R\times\C^2\to\Csp$ with the compact-open topology, i.e. the topology generated by the sets
    \[\label{e-Fsp-subbasis}
        \{u\in\Fsp:u(K)\subset V\},\qquad K\subset\R\times\C^2\text{ compact},\ V\subset\Csp\text{ open}.\]
(\Cref{l-continuity}(b) says in particular that $\mathcal G(\mathcal f)$ is continuous, i.e. that $\mathcal G(\mathcal f)\in\Fsp$.) For $\mathcal f\in\Haut$, the flow $\mathcal G(\mathcal f)$ does not depend on $t_0$.

\emph{The first crossing.} Let $M\subset\C^2$ be a \emph{closed} subset which is a smooth embedded real hypersurface; in all of our applications, $M=\{(z,w)\in\C^2:|z|=\rho\}$ for some $\rho>0$. Let $\mathcal P(M)\subset\Csp$ denote the set of all curves $\gamma$ with the following properties: $\gamma(0)\notin M$; the curve $\gamma$ meets $M$ at some positive time; and at the \emph{first} positive time $t_1$ at which $\gamma(t_1)\in M$ (which exists because $M$ is closed), the curve $\gamma$ is transversal to $M$, i.e. $\dot\gamma(t_1)\notin T_{\gamma(t_1)}M$. For $\gamma\in\mathcal P(M)$ we write
    \[t_M(\gamma):=t_1\in(0,\infty)\qquad\text{and}\qquad p_M(\gamma):=\gamma(t_1)\in M\]
for the first crossing time and the first crossing point. If $M=\{|z|=\rho\}$, then $\gamma=(\bm z,\bm w)$ is transversal to $M$ at the time $t_1$ if and only if $\frac d{dt}|\bm z|^2(t_1)\neq0$; for solutions of \cref{e-appendix-first-order}, this means that $\Re\p{\overline{\bm z(t_1)}\bm w(t_1)}\neq0$. By a slight abuse of notation, we write
    \[(\mathcal f,z,w)\in\mathcal P(M),\qquad t_M(\mathcal f,z,w),\qquad p_M(\mathcal f,z,w)\]
when we mean $\mathcal F(\mathcal f,z,w)\in\mathcal P(M)$, $t_M(\mathcal F(\mathcal f,z,w))$ and $p_M(\mathcal F(\mathcal f,z,w))$.

\begin{lemma}[Continuity of the flow and of the first crossing]\label{l-continuity}~\\
    \begin{enumerate}
        \item[(a)] The solution map $\mathcal F:\Hsp\times\C^2\to\Csp$ is continuous.
        \item[(b)] The flow $\mathcal G:\Hsp\to\Fsp$ is well defined and continuous. Consequently, for every compact set $K\subset\R\times\C^2$ and every $T>0$, the set of all $\mathcal f\in\Hsp$ such that, for every $(t_0,z,w)\in K$, the solution of \cref{e-appendix-first-order} with data $(z,w)$ at the time $t_0$ exists on $[t_0-T,t_0+T]$, is open in $\Hsp$.
        \item[(c)] For every closed smooth embedded real hypersurface $M\subset\C^2$, the set $\mathcal P(M)$ is open in $\Csp$, and the maps $t_M:\mathcal P(M)\to(0,\infty)$ and $p_M:\mathcal P(M)\to M$ are continuous.
    \end{enumerate}
\end{lemma}

Without transversality, (c) is false: a solution which touches the circle $|z|=\rho$ tangentially from the inside can be perturbed so that it does not reach this circle until much later, or not at all. We will use (c) mostly through the following consequence: if $V\subset M$ is relatively open and $T>0$, then
    \[\label{e-crossing-open}
        \mathcal W(M,V,T):=\left\{(\mathcal f,z,w)\in\Hsp\times\C^2:(\mathcal f,z,w)\in\mathcal P(M),\ t_M(\mathcal f,z,w)<T,\ p_M(\mathcal f,z,w)\in V\right\}\qquad\text{is open,}\]
being the preimage under $\mathcal F$ of the open set $\{\gamma\in\mathcal P(M):t_M(\gamma)<T,\ p_M(\gamma)\in V\}$. In words, $(\mathcal f,z,w)\in\mathcal W(M,V,T)$ means that the solution for $\mathcal f$ with data $(z,w)\notin M$ at time $0$ reaches $M$ for the first time before the time $T$, transversally and at a point of $V$.

\begin{proof}
    \emph{(a) Continuity of $\mathcal F$.} Since $\Hsp\times\C^2$ is metrizable, it suffices to consider a sequence $(\mathcal f_k,z_k,w_k)\to(\mathcal g,z,w)$. Let $\gamma=(\bm z,\bm w):=\mathcal F(\mathcal g,z,w)$ and $\gamma_k=(\bm z_k,\bm w_k):=\mathcal F(\mathcal f_k,z_k,w_k)$, and let $J\subset\dom\gamma$ be a compact interval; we may assume that $0\in J$. We have to show that $J\subset\dom\gamma_k$ for large $k$, and that $\gamma_k\to\gamma$ in $C^1(J)$. Let $L:=1+\max_J|\bm z|$, let $\Lambda:=\sqrt2+\sup_{J\times D_L}|\mathcal g''|$, and let
        \[\eta_k:=|z_k-z|+|w_k-w|+\sup_{J\times D_L}|\mathcal f_k'-\mathcal g'|\to0.\]
    Let $\bm\Delta_k(t):=|\bm z_k(t)-\bm z(t)|+|\bm w_k(t)-\bm w(t)|$. As long as $t\in J$ and $\bm z_k$ exists and satisfies $|\bm z_k|<L$ between $0$ and $t$, we have
        \[\abs{\mathcal f_k'(t;\bm z_k)-\mathcal g'(t;\bm z)}\leq\abs{\mathcal f_k'(t;\bm z_k)-\mathcal g'(t;\bm z_k)}+\abs{\mathcal g'(t;\bm z_k)-\mathcal g'(t;\bm z)}\leq\eta_k+(\Lambda-\sqrt2)|\bm z_k-\bm z|,\]
    so that, integrating \cref{e-appendix-first-order} and using Gr\"onwall's inequality,
        \[\bm\Delta_k(t)\leq\eta_k+\int_{[0,t]}\p{\Lambda\bm\Delta_k+\eta_k},\qquad\text{and hence}\qquad\bm\Delta_k(t)\leq\eta_k\p{1+|J|}e^{\Lambda|J|}.\]
    For large $k$ the right-hand side is at most $\frac12$, and then $|\bm z_k|\leq|\bm z|+\frac12\leq L-\frac12<L$ for as long as the above applies. By a continuity argument, the above therefore applies for as long as $\gamma_k$ exists within $J$; since $\gamma_k$ remains in a compact set during this time, it exists on all of $J$, and $\sup_J\bm\Delta_k\to0$. Finally, $\dot\gamma_k\to\dot\gamma$ uniformly on $J$ by \cref{e-appendix-first-order}, because $\mathcal f_k'(t;\bm z_k)\to\mathcal g'(t;\bm z)$ uniformly on $J$ by the last display.

    \emph{(b) Continuity of $\mathcal G$.} The map $(t_0,\mathcal f)\mapsto\mathcal f(t_0+\cdot\,;\cdot)$ is continuous from $\R\times\Hsp$ to $\Hsp$, since a continuous function is uniformly continuous on compact subsets of $\R\times\C$. Hence, by (a), the map
        \[\Hsp\times\R\times\C^2\to\Csp,\qquad(\mathcal f,t_0,z,w)\mapsto\mathcal G(\mathcal f)(t_0,z,w)\]
    is continuous; in particular, $\mathcal G(\mathcal f)\in\Fsp$ for every $\mathcal f\in\Hsp$. Now let $K\subset\R\times\C^2$ be compact and let $V\subset\Csp$ be open. Then $\{\mathcal f\in\Hsp:\mathcal G(\mathcal f)(K)\subset V\}$ is the set of all $\mathcal f$ for which $\{\mathcal f\}\times K$ is contained in the preimage of $V$ under the map above, which is an open subset of $\Hsp\times\R\times\C^2$; since $K$ is compact, this set is open in $\Hsp$ (tube lemma). As the sets \cref{e-Fsp-subbasis} generate the topology of $\Fsp$, this shows that $\mathcal G$ is continuous. The set in the last assertion of (b) is $\mathcal G\inv\p{\{u\in\Fsp:u(K)\subset\Csp_T\}}$, which is open by \cref{e-Csp-T}.

    \emph{(c) The first crossing.} Since $\Csp$ is first countable, it suffices to show the following: if $\gamma\in\mathcal P(M)$ and $\gamma_k\to\gamma$ in $\Csp$, then $\gamma_k\in\mathcal P(M)$ for all large $k$, and $t_M(\gamma_k)\to t_M(\gamma)$ and $p_M(\gamma_k)\to p_M(\gamma)$. Let $t_1:=t_M(\gamma)$ and $p:=\gamma(t_1)$.

    Since $M$ is an embedded hypersurface, there are an open neighborhood $U\subset\C^2$ of $p$ and a smooth function $h:U\to\R$ with $M\cap U=h\inv(0)$ and $dh\neq0$ on $U$. By transversality, $(h\circ\gamma)'(t_1)\neq0$; replacing $h$ by $-h$, we may assume that $(h\circ\gamma)'(t_1)=2c>0$. Hence there is some $\delta_*\in(0,t_1)$ such that $I:=[0,t_1+\delta_*]\subset\dom\gamma$, such that $\gamma([t_1-\delta_*,t_1+\delta_*])$ is contained in a compact subset $U'$ of $U$, and such that $(h\circ\gamma)'\geq c$ on $[t_1-\delta_*,t_1+\delta_*]$. For large $k$ we have $I\subset\dom\gamma_k$, and $\gamma_k\to\gamma$ in $C^1(I)$.

    Let $\delta\in(0,\delta_*]$. Then
        \[h(\gamma(t_1-\delta))<0<h(\gamma(t_1+\delta)),\qquad\text{and}\qquad\gamma([0,t_1-\delta])\text{ is a compact set disjoint from the closed set }M,\]
    by the definition of $t_1$ and because $\gamma(0)\notin M$; so $\gamma([0,t_1-\delta])$ has a positive distance from $M$. Therefore, for all large $k$: the set $\gamma_k([0,t_1-\delta])$ is disjoint from $M$ (in particular $\gamma_k(0)\notin M$); $\gamma_k([t_1-\delta_*,t_1+\delta_*])$ is contained in a fixed compact neighborhood of $U'$ in $U$, on which $dh$ is uniformly continuous, so that $(h\circ\gamma_k)'\geq\frac c2$ on $[t_1-\delta_*,t_1+\delta_*]$; and $h(\gamma_k(t_1-\delta))<0<h(\gamma_k(t_1+\delta))$. Hence $h\circ\gamma_k$ has exactly one zero $t_{1,k}$ in $[t_1-\delta_*,t_1+\delta_*]$, it lies in $(t_1-\delta,t_1+\delta)$, it is the first positive time at which $\gamma_k$ meets $M$, and $\gamma_k$ is transversal to $M$ at this time. Thus $\gamma_k\in\mathcal P(M)$ and $|t_M(\gamma_k)-t_1|<\delta$ for all large $k$. Since $\delta$ was arbitrary, $t_M(\gamma_k)=t_{1,k}\to t_1$, and then
        \[|p_M(\gamma_k)-p|\leq\sup_I|\gamma_k-\gamma|+|\gamma(t_{1,k})-\gamma(t_1)|\to0.\]
\end{proof}

\section{From global completeness to local completeness}\label{s-localization}
This section concerns a local version of completeness for systems of the form
    \[\label{e-ode}
        \dot{\bm z}=\sqrt2\bm w,\qquad\dot{\bm w}=\frac1{\sqrt2}\overline{\mathcal f'(t;\bm z)}
        \]
with $\mathcal f\in\Hsp$ (see \cref{ss-appendix-continuity}); recall that $\bm z$ then solves $\ddot{\bm z}=\overline{\mathcal f'(t;\bm z)}$, and conversely. We say that $\mathcal f$ is \textit{complete} if every maximal solution of \cref{e-ode}, for every initial time, is defined on all of $\R$; by \cref{t-reduction}(c), the pp-wave $M_{\mathcal f}$ is geodesically complete if and only if $\mathcal f$ is complete.
\subsection{Locally complete functions and the proof of \texorpdfstring{\cref{main-theorem}}{the main theorem}}\label{ss-general-localization}
Our goal is to prove the following:

\begin{theorem}\label{main-theorem}
    There exists a nonquadratic (i.e. not a polynomial of degree $\leq2$) holomorphic function $\mathcal f:\C\to\C$ so that all solutions to the equation $\ddot{\bm z}(t)=\overline{\mathcal f'(\bm z(t))}$ exist for all time $t\in\R$.
    \end{theorem}

The main idea of this section is that we can reduce the global question of completeness to a local one, in the following sense:

\begin{definition}\label{d-locally-complete}
    Let $\mathcal f\in\Hsp$ and let $K\subset\C^2$ be a compact set. We say that $\mathcal f$ is $K$-\textit{locally complete} if, for every $t_0\in\R$ and every $(z,w)\in K$, the maximal solution of \cref{e-ode} with $(\bm z(t_0),\bm w(t_0))=(z,w)$ is defined on all of $\R$. We denote by $\mathfrak f_K\subset\Hsp$ the set of $K$-locally complete functions. For $a,b>0$, we say that $\mathcal f$ is $(a,b)$-\textit{locally complete} if it is $K$-locally complete for $K=\{(z,w)\in\C^2:|z|\leq a,\ \sqrt2|w|\leq b\}$, i.e. if every solution with $|\bm z(t_0)|\leq a$ and $|\dot{\bm z}(t_0)|\leq b$ at some time $t_0$ exists for all time; and we say that $\mathcal f$ is \textit{locally complete} if it is $(a,b)$-locally complete for some $a,b>0$.
    \end{definition}

Thus $\mathfrak f_{K'}\subset\mathfrak f_K$ whenever $K\subset K'$, and $\mathcal f$ is complete if and only if it is $(i,i)$-locally complete for every integer $i\geq1$. The main point is that there is a dense family of locally complete functions:

\begin{lemma}\label{l-locally-complete-family}
    Let $a,b>0$. The set of $(a,b)$-locally complete functions in $\Haut$ is dense in $\Haut$, and the set of $(a,b)$-locally complete functions in $\Hinf$ is dense in $\Hinf$.
    \end{lemma}

The proof of \cref{l-locally-complete-family} occupies \cref{ss-locally-complete-candidate,s-logic,s-focusing-lemmas}. We now show that it implies \cref{main-theorem}, and in fact more:

\begin{theorem}[Completeness is generic]\label{t-baire}
    The set of complete functions in $\Haut$ is a dense $G_\delta$ subset of $\Haut$, and the set of complete functions in $\Hinf$ is a dense $G_\delta$ subset of $\Hinf$. In particular, both sets are nonempty.
    \end{theorem}

\begin{proof}
    For a compact set $K\subset\C^2$ and an integer $i\geq1$, let $\mathfrak f_{K,i}\subset\Hsp$ denote the set of all $\mathcal f$ such that, for every $t_0\in[-i,i]$ and every $(z,w)\in K$, the maximal solution of \cref{e-ode} with $(\bm z(t_0),\bm w(t_0))=(z,w)$ is defined on $[t_0-i,t_0+i]$. By \cref{l-continuity}(b), applied with the compact set $[-i,i]\times K$ and $T=i$, the set $\mathfrak f_{K,i}$ is open in $\Hsp$, and clearly $\mathfrak f_K=\bigcap_{i\geq1}\mathfrak f_{K,i}$. Writing $K_i:=\{(z,w):|z|\leq i,\ \sqrt2|w|\leq i\}$, the set of complete functions is therefore
        \[\label{e-FU-countable-intersection}
            \bigcap_{i\geq1}\mathfrak f_{K_i}=\bigcap_{i,j\geq1}\mathfrak f_{K_i,j},
            \]
    a countable intersection of open subsets of $\Hsp$. Since the inclusion $\Haut\hookrightarrow\Hsp$ is continuous, each $\mathfrak f_{K_i,j}\cap\Haut$ is open in $\Haut$; and it contains $\mathfrak f_{K_i}\cap\Haut$, which is dense in $\Haut$ by \cref{l-locally-complete-family}. So the set of complete functions in $\Haut$ is a countable intersection of open dense subsets of $\Haut$, which is dense by the Baire category theorem, because $\Haut$ is completely metrizable (\cref{ss-appendix-continuity}). The same argument applies to $\Hinf$.
\end{proof}

\begin{proof}[Proof of \cref{main-theorem}]
    The polynomials of degree $\leq2$ form a finite-dimensional, hence closed, proper subspace of $\Haut$, and a proper closed subspace has empty interior; so the nonquadratic functions form an open dense subset of $\Haut$. By \cref{t-baire} and the Baire category theorem, this set contains a complete function.
\end{proof}

\begin{proof}[Proof of \cref{t-counterexample}]
    This is \cref{main-theorem} and \cref{c-ek-false}.
\end{proof}

\subsection{Universal pp-waves}\label{ss-universal}
The Baire category argument leaves a lot of room: any other property which is shared by a dense $G_\delta$ set of functions can be combined with completeness. In this subsection we use this to construct complete gravitational pp-waves which are \emph{universal}, in the sense that their translates approximate every other gravitational pp-wave. There are two versions of this notion, one for pp-waves with $H$ independent of $u$ and one for general pp-waves.

\begin{definition}\label{d-universal}~\\
    \begin{enumerate}
        \item[(i)] An entire function $\mathcal f\in\Haut$ is \textit{universal} if, for every entire function $\mathcal g$, every compact set $K\subset\C$ and every $\eps>0$, there exists some $z_0\in\C$ so that
            \[\label{e-universal}
                \abs{\mathcal f(z+z_0)-\mathcal g(z)}\leq\eps\qquad\text{for all }z\in K.
                \]
        \item[(ii)] A function $\mathcal f\in\Hinf$ is \textit{$u$-universal} if, for every $\mathcal g\in\Hinf$, every $S>0$, every integer $k\geq0$ and every $\eps>0$, there exists some $(t_0,z_0)\in\R\times\C$ so that
            \[\label{e-universal-time}
                \big\|\mathcal f(\cdot+t_0;\cdot+z_0)-\mathcal g\big\|_{k,S}\leq\eps,
                \]
            where $\|\cdot\|_{k,S}$ is the seminorm \cref{e-Hinf-seminorms}.
        \end{enumerate}
    We say that the gravitational pp-wave $M_{\mathcal f}$ of \cref{e-g-f} is a \textit{universal pp-wave} if $\mathcal f\in\Haut$ is universal, and a \textit{$u$-universal pp-wave} if $\mathcal f\in\Hinf$ is $u$-universal.
    \end{definition}

\begin{remark}\label{r-universal}
    \begin{enumerate}
        \item[(i)] The existence of universal entire functions is a classical theorem of Birkhoff (\cite{Birkhoff_1929}; see \cref{ss-related-work}); the content of \cref{t-universal} below is that universality is compatible with completeness. A universal function is not a polynomial (a polynomial of degree $d$ and its translates stay at a positive distance from $z^{d+1}$ in $L^\infty(\overline{D_1})$), and if $\mathcal f$ is $u$-universal, then $\mathcal f(t;\cdot)$ is not a polynomial of degree $\leq2$ for some $t$ (apply the definition with $\mathcal g(t;z)=z^3$, $k=0$ and $\eps$ smaller than the distance from $z^3$ to the quadratic polynomials in $L^\infty(\overline{D_1})$). So in both cases $M_{\mathcal f}$ is not a plane wave, by \cref{t-reduction}(d). A $u$-universal function depends on $t$, so neither of the two notions contains the other.
        \item[(ii)] In terms of the metric, \cref{d-universal} says the following. By \cref{t-reduction}(a), every gravitational pp-wave is of the form $M_{\mathcal g}$ for some $\mathcal g\in\Hinf$, and those with $H$ independent of $u$ are of the form $M_{\mathcal g}$ with $\mathcal g\in\Haut$. The translation $\tau:(u,v,z)\mapsto(u+t_0,v,z+z_0)$ satisfies $\tau^*g_{\mathcal f}=g_{\mathcal f(\cdot+t_0;\cdot+z_0)}$, and the metrics $g_{\mathcal f(\cdot+t_0;\cdot+z_0)}$ and $g_{\mathcal g}$ differ only in the coefficient of $du^2$. By the Cauchy estimates, a bound of the form \cref{e-universal} on a neighborhood of $K$ controls all $z$-derivatives of $\mathcal f(\cdot+z_0)-\mathcal g$ on $K$, and a bound of the form \cref{e-universal-time} with $S+1$ in place of $S$ controls all mixed derivatives $\partial_t^m\partial_z^l$, $m\leq k$, on $[-S,S]\times\overline{D_S}$. Hence, if $\mathcal f$ is $u$-universal (resp.\ universal), then for every gravitational pp-wave $g'$ (resp.\ every gravitational pp-wave with $H'$ independent of $u$), every compact $K'\subset\R^4$, every integer $k\geq0$ and every $\eps>0$ there is a translation $\tau$ (resp.\ a translation with $t_0=0$) so that $\|\tau^*g_{\mathcal f}-g'\|_{C^k(K')}\leq\eps$. This is the sense in which the word ``universal'' is used in \cref{t-universal-intro}.
        \end{enumerate}
    \end{remark}

We will use the following density statement.

\begin{lemma}\label{l-density-time}
    \begin{enumerate}
        \item[(a)] Let $\mathcal g\in\Hinf$, and let $c_j(t):=\tfrac1{j!}\partial_z^j\mathcal g(t;0)$ denote the Taylor coefficients of $\mathcal g(t;\cdot)$. Then each $c_j$ is smooth, and for every $k\geq0$, $S>0$ and $\eps>0$ we have $\big\|\mathcal g-\sum_{j\leq N}c_j(t)z^j\big\|_{k,S}\leq\eps$ for $N$ sufficiently large.
        \item[(b)] The polynomials in $z$ with Gaussian rational coefficients (i.e. coefficients whose real and imaginary parts are rational) are dense in $\Haut$, and the polynomials in $(t,z)$ with Gaussian rational coefficients are dense in $\Hinf$.
        \end{enumerate}
    \end{lemma}
\begin{proof}
    (a) Let $M:=\|\mathcal g\|_{k,2S}$. Differentiating $c_j(t)=\frac1{2\pi i}\oint_{|\zeta|=2S}\mathcal g(t;\zeta)\zeta^{-j-1}d\zeta$ under the integral sign, we see that $c_j$ is smooth and that $|\partial_t^mc_j(t)|\leq M(2S)^{-j}$ for $m\leq k$ and $|t|\leq S$. Hence $\big\|\sum_{j>N}c_j(t)z^j\big\|_{k,S}\leq M2^{-N}$. (b) For $\Haut$, this follows from the uniform convergence of the Taylor series on compact sets, since the finitely many coefficients of a Taylor polynomial may be perturbed so that they become Gaussian rational. For $\Hinf$, we use (a) and the fact that polynomials are dense in $C^k([-S,S])$ (approximate the $k$-th derivative uniformly by the Weierstrass approximation theorem, and integrate $k$ times) to replace each $c_j$, $j\leq N$, by a polynomial in $t$, and then perturb the finitely many coefficients so that they become Gaussian rational.
\end{proof}

\begin{theorem}\label{t-universal}
    The universal functions form a dense $G_\delta$ subset of $\Haut$, and the $u$-universal functions form a dense $G_\delta$ subset of $\Hinf$. Consequently, by \cref{t-baire} and the Baire category theorem, there exist complete universal functions $\mathcal f\in\Haut$, and there exist complete $u$-universal functions $\mathcal f\in\Hinf$; that is, by \cref{t-reduction}(c) and \cref{r-universal}(ii), there exist geodesically complete universal pp-waves, and geodesically complete $u$-universal pp-waves.
    \end{theorem}

\begin{proof}
    \emph{The $u$-independent case.} By \cref{l-density-time}(b), an entire function $\mathcal f$ is universal if and only if, for every polynomial $q$ with Gaussian rational coefficients and all integers $S,m\geq1$, there is some $z_0\in\C$ with $\sup_{\overline{D_S}}|\mathcal f(\cdot+z_0)-q|<\frac1m$: indeed, given $\mathcal g,K,\eps$ as in \cref{d-universal}(i), choose $S$ with $K\subset\overline{D_S}$, then $q$ with $\sup_{\overline{D_S}}|q-\mathcal g|\leq\frac\eps2$, and $m$ with $\frac1m\leq\frac\eps2$. There are countably many triples $(q,S,m)$, and for each of them the set
        \[\mathcal U(q,S,m):=\bigcup_{z_0\in\C}\Big\{\mathcal f\in\Haut:\sup_{\overline{D_S}}|\mathcal f(\cdot+z_0)-q|<\tfrac1m\Big\}\]
    is open in $\Haut$, being a union of open sets (the translation $\mathcal f\mapsto\mathcal f(\cdot+z_0)$ is continuous on $\Haut$). So the universal functions form the countable intersection $\bigcap_{(q,S,m)}\mathcal U(q,S,m)$ of open sets, and it remains to show that each $\mathcal U(q,S,m)$ is dense. Let $\mathcal h\in\Haut$, $S'>0$ and $\eta>0$; we have to find $\mathcal f\in\mathcal U(q,S,m)$ with $\sup_{\overline{D_{S'}}}|\mathcal f-\mathcal h|<\eta$. Choose $z_0\in\C$ with $|z_0|>S+S'$, so that the closed discs $\overline{D_{S'}}$ and $\overline{D_S(z_0)}$ are disjoint; their union is a compact set with connected complement, and the function which equals $\mathcal h(z)$ near the first disc and $q(z-z_0)$ near the second disc is holomorphic on a neighborhood of this union. By Runge's theorem, there is a polynomial $\mathcal f$ with $|\mathcal f-\mathcal h|<\eta$ on $\overline{D_{S'}}$ and $|\mathcal f(z)-q(z-z_0)|<\frac1m$ on $\overline{D_S(z_0)}$, i.e. $\sup_{\overline{D_S}}|\mathcal f(\cdot+z_0)-q|<\frac1m$. This $\mathcal f$ is as required.

    \emph{The general case.} In the same way, by \cref{l-density-time}(b), a function $\mathcal f\in\Hinf$ is $u$-universal if and only if, for every polynomial $q$ in $(t,z)$ with Gaussian rational coefficients and all integers $k\geq0$ and $S,m\geq1$, there is some $(t_0,z_0)$ with $\|\mathcal f(\cdot+t_0;\cdot+z_0)-q\|_{k,S}<\frac1m$; and for each of the countably many quadruples $(q,k,S,m)$, the set $\mathcal U(q,k,S,m)$ of these $\mathcal f$ is open in $\Hinf$, because the translations are continuous on $\Hinf$. For the density, let $\mathcal h\in\Hinf$, let $k'\geq0$, $S'>0$ and $\eta>0$; we have to find $\mathcal f\in\mathcal U(q,k,S,m)$ with $\|\mathcal f-\mathcal h\|_{k',S'}<\eta$. Here we may even take $t_0:=S'+S+2$ and $z_0:=0$: let $\chi:\R\to\ui$ be a smooth function which equals $1$ on $(-\infty,S'+\tfrac12]$ and $0$ on $[S'+1,\infty)$, and let
        \[\mathcal f(t;z):=\chi(t)\,\mathcal h(t;z)+(1-\chi(t))\,q(t-t_0;z)\in\Hinf.\]
    Then $\mathcal f=\mathcal h$ on $[-S',S']\times\C$, and $\mathcal f(t+t_0;z)=q(t;z)$ for $t\geq-S-1$, so that $\|\mathcal f-\mathcal h\|_{k',S'}=0$ and $\|\mathcal f(\cdot+t_0;\cdot)-q\|_{k,S}=0$.
\end{proof}

\begin{proof}[Proof of \cref{t-universal-intro}]
    This is \cref{t-universal}, together with \cref{r-universal}(ii). In (b), the translation may be taken with $u_0=0$ since $\mathcal f\in\Haut$ does not depend on $u$.
\end{proof}

\subsection{What the counterexamples look like}\label{ss-construction}
The proof of \cref{t-baire} does not exhibit a complete function, and it hides what the locally complete functions of \cref{l-locally-complete-family} look like. Since the Baire category theorem is itself proved by an iteration, it is worth spelling out this iteration in our case; it is how the counterexamples were originally constructed, and it is the picture that we recommend the reader keep in mind. Nothing in this subsection is used in the proofs.

The building blocks are the functions
    \[\label{e-building-block}
        A-Ae^{-A\inv p(z)},\qquad p\text{ a polynomial of degree }n\geq3,\quad A\gg1,\]
which are $(a,b)$-locally complete for $A$ sufficiently large, by \cref{l-poly-exp} below (applied to $-p$; the additive constant $A$ does not affect the equation). Where $|p(z)|\ll A$, this function is $p(z)+O(|p(z)|^2/A)$, so on any given disc it is as close to $p$ as we like once $A$ is large. Its potential $\tfrac12H=-\Re\p{A-Ae^{-A\inv p}}$ has the following shape (\cref{fig-stage2}, center). Let $c$ be the leading coefficient of $p$. Along the $n$ rays $\arg(cz^n)=0$, on which $\Re\p{p}$ grows fastest, the potential first descends like $-\Re\p{p}$ and then levels off at the terrace $-A$, which it reaches where $\Re\p{p}\sim A$, i.e. at $|z|\sim(A/|c|)^{1/n}$; in between these rays, $e^{-A\inv p}$ oscillates with an amplitude which grows like $e^{|p|/A}$. \Cref{main-lemma,t-focusing} say that a solution with bounded initial data is focused into one of the $n$ valleys by the time it reaches the scale $(A/|c|)^{1/n}$, and that it then travels along the terrace almost freely, without ever entering the oscillating sectors.

Now fix a sequence of times $T_i\to\infty$, and start with any $(1,1)$-locally complete function $\mathcal f_1$ of the form \cref{e-building-block}, for instance a rescaled copy $A_1-A_1e^{-A_1\inv(cz)^3}$ of $1-e^{-z^3}$ (\cref{main-lemma,l-rescaling-locally-complete}). Suppose that $\mathcal f_i$ has been constructed and is $(i,i)$-locally complete. By \cref{l-continuity}(b), the set of $\mathcal g\in\Haut$ for which every solution with $|\bm z(t_0)|\leq i$ and $|\dot{\bm z}(t_0)|\leq i$ exists for $|t-t_0|\leq T_i$ is open, and it contains $\mathcal f_i$; so it contains a basic neighborhood of $\mathcal f_i$, i.e. there are $R_i\geq\max(R_{i-1},i)$ and $\eps_i\leq\tfrac12\eps_{i-1}$ (with $R_0:=\eps_0:=1$) such that every entire function $\mathcal g$ with $|\mathcal g-\mathcal f_i|\leq\eps_i$ on $D_{R_i}$ has this property. Let $p_i$ be a Taylor polynomial of $\mathcal f_i$ with $|p_i-\mathcal f_i|\leq\tfrac14\eps_i$ on $D_{R_i}$, and let
    \[\mathcal f_{i+1}:=A_{i+1}-A_{i+1}e^{-A_{i+1}\inv p_i},\]
where $A_{i+1}$ is so large that $|\mathcal f_{i+1}-p_i|\leq\tfrac14\eps_i$ on $D_{R_i}$ and that $\mathcal f_{i+1}$ is $(i+1,i+1)$-locally complete (\cref{l-poly-exp}). Then $|\mathcal f_{i+1}-\mathcal f_i|\leq\tfrac12\eps_i$ on $D_{R_i}$, and since $R_i\to\infty$ and $\eps_i\leq2^{-i}\eps_1$, the sequence $\mathcal f_i$ converges locally uniformly to an entire function $\mathcal f$ with
    \[|\mathcal f-\mathcal f_i|\leq\sum_{j\geq i}\tfrac12\eps_j\leq\eps_i\qquad\text{on }D_{R_i},\text{ for every }i.\]
Given initial data with $|\bm z(t_0)|,|\dot{\bm z}(t_0)|\leq i_0$, the solution for $\mathcal f$ therefore exists for $|t-t_0|\leq T_i$ for every $i\geq i_0$, and hence for all time: $\mathcal f$ is complete. (It is nonquadratic if $\eps_1$ is smaller than the distance from $\mathcal f_1$ to the quadratic polynomials in $L^\infty(D_{R_1})$.) This is the Baire category argument with the choices made explicit; the proof of \cref{t-baire} does the bookkeeping.

The potential $-\Re\p{\mathcal f}$ of the resulting function is a staircase. On $D_{R_1}$ it is, up to $\eps_1$, the potential of $\mathcal f_1$, i.e. a rescaled copy of the potential of $1-e^{-z^3}$ (\cref{fig-profile}): three valleys which descend like $-\Re\p{(cz)^3}$ and level off at the first terrace, at height $-A_1$. Outside $D_{R_1}$, the Taylor polynomial $p_1$ takes over, and its valleys, of which there are $\deg p_1$, descend to the second terrace at height $-A_2$; then the valleys of $p_2$ descend to the third terrace at height $-A_3$, and so on, with $A_1\ll A_2\ll A_3\ll\cdots$. In between the valleys, the potential oscillates with enormous amplitude. A particle with bounded data is focused into a valley of stage $i$ and then travels along the $i$-th terrace; when it leaves $D_{R_i}$ it enters the region where $\mathcal f$ is governed by $p_i$, and there it is a particle with small data at the scale $(A_{i+1}/|c_i|)^{1/\deg p_i}$ of the next terrace, so that it is focused again, this time into a valley of $p_i$ (the valleys of consecutive stages need not be aligned), and travels along the $(i+1)$-st terrace, and so on. One step of the iteration is pictured in \cref{fig-stage2}; the figure shows, in particular, the descent from the first terrace to the second one along the rays, and the fact that the two functions are indistinguishable on the disc where the first one is locally complete.

\begin{figure}[htbp]
    \centering
    \includegraphics[width=\textwidth]{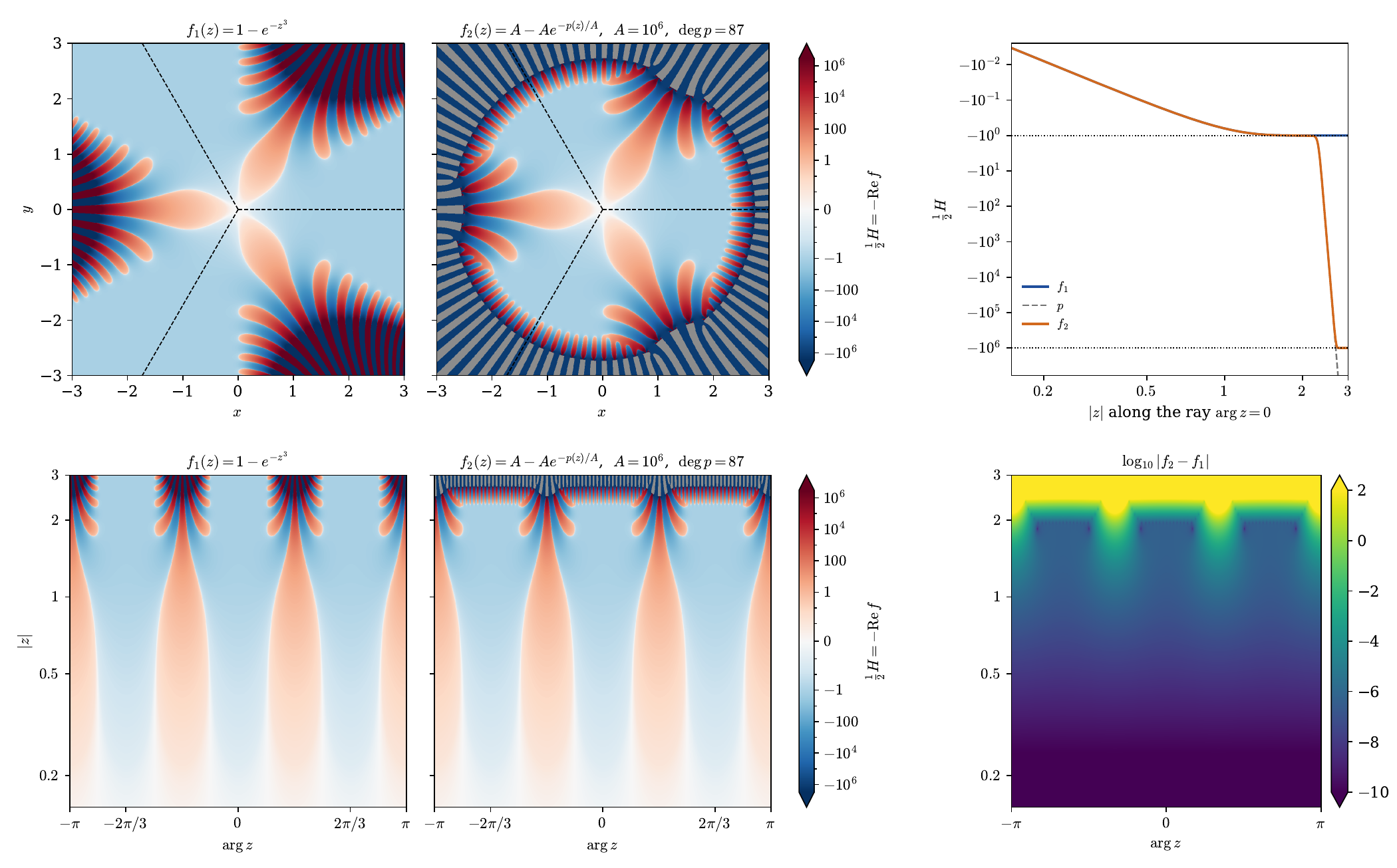}
    \caption{One step of the iteration of \cref{ss-construction}: $\mathcal f_1(z)=1-e^{-z^3}$ and $\mathcal f_2(z)=A-Ae^{-A\inv p(z)}$, where $A=10^6$ and $p(z)=\sum_{k=1}^{29}\frac{(-1)^{k+1}}{k!}z^{3k}$ is the Taylor polynomial of $\mathcal f_1$ of degree $87$. Top left and center: the potentials $\tfrac12H=-\Re\p{\mathcal f_1}$ and $-\Re\p{\mathcal f_2}$ in the $z$-plane; the dashed rays are the directions $\arg(z^3)=0$. Bottom left and center: the same potentials as functions of $\arg z$ and $|z|$, with $|z|$ on a logarithmic scale. Top right: the potentials of $\mathcal f_1$, $p$ and $\mathcal f_2$ along the ray $\arg z=0$. Bottom right: the error $\log_{10}|\mathcal f_2-\mathcal f_1|$. The color scale of the potentials is logarithmic, with linear threshold $0.01$, and saturates at $\pm3\times10^6$; in the gray regions the potential is beyond this range and oscillates faster than the resolution of the plot. For $|z|\lesssim2$ the two functions are indistinguishable at this resolution: numerically, $\sup|\mathcal f_2-\mathcal f_1|$ is approximately $1.5\times10^{-6}$ on $D_1$, $1.8\times10^{-4}$ on $D_{1.44}$, $9\times10^{-3}$ on $D_{1.7}$ and $4.4$ on $D_2$ (where $\sup|\mathcal f_1|\approx3\times10^3$). For these parameters the error is $\approx|\mathcal f_1|^2/2A$; the error from truncating the Taylor series is less than $10^{-5}$ on $D_2$. The approximation breaks down at $|z|\approx2.3$, where $|\mathcal f_1|$ becomes comparable to $A$ in the sectors in which it is large, and where the term $z^{87}/29!$ begins to dominate along the rays $\arg(z^3)=0$. Along these rays the potential of $\mathcal f_2$ then descends from its first terrace at height $-1$ to a second terrace at height $-A$, which it reaches at $|z|\approx2.7$ (top right); the $87$ valleys of $\mathcal f_2$ at height $-A$ are the dark blue stripes outside the circle $|z|\approx2.7$ in the top center panel, and they alternate with the $87$ sectors in which the potential oscillates between values far beyond $\pm A$. The parameters were chosen for legibility; no claim is made that this particular $\mathcal f_2$ is $(a,b)$-locally complete for any specific $a,b$.}
    \label{fig-stage2}
\end{figure}

\subsection{A candidate for a locally complete family}\label{ss-locally-complete-candidate}
Note that locally complete functions satisfy the following rescaling property:

\begin{lemma}\label{l-rescaling-locally-complete}
    Suppose that $\mathcal f\in\Haut$ is $(a,b)$-locally complete, for some $a,b>0$. Then the following are true for any real $R>0$ and any complex $\lambda\neq0$:
        \begin{enumerate}
            \item The function $z\mapsto \mathcal f(\lambda z)$ is $(|\lambda|\inv a,b)$-locally complete.
            \item The function $z\mapsto R\mathcal f(z)$ is $(a,R\ef12b)$-locally complete.
            \item The function $z\mapsto \mathcal f(z)-\lambda$ is $(a,b)$-locally complete.
        \end{enumerate}
    \end{lemma}

The proof of \cref{l-rescaling-locally-complete} is trivial from the rescaling properties of the ODE $\ddot{\bm z}(t)=\overline{\mathcal f'(\bm z(t))}$: if $\bm z(t)$ solves this ODE, then $\bm\zeta(t):=\lambda\inv\bm z(|\lambda|t)$ solves $\ddot{\bm\zeta}=\overline{\mathcal f_\lambda'(\bm\zeta)}$ where $\mathcal f_\lambda(z)=\mathcal f(\lambda z)$, and $|\bm\zeta(0)|=|\lambda|\inv|\bm z(0)|$, $|\dot{\bm\zeta}(0)|=|\dot{\bm z}(0)|$; this gives (1). Similarly $\bm\zeta(t):=\bm z(R\ef12t)$ solves the ODE for $R\mathcal f$, which gives (2), and (3) is immediate.

As we will see, \cref{l-locally-complete-family} will follow from the following remarkable fact that bridges exponentials and polynomials:

\begin{lemma}\label{l-poly-exp}
    Let $n\geq3$, let $a,b>0$, and let
        \[p(t;z)=c_nz^n+\sum_{j<n}c_j(t)z^j,\]
    where $c_n\in\C\setminus\{0\}$ is a \emph{constant} and where each $c_j:\R\to\C$, $j<n$, is a bounded $C^1$ function with $\int_\R|\dot c_j|<\infty$. Then for $A$ sufficiently large, $\mathcal f(t;z)=-Ae^{A\inv p(t;z)}$ is $(a,b)$-locally complete. In particular, this holds when $p(z)$ is any polynomial of degree $\geq3$ with constant coefficients.
\end{lemma}

\begin{proof}[Proof of \cref{l-locally-complete-family} using \cref{l-poly-exp}]
    \emph{The space $\Haut$.} The polynomials of degree $\geq3$ are dense in $\Haut$. Indeed, all polynomials are dense (the Taylor series of an entire function converges uniformly on compact sets), and any polynomial $p$ of degree $\leq2$ is the limit in $\Haut$ of the polynomials $p(z)+\eta z^3$ of degree $3$ as $\eta\to0$. Next, for a polynomial $p$ of degree $\geq3$, note that
        \[p(z)=\lim_{A\to\infty}\p{A-Ae^{-A\inv p(z)}}\]
    uniformly on bounded sets, and that $A-Ae^{-A\inv p(z)}$ is $(a,b)$-locally complete for $A$ sufficiently large by \cref{l-poly-exp} (applied to $-p$) and \cref{l-rescaling-locally-complete}(3). Hence every polynomial of degree $\geq3$ is a limit of $(a,b)$-locally complete entire functions, and the claim follows.

    \emph{The space $\Hinf$.} Let $\mathcal h\in\Hinf$, $k\geq0$, $S>0$ and $\eps>0$; we have to find an $(a,b)$-locally complete $\mathcal f\in\Hinf$ with $\|\mathcal f-\mathcal h\|_{k,S}\leq\eps$. By \cref{l-density-time}(a), there is some $N$ so that the Taylor polynomial $\sum_{j\leq N}c_j(t)z^j$ of $\mathcal h$ is within $\frac13\eps$ of $\mathcal h$ in $\|\cdot\|_{k,S}$. Let $\chi:\R\to\ui$ be a smooth cutoff which equals $1$ on $[-S,S]$ and is supported in $[-S-1,S+1]$. The top coefficient $c_N(t)$ may depend on $t$ (and may vanish), so we add a term of higher degree with a small constant coefficient: let $n:=\max(N+1,3)$, let $\eta>0$ be so small that $\eta S^n\leq\frac13\eps$, and let
        \[p(t;z):=\eta z^n+\sum_{j\leq N}\chi(t)c_j(t)z^j,\qquad\text{so that}\qquad\|p-\mathcal h\|_{k,S}\leq\tfrac23\eps.\]
    The coefficients $\chi c_j$ are smooth and compactly supported, so $-p$ satisfies the hypotheses of \cref{l-poly-exp}. Hence $\mathcal f_A:=A-Ae^{-A\inv p}\in\Hinf$ is $(a,b)$-locally complete for $A$ sufficiently large (the additive constant $A$ does not affect the equation). Finally, write $e^x-1-x=x^2\Theta(x)$ with $\Theta$ entire, so that
        \[\mathcal f_A-p=-A\inv p^2\,\Theta(-A\inv p).\]
    Since $p$ and its $t$-derivatives of order $\leq k$ are bounded on $\R\times D_S$, the chain rule shows that $\|\mathcal f_A-p\|_{k,S}\lesssim A\inv$ for $A\geq1$, with a constant depending on $p$, $k$ and $S$. So $\|\mathcal f_A-\mathcal h\|_{k,S}\leq\eps$ for $A$ sufficiently large.
\end{proof}

It will be more technically convenient to work with the following version, which implies \cref{l-poly-exp} (as we show below). It allows the perturbation to depend on time. This costs very little in the proof, and it is what the $u$-universal case of \cref{t-baire} needs; the reader who is only interested in \cref{main-theorem} may assume throughout that the perturbations do not depend on $t$.

\begin{definition}\label{d-F-eps}
    Let $n\geq3$ be an integer, let $\mathcal f(z)=1-e^{-z^n}$, and let $\eps\geq0$. We denote by $F_\eps$ the set of all functions $\tilde{\mathcal f}:\R\times\C\to\C$ of the form
        \[\label{f-tilde-conditions}
            \tilde{\mathcal f}(t;z)=1-e^{-z^n+\mathcal p(t;z)},\qquad\mathcal p(t;z)=\sum_{j<n}a_j(t)z^j,\qquad a_j\in C^1(\R;\C),
            \]
    whose coefficients satisfy
        \[\label{e-F-eps-bounds}
            \sup_{t\in\R}|a_j(t)|\leq\eps\qquad\text{for all }j<n.
            \]
    For such a function, we write
        \[\label{e-total-variation}
            V(\mathcal p):=\sum_{j<n}\int_\R|\dot a_j(t)|dt\in[0,\infty]
            \]
    for the total variation of the coefficients of $\mathcal p$.
    \end{definition}

Thus $F_0=\{\mathcal f\}$, and $F_\eps\subset F_{\eps_1}$ whenever $\eps\leq\eps_1$. If $\mathcal p$ is a polynomial of degree $<n$ (with constant coefficients) whose coefficients all have modulus $\leq\eps$, then $1-e^{-z^n+\mathcal p(z)}$ belongs to $F_\eps$, and $V(\mathcal p)=0$. Note also that $F_\eps$ is invariant under time translations and under time reversal: if $\tilde{\mathcal f}\in F_\eps$ and $t_0\in\R$, then $\tilde{\mathcal f}(t_0+\cdot\,;\cdot)$ and $\tilde{\mathcal f}(-\cdot\,;\cdot)$ belong to $F_\eps$, and they have the same total variation. The pointwise bound \cref{e-F-eps-bounds} is all that is needed for the focusing results of \cref{s-logic,s-focusing-lemmas}. The total variation only enters in \cref{main-lemma}: when $\mathcal p$ depends on $t$, the energy is no longer conserved, and a solution with small data may remain near the degenerate critical point $z=0$ for an arbitrarily long time, during which $V(\mathcal p)$ is what controls the drift of the energy (see \cref{e-energy-drift} below).

\begin{lemma}\label{main-lemma}
    Let $n\geq3$ be a fixed integer and let ${\mathcal f}(z)=1-e^{-z^n}$. Then there exist $R>0$ and $\eps>0$ so that the following holds for every $\tilde{\mathcal f}\in F_\eps$ with $V(\mathcal p)\leq\eps$ and every $t_*\in\R$: any solution ${\bm z}(t)$ to $\ddot{\bm z}(t)=\overline{\tilde{\mathcal f}'(t;\bm z(t))}$ with $|{\bm z}(t_*)|\leq R$ and $|\dot{\bm z}(t_*)|^2\leq R^n$ exists for all time $t\in\R$. In particular, this holds for $\mathcal f$ itself, and for $\tilde {\mathcal f}(z)=1-e^{-z^n+\mathcal p(z)}$ whenever $\mathcal p(z)$ is a polynomial of degree $<n$ whose coefficients all have modulus $\leq\eps$.
    \end{lemma}

Throughout \cref{s-logic,s-focusing-lemmas}, $n\geq3$ denotes the fixed integer of \cref{main-lemma}. (In the terminology of \cref{d-locally-complete}, \cref{main-lemma} says that every $\tilde{\mathcal f}\in F_\eps$ with $V(\mathcal p)\leq\eps$, and in particular $\mathcal f$ itself, is $(R,R\ef n2)$-locally complete.) The potential associated with $\mathcal f$ is pictured in \cref{fig-profile}.
\begin{figure}[htbp]
    \centering
    \includegraphics[width=\textwidth]{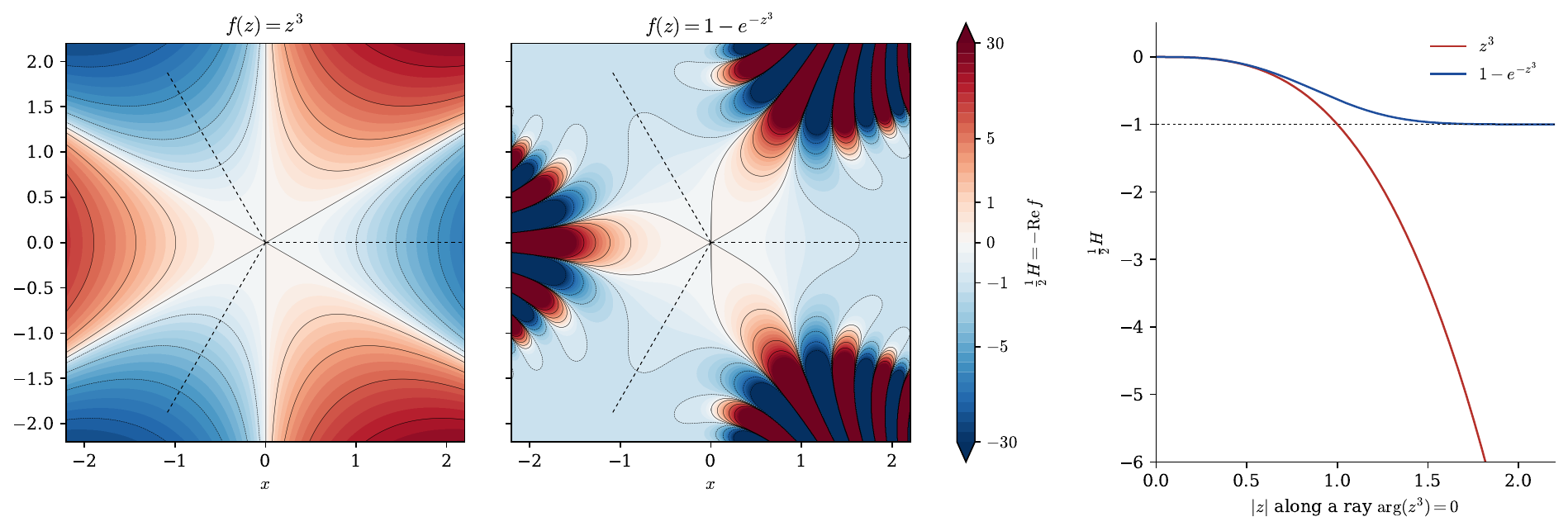}
    \caption{The potential $\tfrac12H=-\Re\p{\mathcal f}$ of $\ddot{\bm z}=\overline{\mathcal f'(\bm z)}$ for $\mathcal f(z)=z^3$ (left) and for the function $\mathcal f(z)=1-e^{-z^3}$ of \cref{main-lemma} (center); the color scale is logarithmic and saturates at $\pm30$, and the dashed rays are the directions $\arg(z^3)=0$. Near the origin the two potentials agree to leading order, and particles are pushed into the three valleys along the dashed rays. For $z^3$ the valleys fall away like $-|z|^3$, which leads to blow-up in finite time (\cref{t-polynomial}); for $1-e^{-z^3}$ they level off at $-1$ (right), so that a particle which stays close to a ray is eventually almost free. In the remaining sectors $1-e^{-z^3}$ oscillates with amplitude $e^{|z|^3}$; the content of \cref{main-lemma} is that solutions with small initial data never see these sectors. We emphasize that $1-e^{-z^3}$ is only a locally complete building block, and is not itself the function of \cref{main-theorem}.}
    \label{fig-profile}
\end{figure}

\begin{proof}[Proof of \cref{l-poly-exp} using \cref{main-lemma}]
    Let $R,\eps$ be as in \cref{main-lemma}. For $A>0$, choose $\mu\in\C$ with $\mu^n=-c_n/A$, so that $|\mu|=(|c_n|/A)\ef1n$, let $\sigma:=A\ef12|\mu|=|c_n|\ef1nA^{\frac12-\frac1n}$, and set
        \[\mathcal p(s;\zeta):=A\inv p(\sigma\inv s;\mu\inv\zeta)+\zeta^n=\sum_{j<n}a_j(s)\zeta^j,\qquad a_j(s):=A\inv\mu^{-j}c_j(\sigma\inv s),\]
    and $\tilde{\mathcal f}(s;\zeta):=1-e^{-\zeta^n+\mathcal p(s;\zeta)}$, so that $\mathcal f(t;z)=A\tilde{\mathcal f}(\sigma t;\mu z)-A$. If $\bm z(t)$ solves $\ddot{\bm z}=\overline{\mathcal f'(t;\bm z)}$, then $\bm\zeta(s):=\mu\bm z(\sigma\inv s)$ satisfies
        \[\frac{d^2\bm\zeta}{ds^2}=\mu\sigma^{-2}\,\overline{A\mu\tilde{\mathcal f}'(s;\bm\zeta)}=\overline{\tilde{\mathcal f}'(s;\bm\zeta)}.\]
    Since the total variation is invariant under reparametrization,
        \[\sup_s|a_j(s)|=|c_n|^{-\frac jn}A^{\frac jn-1}\sup_t|c_j(t)|\qquad\text{and}\qquad\int_\R\abs{\tfrac{da_j}{ds}}ds=|c_n|^{-\frac jn}A^{\frac jn-1}\int_\R|\dot c_j|dt,\]
    and both tend to $0$ as $A\to\infty$, since $j<n$; so $\tilde{\mathcal f}\in F_\eps$ and $V(\mathcal p)\leq\eps$ (see \cref{e-F-eps-bounds,e-total-variation}) for $A$ sufficiently large. Moreover, if $|\bm z(t_0)|\leq a$ and $|\dot{\bm z}(t_0)|\leq b$, then at $s_0=\sigma t_0$ we have $|\bm\zeta(s_0)|\leq|\mu|a=(|c_n|/A)\ef1na\leq R$ and $|\tfrac{d\bm\zeta}{ds}(s_0)|\leq|\mu|\sigma\inv b=A\ief12b\leq R\ef n2$ for $A$ sufficiently large. By \cref{main-lemma} (with $t_*=s_0$), $\bm\zeta$ exists for all $s\in\R$, and hence $\bm z$ exists for all $t\in\R$.
\end{proof}

\section{Computing the coordinate transformation}\label{ss-coordinate-transformation}
Let $I\subset\R$ be an open interval, let $\Omega\subset\C$ be an open subset of $\C$, let $\mathcal f:I\times\Omega\to\C$ be a $C^1$ function holomorphic in the second argument, let $\mathcal f'=\partial_z\mathcal f$, and let $\mathcal f_t=\partial_t\mathcal f$. Fix an integer $n\ge2$ and put $\beta=\tfrac12+\tfrac1n$. We denote by $(\bm z,\bm w):J\to\Omega\times\C$ a solution to \cref{e-ode}, i.e.\ to
\[
\dot{\bm z}=\sqrt2\,\bm w,\qquad \dot{\bm w}=\tfrac1{\sqrt2}\,\overline{\mathcal f'(t;\bm z)} .
\]
Along the solution write $\bm f:=\mathcal f(t;\bm z)$, $\bm f_t:=\mathcal f_t(t;\bm z)$, and
\[
\bm G:=\tfrac1n\,\bm z\,\mathcal f'(t;\bm z),\qquad \bm\kappa:=|\bm w|^2-\Re\p{\bm f},\qquad \bm c:=\frac{\bm\kappa}{\Re\p{\bm f}}\quad(\text{where }\Re\p{\bm f}\ne0).
\]
Note that $\bm G=\bm f$ when $\mathcal f(t;z)=z^n$ (the model case).
\begin{lemma}[Change of variables]\label{l-cov}
Let $(\bm z,\bm w)$ solve \cref{e-ode} on $J$ and suppose
\begin{equation}\label{e-hyp}
\bm z(t)\neq0,\qquad \tfrac d{dt}|\bm z(t)|>0,\qquad \Re\p{\mathcal f(t;\bm z(t))}\ne0\qquad\text{for all }t\in J.
\end{equation}
Let $\bm\zeta=n\log\bm z$ for a continuous branch of $\log\bm z$ on $J$, and $\bm\xi=\log(\bm w/\bm z)$ with the principal branch. Write $\bm\zeta=\bm r+i\bm\theta$, $\bm\xi=\bm s+i\bm\varphi$, so that
\[
\bm z=e^{\bm\zeta/n},\quad \bm w=e^{\bm\xi+\bm\zeta/n},\quad
\bm r=n\log|\bm z|,\quad \bm\theta\equiv\arg(\bm z^n),\quad \bm s=\log\tfrac{|\bm w|}{|\bm z|},\quad \bm\varphi=\arg\bm w-\arg\bm z\in\p{-\tfrac\pi2,\tfrac\pi2}.
\]
Then $\dot{\bm r}=\sqrt2\,n\,\frac{|\bm w|}{|\bm z|}\cos\bm\varphi>0$, so $r=\bm r(t)$ may be used as the independent variable, and with primes denoting $d/dr$,
\begin{align}
\bm\theta'&=\tan\bm\varphi,\label{e-theta-prime}\\
\bm\varphi'&=-\tfrac1n\tan\bm\varphi-\frac{1}{2(1+\bm c)}\p{\frac{\Im\p{\bm G}}{\Re\p{\bm f}}+\frac{\Re\p{\bm G}}{\Re\p{\bm f}}\tan\bm\varphi},\label{e-varphi-prime}\\
\bm t'&=\frac{|\bm z|}{\sqrt2\,n\,|\bm w|\cos\bm\varphi},\label{e-t-prime}\\
\bm c'&=-\bm c\p{\frac{\Re\p{\bm G}}{\Re\p{\bm f}}-\frac{\Im\p{\bm G}}{\Re\p{\bm f}}\tan\bm\varphi}-(1+\bm c)\,\bm t'\,\frac{\Re\p{\bm f_t}}{\Re\p{\bm f}}.\label{e-c-prime}
\end{align}
Moreover $|\bm w|^2=(1+\bm c)\Re\p{\bm f}$ and $\dot{\bm\kappa}=-\Re\p{\bm f_t}$.
\end{lemma}

\begin{proof}
\emph{Coordinates.} Since $\bm z$ is continuous and nonvanishing on the interval $J$, a $C^1$ branch of $\log\bm z$ exists. Next, $\frac d{dt}|\bm z|^2=2\Re\p{\bar{\bm z}\dot{\bm z}}=2\sqrt2\,\Re\p{\bar{\bm z}\bm w}$, so the second condition in \cref{e-hyp} says $\Re\p{\bar{\bm z}\bm w}>0$; hence $\bm w/\bm z=\bar{\bm z}\bm w/|\bm z|^2$ lies in the open right half-plane and $\bm\xi=\log(\bm w/\bm z)$ (principal branch) is $C^1$ with $|\Im\p{\bm\xi}|<\frac\pi2$. The identities $\bm z=e^{\bm\zeta/n}$, $\bm w=\bm z e^{\bm\xi}$ and the formulas for $\bm r,\bm s,\bm\varphi$ follow. Note also
\begin{equation}\label{e-zbarw}
\bar{\bm z}\bm w=|\bm z|^2e^{\bm\xi}=|\bm z|^2e^{\bm s}e^{i\bm\varphi}=|\bm z|\,|\bm w|\,e^{i\bm\varphi}.
\end{equation}

\emph{Equations for $\bm\zeta,\bm\xi$.} From \cref{e-ode}, $\dot{\bm\zeta}=n\dot{\bm z}/\bm z=\sqrt2\,n\,\bm w/\bm z=\sqrt2\,n\,e^{\bm\xi}$. Since $\mathcal f'(t;\bm z)=n\bm G/\bm z$, \cref{e-ode} and \cref{e-zbarw} give
\[
\frac{\dot{\bm w}}{\bm w}=\frac{\overline{\mathcal f'(t;\bm z)}}{\sqrt2\,\bm w}=\frac{n\,\overline{\bm G}}{\sqrt2\,\bar{\bm z}\bm w}=\frac{n\,\overline{\bm G}}{\sqrt2\,|\bm z||\bm w|}e^{-i\bm\varphi},
\qquad
\dot{\bm\xi}=\frac{\dot{\bm w}}{\bm w}-\frac{\dot{\bm z}}{\bm z}=\frac{n\,\overline{\bm G}}{\sqrt2\,|\bm z||\bm w|}e^{-i\bm\varphi}-\sqrt2\,e^{\bm\xi}.
\]

\emph{Energy.} Using \cref{e-ode} and $\frac d{dt}\Re\p{\bm f}=\Re\p{\mathcal f'\dot{\bm z}}+\Re\p{\bm f_t}$,
\[
\dot{\bm\kappa}=2\Re\p{\bar{\bm w}\dot{\bm w}}-\Re\p{\mathcal f'\dot{\bm z}}-\Re\p{\bm f_t}
=\sqrt2\Re\p{\overline{\bm w\mathcal f'}}-\sqrt2\Re\p{\mathcal f'\bm w}-\Re\p{\bm f_t}=-\Re\p{\bm f_t},
\]
and $|\bm w|^2=\Re\p{\bm f}+\bm\kappa=(1+\bm c)\Re\p{\bm f}$ by definition of $\bm c$.

\emph{Reparameterization.} Real and imaginary parts of $\dot{\bm\zeta}=\sqrt2\,n\,e^{\bm\xi}$ give
\[
\dot{\bm r}=\sqrt2\,n\,e^{\bm s}\cos\bm\varphi=\sqrt2\,n\,\frac{|\bm w|}{|\bm z|}\cos\bm\varphi>0,\qquad
\dot{\bm\theta}=\sqrt2\,n\,e^{\bm s}\sin\bm\varphi,
\]
whence $\bm\theta'=\dot{\bm\theta}/\dot{\bm r}=\tan\bm\varphi$ and $\bm t'=1/\dot{\bm r}$, i.e.\ \cref{e-theta-prime}, \cref{e-t-prime}. Taking the imaginary part of $\dot{\bm\xi}$ and using $\Im\p{\overline{\bm G}e^{-i\bm\varphi}}=-\Im\p{\bm Ge^{i\bm\varphi}}=-(\Im\p{\bm G}\cos\bm\varphi+\Re\p{\bm G}\sin\bm\varphi)$,
\[
\dot{\bm\varphi}=-\frac{n\,(\Im\p{\bm G}\cos\bm\varphi+\Re\p{\bm G}\sin\bm\varphi)}{\sqrt2\,|\bm z||\bm w|}-\sqrt2\,e^{\bm s}\sin\bm\varphi .
\]
Dividing by $\dot{\bm r}$ and using $|\bm z|e^{\bm s}=|\bm w|$ and $|\bm w|^2=(1+\bm c)\Re\p{\bm f}$,
\[
\bm\varphi'=-\frac{\Im\p{\bm G}+\Re\p{\bm G}\tan\bm\varphi}{2|\bm w|^2}-\tfrac1n\tan\bm\varphi
=-\tfrac1n\tan\bm\varphi-\frac{1}{2(1+\bm c)}\p{\frac{\Im\p{\bm G}}{\Re\p{\bm f}}+\frac{\Re\p{\bm G}}{\Re\p{\bm f}}\tan\bm\varphi},
\]
which is \cref{e-varphi-prime}.

\emph{Normalized energy.} Put $F=\Re\p{\bm f}$. Along the solution, by \cref{e-ode} and $\bm w/\bm z=e^{\bm s}e^{i\bm\varphi}$,
\[
\dot F=\Re\p{\mathcal f'\dot{\bm z}}+\Re\p{\bm f_t}=\sqrt2\,\Re\p{\tfrac{n\bm G}{\bm z}\bm w}+\Re\p{\bm f_t}
=\sqrt2\,n\,e^{\bm s}\Re\p{\bm Ge^{i\bm\varphi}}+\Re\p{\bm f_t} .
\]
Dividing by $\dot{\bm r}=\sqrt2\,n\,e^{\bm s}\cos\bm\varphi$ and using $\Re\p{\bm Ge^{i\bm\varphi}}=\Re\p{\bm G}\cos\bm\varphi-\Im\p{\bm G}\sin\bm\varphi$,
\[
F'=\Re\p{\bm G}-\Im\p{\bm G}\tan\bm\varphi+\bm t'\Re\p{\bm f_t},\qquad \bm\kappa'=-\bm t'\Re\p{\bm f_t} .
\]
Therefore
\[
\bm c'=\frac{\bm\kappa'}{F}-\frac{\bm\kappa F'}{F^2}
=-\frac{\bm t'\Re\p{\bm f_t}}{F}-\bm c\p{\frac{\Re\p{\bm G}}{F}-\frac{\Im\p{\bm G}}{F}\tan\bm\varphi}-\bm c\,\frac{\bm t'\Re\p{\bm f_t}}{F},
\]
which is \cref{e-c-prime}.
\end{proof}

\begin{remark}
Only the values of $\mathcal f$ along the trajectory enter, so $\Omega$ may be any open set containing $\bm z(J)$, e.g.\ a sector. The lemma is purely kinematic: whether \cref{e-hyp} persists is a property of the particular $\mathcal f$.
\end{remark}
\begin{remark}
    As stated in the introduction, \cref{l-cov} and its consequence for the model case (\cref{c-model} below) do most of the work necessary to prove the Ehlers--Kundt conjecture for polynomials; we work out the rest of the details in \cref{s-polynomial-conjecture}.
\end{remark}

\bigskip
\begin{corollary}[Model case]\label{c-model}
Let $n\ge2$ and $\mathcal f(t;z)=z^n$. Let $\Qsq=(-\frac\pi2,\frac\pi2)^2$, and let $\overline\Qsq=[-\frac\pi2,\frac\pi2]^2$ denote its closure. Let $(\bm z,\bm w)$ solve \cref{e-ode} with
    \[\label{e-model-hypotheses}
        \bm z(t_0)\neq0,\qquad\Re\p{\overline{\bm z(t_0)}\bm w(t_0)}\geq0,\qquad\text{and}\qquad\bm\kappa(t_0)=0,\text{ i.e. }|\bm w(t_0)|^2=\Re\p{\bm z(t_0)^n}.\]
Then the following are true:
\begin{enumerate}
\item[(1)] When $t>t_0$, $\frac d{dt}|\bm z|>0$, $\bm\kappa=0$, $\Re\p{\bm z^n}>0$, and $(\bm\theta,\bm\varphi)\in\Qsq$.
\item[(2)] When $t>t_0$, $\bm\theta,\bm\varphi$ satisfy
\begin{equation}\label{e-model}
\bm\theta'=\tan\bm\varphi,\qquad \bm\varphi'=-\tfrac12\tan\bm\theta-\beta\tan\bm\varphi,\qquad \beta=\tfrac12+\tfrac1n .
\end{equation}
\item[(3)] There is a continuous function $\Phi:\overline\Qsq\to[0,\infty)$, which is $C^\infty$ on $\Qsq$, such that the following hold:
\begin{enumerate}
\item[(a)] We have $\frac d{dr}\Phi(\bm\theta,\bm\varphi)=-\beta\,\Phi(\bm\theta,\bm\varphi)$,
    \[\label{e-Phi-decay}
        \frac{\Phi(\bm\theta(t),\bm\varphi(t))}{\Phi(\bm\theta(t_0),\bm\varphi(t_0))}=\p{\frac{|\bm z(t)|}{|\bm z(t_0)|}}^{-\beta n},\]
        and
    \[\frac{|\bm\theta(t)|+|\bm\varphi(t)|}{|\bm\theta(t_0)|+|\bm\varphi(t_0)|}\lesssim\p{\frac{|\bm z(t)|}{|\bm z(t_0)|}}\ief{\beta n}2.\]
\item[(b)] $\Phi(\theta,\varphi)=\tfrac12\theta^2+\beta\theta\varphi+\varphi^2+O(\theta^4+\varphi^4).$
\item[(c)] $\Phi(\theta,\varphi)\sim\theta^2+\varphi^2$ on all of $\overline\Qsq$, with constants depending only on $n$; in particular, $\Phi$ is bounded.
\end{enumerate}
\item[(4)] The forward solution ceases to exist within time $\lesssim|\bm z(t_0)|\ief{n-2}2$ if $n\geq3$.
\end{enumerate}
\end{corollary}
\begin{proof}
    Throughout, $\phi_\tau$ denotes the flow of the autonomous system \cref{e-model} on $\Qsq$.
    \begin{enumerate}
        \item
            Since $\mathcal f_t=0$, the energy is conserved, so $\bm\kappa=0$, i.e. $|\bm w|^2=\Re\p{\bm z^n}$ wherever the solution exists. Put $\bm\psi:=\tfrac12|\bm z|^2$. Then
                \[\ddot{\bm\psi}=|\dot{\bm z}|^2+\Re\p{\bar{\bm z}\ddot{\bm z}}=2|\bm w|^2+n\Re\p{\bm z^n}=(n+2)|\bm w|^2\geq0,\]
            with equality iff $\bm\theta=\pm\tfrac\pi2$ (and in this case, $|\dot{\bm w}|>0$, so $\bm\psi^{(4)}>0$). In particular, since $\dot{\bm\psi}(t_0)\geq0$, we conclude that $\dot{\bm\psi}(t)>0$ when $t>t_0$, so $\frac d{dt}|\bm z|>0$ and $|\bm w|^2=\Re\p{\bm z^n}>0$. This implies that $(\bm\theta,\bm\varphi)\in\Qsq$ for $t>t_0$, and the hypotheses of \cref{l-cov} are met. Furthermore, $\phi_\tau$ has a continuous extension to $\bar\Qsq$ which satisfies $\phi_\tau(\bar\Qsq)\subset\Qsq$ for all $\tau>0$.
        \item
            With $\mathcal f=z^n$ we have $\bm G=\tfrac1n\bm z\cdot n\bm z^{n-1}=\bm z^n=e^{\bm\zeta}$ and $\Re\p{\bm f}=\Re\p{\bm z^n}=e^{\bm r}\cos\bm\theta$, so
            \[
            \frac{\Re\p{\bm G}}{\Re\p{\bm f}}=1,\qquad \frac{\Im\p{\bm G}}{\Re\p{\bm f}}=\frac{e^{\bm r}\sin\bm\theta}{e^{\bm r}\cos\bm\theta}=\tan\bm\theta,
            \]
            and $\bm c=0$ by (1). Substituting into \cref{e-varphi-prime} gives $\bm\varphi'=-\tfrac1n\tan\bm\varphi-\tfrac12(\tan\bm\theta+\tan\bm\varphi)$, which is \cref{e-model}.
            
        \item
            By \cite{Sternberg_1957} (a refinement of Grobman-Hartman), there is a local $C^\infty$ conjugacy to the linearized version of the problem near the $(0,0)$ equilibrium; it follows that there is some function $\tilde\Phi(\theta,\varphi)=\tfrac12\theta^2+\beta\theta\varphi+\varphi^2+O(\theta^4+\varphi^4)$ which satisfies (a) on some open set $U\owns(0,0)$. Let
            \[
            \Lambda(\theta,\varphi):=1-\cos\theta\cos^2\varphi\in[0,1]\quad\text{on }\bar\Qsq .
            \]
            and note that $\Lambda'=-2\beta\cos\theta\sin^2\varphi\le0$ along integral curves of \cref{e-model}. That is, the sublevel sets of $\Lambda$ are forward-invariant under $\phi$ and the monotone convergence theorem implies that $\Lambda$ has some limit along integral curves of \cref{e-model}; thus, $\Lambda'\to0$ and the limit points of the integral curves of \cref{e-model} lie in the zero set $\{(\theta,\varphi)\in\Qsq:\Lambda'(\theta,\varphi)=0\}=\{(\theta,\varphi)\in\Qsq:\varphi=0\}$ of $\Lambda'$, which contains only one equilibrium, namely $(0,0)$. In other words, every integral curve of \cref{e-model} converges to $(0,0)$ as $r\to\infty$.
            Let $a>0$ be so small that $K:=\tilde\Phi\inv([0,a])$ is a compact subset of $U$ on which $\tilde\Phi\sim\theta^2+\varphi^2$ (the quadratic part of $\tilde\Phi$ is positive definite, since $\beta^2<2$). Then $K$ is forward invariant, because $\tilde\Phi$ decreases along the flow in $U$ and an orbit could only leave $K$ through the level set $\{\tilde\Phi=a\}$. By (1), $\phi_1(\bar\Qsq)$ is a compact subset of $\Qsq$. Every point of it is carried into the open set $\operatorname{int}K$ at some time, and hence, by the continuity of the flow, so is a neighborhood of it; by compactness and the forward invariance of $K$, there is some $b>0$ so that $\phi_b(\bar\Qsq)\subset K$. We define
                \[\Phi(\theta,\varphi):=e^{\beta b}\tilde\Phi(\phi_b(\theta,\varphi)),\]
            which is continuous on $\bar\Qsq$ and smooth on $\Qsq$, since $\phi_b$ is. Since $\phi_b$ commutes with $\phi_r$ and the forward orbit of $\phi_b(\theta,\varphi)$ stays in $K\subset U$, where (a) holds for $\tilde\Phi$, we have $\Phi(\phi_r(\theta,\varphi))=e^{-\beta r}\Phi(\theta,\varphi)$ for all $(\theta,\varphi)\in\bar\Qsq$ and $r>0$, and $\Phi=\tilde\Phi$ on $K$; this gives \cref{e-Phi-decay} and (b). Finally, $\Phi$ is continuous and positive on the compact set $\bar\Qsq\setminus\operatorname{int}K$ (it vanishes only at $(0,0)$, which lies in $\operatorname{int}K$), on which $\theta^2+\varphi^2\sim1$; together with $\Phi=\tilde\Phi$ on $K$, this gives (c), and the second bound in (a) follows from \cref{e-Phi-decay} and (c).
        \item 
            It suffices to show that $|\bm z|$ doubles within time $\lesssim|\bm z(t_0)|\ief{n-2}2$. This follows from the rescaling symmetry of \cref{e-ode} (if $(\bm z,\bm w)$ is a solution, then so is $(R\inv\bm z(R\ief{n-2}2\tau),R\ief n2\bm w(R\ief{n-2}2\tau))$, and \cref{e-model-hypotheses} is invariant under this rescaling), along with the fact that the doubling time is uniformly bounded when $|\bm z(t_0)|=1$: by (1) and \cref{l-continuity}, the first crossing time of $\{|z|=2\}$ is a continuous function of the data on the compact set $\{|z|=1,\ \Re(\bar zw)\geq0,\ |w|^2=\Re(z^n)\}$, hence bounded.
        \end{enumerate}
\end{proof}

The model system \cref{e-model} is illustrated in \cref{fig-model}.

\begin{figure}[htbp]
    \centering
    \includegraphics[width=.5\textwidth]{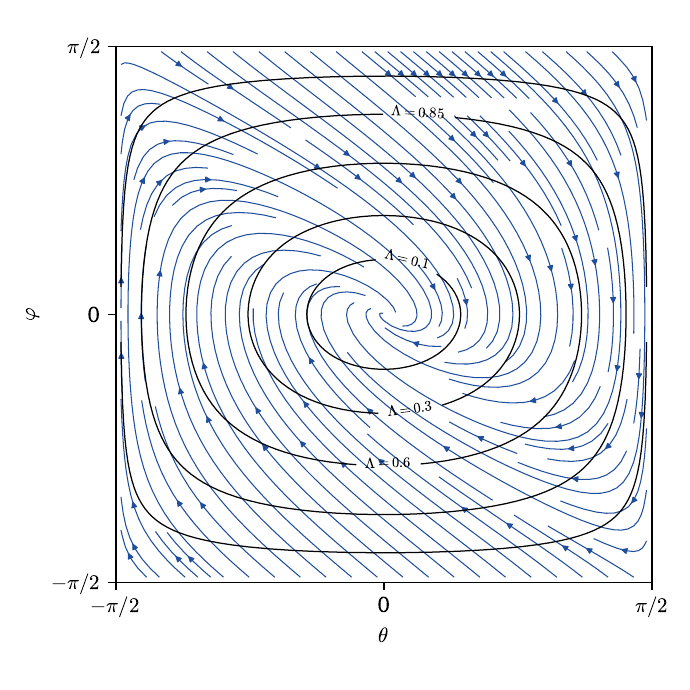}
    \caption{The phase portrait of the model system \cref{e-model} on the square $\Qsq$, for $n=3$ (so $\beta=\tfrac56$), together with some level sets of the function $\Lambda=1-\cos\theta\cos^2\varphi$ from the proof of \cref{c-model}. The origin is a spiral sink; $\Lambda$ is non-increasing along the flow (so that $\Qsq$ is forward invariant).}
    \label{fig-model}
\end{figure}

\section{Strong focusing and local completeness}\label{s-logic}
\def\PsiS{{\Psi_{\text{small}}}}
\def\PsiL{{\Psi_{\text{large}}}}

In this section we would like to prove \cref{main-lemma}. This will require studying the following ODE, where $\mathcal f(z)=1-e^{-z^n}$:
    \[\label{z-w-ode}
        \dot{\bm z}=\sqrt2\bm w,\qquad\dot{\bm w}=\frac1{\sqrt2}\overline{{\mathcal f}'(\bm z)}
        \]
We include the $\sqrt2$ factor to simplify the conservation of energy law:
    \[\label{e-conservation-of-energy-focusing-section}
        \frac{d}{dt}\p{|\bm w|^2-\Re\p{\mathcal f(\bm z)}}=0
        \]
We will also be interested in solutions to equations of the form
    \[\label{z-w-tilde-f-ode}
        \dot{\tilde{\bm z}}=\sqrt2\tilde{\bm w},\qquad\dot{\tilde{\bm w}}=\frac1{\sqrt2}\overline{\tilde{\mathcal f}'(t;\tilde{\bm z})}
        \]
where $\tilde{\mathcal f}\in F_\eps$ (see \cref{d-F-eps}; recall that $\tilde{\mathcal f}$ is allowed to depend on $t$). We will fix the constant $\eps$ later on in this section. For \cref{z-w-tilde-f-ode}, the energy is no longer conserved. Instead, by \cref{l-cov} (this part of \cref{l-cov} does not use \cref{e-hyp}; see also \cref{r-energy}),
    \[\label{e-energy-drift}
        \frac d{dt}\p{|\tilde{\bm w}|^2-\Re\p{\tilde{\mathcal f}(t;\tilde{\bm z})}}=-\Re\p{\tilde{\mathcal f}_t(t;\tilde{\bm z})},\qquad\text{where}\qquad\tilde{\mathcal f}_t(t;z)=-\Big(\sum_{j<n}\dot a_j(t)z^j\Big)e^{-z^n+\mathcal p(t;z)}.
        \]
In particular, if $|z^je^{-z^n+\mathcal p(t;z)}|\leq M_0$ for all $j<n$ along the solution on some time interval, then the energy changes by at most $M_0V(\mathcal p)$ on that time interval, where $V(\mathcal p)$ is the total variation \cref{e-total-variation}; we will only use this in the proof of \cref{main-lemma}. Alternatively, one can measure the energy with respect to the unperturbed function $\mathcal f$, as we do in \cref{d-S-region} below. By \cref{z-w-tilde-f-ode},
    \[\label{e-energy-drift-f}
        \frac d{dt}\p{|\tilde{\bm w}|^2-\Re\p{\mathcal f(\tilde{\bm z})}}=\sqrt2\,\Re\p{\tilde{\bm w}\,\p{\tilde{\mathcal f}'(t;\tilde{\bm z})-\mathcal f'(\tilde{\bm z})}}=\Re\p{\dot{\tilde{\bm z}}\,\p{\tilde{\mathcal f}'(t;\tilde{\bm z})-\mathcal f'(\tilde{\bm z})}},
        \]
which does not involve any derivatives in $t$: on any time interval, this energy changes by at most $\sup|\tilde{\mathcal f}'-\mathcal f'|$ (along the solution) times the length of the curve $\tilde{\bm z}$. This is how we control the energy in \cref{l-large-scale-focusing}, where $\tilde{\mathcal f}'-\mathcal f'$ decays superexponentially along the solution.

Generally, we will work with the ``zero-energy'' solutions to \cref{z-w-ode}, i.e. those for which $|\bm w|^2-\Re\p{\mathcal f(\bm z)}=0$. We will generalize from these simple solutions to not-necessarily zero-energy solutions to \cref{z-w-tilde-f-ode} in a single step (\cref{l-stable-focusing}), by using the continuous dependence of the solutions, and of their first crossing times, on the function $\mathcal f$ and on the initial data (\cref{ss-appendix-continuity}).

Note that the function $z\mapsto -e^{-z^n}$ is exponentially decaying whenever $\arg(z^n)\approx0$. Our goal will be to show that solutions to \cref{z-w-ode} with ``small'' initial data will stay close to radial trajectories. The radial trajectories will then enter the exponentially-decaying region, where we can hope for completeness.

More specifically, given some solution $(\bm z,\bm w)$ to \cref{z-w-ode}, let $\bm\theta=\arg(\bm z^n)$ and let $\bm\varphi=\arg(\bm w)-\arg(\bm z)$. Suppose that, whenever $(\bm z,\bm w)$ exists, at least one of the following is true:
    \begin{enumerate}
        \item $|\bm z|$ is small.
        \item $|\bm\theta|$ is small.
    \end{enumerate}
We will formalize and prove this later. The trajectories of any incomplete solutions must escape to $\infty$ in finite time, so if $(\bm z,\bm w)$ does not exist for all time, (1) will not be satisfied at some point. Therefore, we will have (2) for sufficiently large $\bm z$. When $|\bm\theta|$ is sufficiently small, we will automatically find that $\mathcal f$ is bounded uniformly along the trajectory. By \cref{e-conservation-of-energy-focusing-section}, we know that $|\bm w|\sim|\dot{\bm z}|$ must be bounded; hence, $|\bm z|$ grows at most linearly and cannot blow up in finite time. This behavior is illustrated in \cref{fig-focusing}.

\begin{figure}[htbp]
    \centering
    \includegraphics[width=\textwidth]{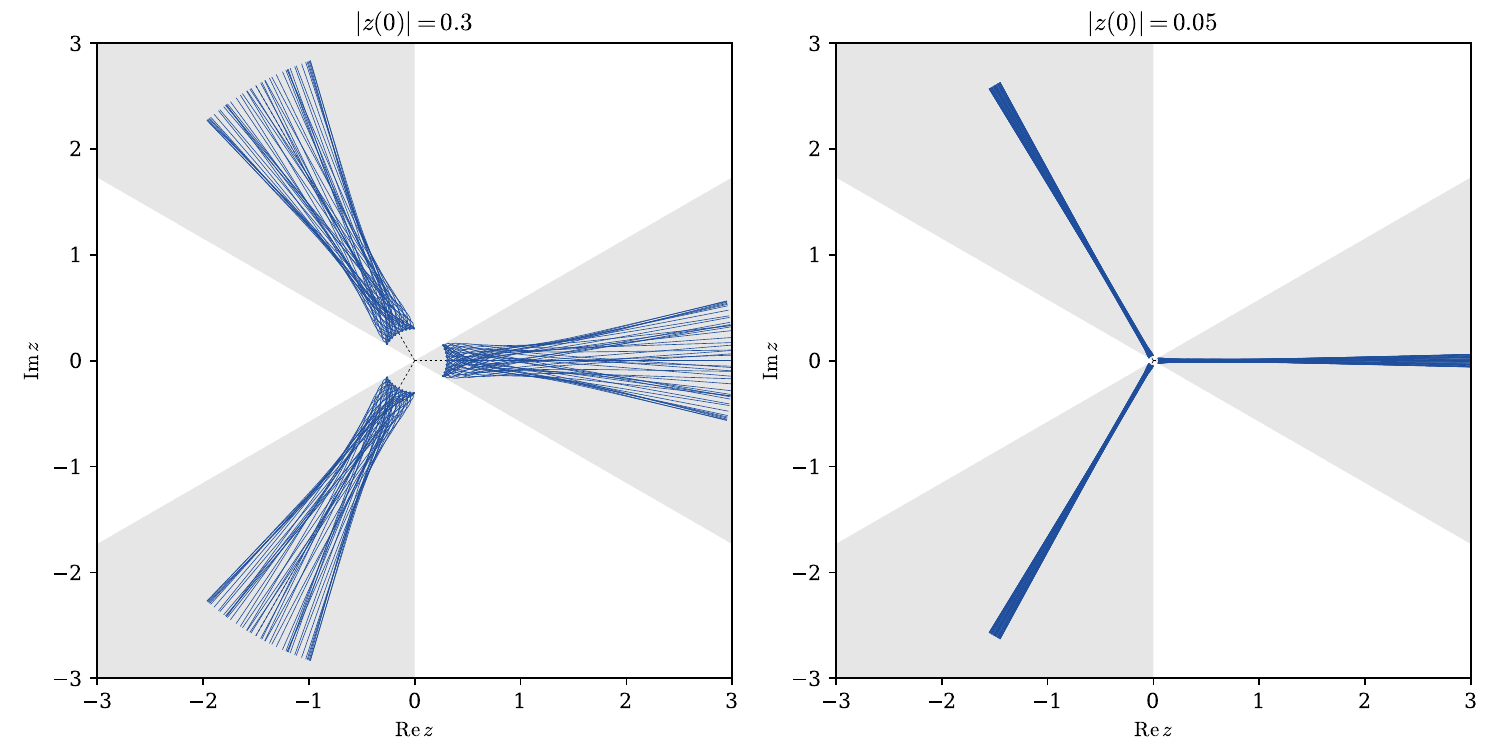}
    \caption{Numerically computed zero-energy solutions to \cref{z-w-ode} for $\mathcal f(z)=1-e^{-z^n}$ with $n=3$, plotted in the $z$-plane until they reach $|z|=3$. The initial data lie on the circle $|\bm z(0)|=0.3$ (left) and $|\bm z(0)|=0.05$ (right), with $\bm\theta(0)=\arg(\bm z(0)^n)$ ranging over $9$ equally spaced values in $[-1.5,1.5]$ and $\bm\varphi(0)\in\{0,\pm0.6,\pm1.2\}$. The shaded sectors are the regions $\Re(z^n)>0$, in which $e^{-z^n}$ decays, and the dashed rays are the radial directions $\arg(z^n)=0$. The trajectories oscillate about the rays while $|z|\lesssim1$ (this is the spiralling of the model system, \cref{fig-model}) and then continue along almost straight lines, since $\mathcal f'$ is negligible for $|z|\gg1$ inside a small aperture. The largest value of $|\bm\theta|$ at $|z|=3$ is approximately $0.57$ on the left and $0.06$ on the right; the ratio is consistent with the rate $s^{\beta n/2}=(1/6)^{5/4}\approx0.11$ coming from \cref{c-model}(3a), cf.\ \cref{Phi-bound-R-Ri0}. This figure is an illustration only and is not used in any proof.}
    \label{fig-focusing}
\end{figure}

The preceding paragraphs give a rough outline of how we intend to use the ``focusing'' phenomenon, i.e. that solutions with small initial data stay near radial lines. At some scales, the solutions to \cref{z-w-ode} will become more focused, and at some scales it will be the opposite. To adequately describe how these two forces interact, it will be convenient to fix a notation that describes the exact type of angular confinement that we are hoping to prove:

\begin{definition}\label{d-aperture}
    Let $\Psi:(-\frac\pi2,\frac\pi2)^2\to[0,\infty)$ be a function that satisfies the following properties:
        \begin{enumerate}
            \item $\Psi$ is continuous
            \item $\Psi(\theta,\varphi)\sim\theta^2+\varphi^2$
        \end{enumerate}
    Then, for any $A\in[0,\infty]$, we say that
        \[\mathcal A(\Psi,A):=\left\{(\theta,\varphi)\in(-\tfrac\pi2,\tfrac\pi2)^2:\Psi(\theta,\varphi)\leq A^2\right\}\]
    is an \textit{aperture}. We call $\Psi$ the \textit{aperture function} and $A$ the \textit{aperture size}.
\end{definition}

The function $\Psi$ will always be chosen to satisfy $\Psi(\theta,\varphi)\sim\theta^2+\varphi^2$, so $\Psi\ef12$ is a norm-like object, although in general we will not assume subadditivity. Together, $\Psi$ and $A$ represent the shape and size of the angular region that $\bm z$ and $\bm w$ are confined to. Note that, since $\Psi\sim\theta^2+\varphi^2$, the aperture $\mathcal A(\Psi,A)$ is a compact subset of the open square $(-\frac\pi2,\frac\pi2)^2$ when $A\ll1$, and $\mathcal A(\Psi,0)=\{(0,0)\}$.

\begin{definition}\label{d-S-region}
    Let $\mathcal A(\Psi,A)$ be an aperture, let $R>0$ be a radius, and let $\delta\geq0$. We denote by $\tilde{\mathcal S}_\Psi(R,A,\delta)$ the set of all $(z,w)\in\C^2$ that satisfy the following conditions:
        \begin{enumerate}
            \item $|z|=R$
            \item $\Re\p{\bar zw}\geq0$ (that is, either $w=0$, or $|\varphi|\leq\frac\pi2$ where $\varphi=\arg(w)-\arg(z)$)
            \item $\abs{|w|^2-\Re({\mathcal f}(z))}\leq\delta$, where $\mathcal f(z)=1-e^{-z^n}$
            \item If $A<\infty$: $w\neq0$ and $(\theta,\varphi)\in\mathcal A(\Psi,A)$, where $\theta=\arg(z^n)$ and $\varphi=\arg(w)-\arg(z)$. (If $A=\infty$, no further condition is imposed.)
            \end{enumerate}
    We write $\mathcal S_\Psi(R,A):=\tilde{\mathcal S}_\Psi(R,A,0)$ for the corresponding set of ``zero-energy'' data.
\end{definition}

Note that the energy in condition (3) is always computed with respect to the unperturbed function $\mathcal f$, even when we consider solutions to \cref{z-w-tilde-f-ode}; in this way, the sets $\tilde{\mathcal S}_\Psi(R,A,\delta)$ depend neither on $\tilde{\mathcal f}$ nor on $t$. The number $\delta$ is the \textit{energy tolerance}.

Let us emphasize the convention in (4). When $A\ll1$, the angles $(\theta,\varphi)$ lie in the \textit{open} square $(-\frac\pi2,\frac\pi2)^2$; in particular $\Re(z^n)>0$, $w\neq0$ and $\Re\p{\bar zw}>0$, so that (2) is redundant. When $A=0$, the set $\mathcal S_\Psi(R,0)$ consists of the $n$ radial data $(z,w)=(R\mu,\mathcal f(R)\ef12\mu)$, where $\mu^n=1$. When $A=\infty$, the sets $\tilde{\mathcal S}_\Psi(R,\infty,\delta)$ do not depend on $\Psi$, and they are \textit{not} contained in the sets defined through the open square: $w$ may vanish, $\varphi$ may equal $\pm\frac\pi2$, and $\theta$ may lie slightly outside of $(-\frac\pi2,\frac\pi2)$. The reason for this convention is that (1)--(3) are exactly the conditions that are satisfied when a trajectory which starts inside the disc $D_R$ reaches the circle $|z|=R$ for the first time (see the proof of \cref{main-lemma} in \cref{ss-main-lemma-proof}); we cannot assume more than this about the data of \cref{main-lemma}. In all cases, $\tilde{\mathcal S}_\Psi(R,A,\delta)\subset\tilde{\mathcal S}_\Psi(R,A_1,\delta_1)$ whenever $A\leq A_1$ and $\delta\leq\delta_1$.

We will use the following compactness properties. The set $\tilde{\mathcal S}_\Psi(R,\infty,\delta)$ is compact for every $\delta\geq0$, since (1)--(3) are closed conditions which bound $|z|$ and $|w|$. If $A\ll1$ (so that $|\theta|\leq1$ on $\mathcal A(\Psi,A)$, and $\mathcal A(\Psi,A)$ is a compact subset of the open square) and $|z|=R$, then
    \[\label{e-Re-f-lower-bound}
        \Re\p{\mathcal f(z)}=1-e^{-R^n\cos\theta}\cos\p{R^n\sin\theta}\geq1-e^{-R^n/2}=:m(R)>0;
        \]
hence, if $\delta\leq\frac12m(R)$, then $|w|^2\geq\frac12m(R)$ on $\tilde{\mathcal S}_\Psi(R,A,\delta)$, so that the condition $w\neq0$ in (4) may be replaced by the closed condition $|w|^2\geq\frac12m(R)$, and $\tilde{\mathcal S}_\Psi(R,A,\delta)$ is compact as well. Finally, for $A<\infty$ we have
    \[\label{e-S-intersection}
        \bigcap_{\delta>0}\tilde{\mathcal S}_\Psi(R,A+\delta,\delta)=\mathcal S_\Psi(R,A),\qquad\text{and}\qquad\bigcap_{\delta>0}\tilde{\mathcal S}_\Psi(R,\infty,\delta)=\mathcal S_\Psi(R,\infty).
        \]

\begin{definition}\label{defn-S-statement}
    Let $\mathcal A(\Psi,A)$ and $\mathcal A(\Psi',A')$ be two apertures, let $R<R'$ denote two radii, and let $\delta,\delta'\geq0$. Then we write
        \[\tilde{\mathcal S}_\Psi(R,A,\delta)\prec\tilde{\mathcal S}_{\Psi'}(R',A',\delta')\]
    if the following statement is true for some $T$. Let $\tilde{\mathcal f}\in F_\eps$, let $t_0\in\R$, and let $(\tilde{\bm z}(t),\tilde{\bm w}(t))$ denote any solution to \cref{z-w-tilde-f-ode} whose data at the time $t_0$ satisfies $(\tilde{\bm z}(t_0),\tilde{\bm w}(t_0))\in\tilde{\mathcal S}_\Psi(R,A,\delta)$. Then there exists some time $t_1\in(t_0,t_0+T)$ such that $(\tilde{\bm z}(t_1),\tilde{\bm w}(t_1))\in\tilde{\mathcal S}_{\Psi'}(R',A',\delta')$. Furthermore, $t_1$ is the smallest time greater than $t_0$ for which $|\tilde{\bm z}(t_1)|=R'$.

    This statement depends on $\eps$. For the sets of zero-energy data, we use the following convention: the statement $\mathcal S_\Psi(R,A)\prec\mathcal S_{\Psi'}(R',A')$ always refers to the case $\eps=0$, i.e. to the solutions of \cref{z-w-ode}. (In this case, condition (3) of \cref{d-S-region} holds automatically at the time $t_1$, by \cref{e-conservation-of-energy-focusing-section}, and the initial time plays no role.)

    The above applies when $R<R'$; when $R=R'$, we define $\prec$ to mean inclusion, i.e. $\tilde{\mathcal S}_\Psi(R,A,\delta)\prec\tilde{\mathcal S}_{\Psi'}(R,A',\delta')$ means $\tilde{\mathcal S}_\Psi(R,A,\delta)\subset\tilde{\mathcal S}_{\Psi'}(R,A',\delta')$ (this corresponds to $t_1=t_0$).
\end{definition}

Note that $\prec$ is transitive (for a fixed value of $\eps$). Also, all statements involving $\tilde{\mathcal S}$ with a positive energy tolerance will depend on $\eps$; if such a statement holds for some value of $\eps$, then it also holds for all smaller values, because $F_\eps$ increases with $\eps$. We will not fix the value of $\eps$ until later.

The following theorem formalizes the type of angular confinement we will need:
    \begin{theorem}[The focusing theorem]\label{t-focusing}
        For any aperture functions $\Psi,\Psi'$, if
            \[\label{e-focusing-requirement}
                \eps\lll\kappa\lll R\lll A'\ll1,\quad\quad {R'}\inv\lll A',\quad\text{and}\quad A=\infty\]
        then
            \[\tilde{\mathcal S}_\Psi(R,A,\kappa)\prec\tilde{\mathcal S}_{\Psi'}(R',A',1)\]
        Moreover, if $\tilde{\mathcal f}\in F_\eps$, if $(\tilde{\bm z},\tilde{\bm w})$ and $t_1$ are as in \cref{defn-S-statement}, and if $[t_0,T_{\max})$ is the maximal forward interval of existence of the solution, then $|\tilde{\bm z}|$ is strictly increasing on $[t_1,T_{\max})$ and $(\tilde{\bm\theta}(t),\tilde{\bm\varphi}(t))\in\mathcal A(\Psi',A')$ for all $t\in[t_1,T_{\max})$.
    \end{theorem}
(Since $A=\infty$, the set $\tilde{\mathcal S}_\Psi(R,A,\kappa)$ does not actually depend on $\Psi$. The value $1$ of the energy tolerance at the radius $R'$ is of no importance; what we will use is the last assertion of the theorem.) Note the order in which the constants are chosen (recall the conventions about chains of $\lll$ in \cref{ss-notation}): first $A'$, then $R$ and a lower bound for $R'$, then $\kappa$, and finally $\eps$; in particular, $\eps$, $\kappa$ and $R$ do not depend on $R'$, so the same constants work for every sufficiently large $R'$. In words, \cref{t-focusing} says more or less the following: ``Given any aperture $\mathcal A$, any trajectory of \cref{z-w-tilde-f-ode} will lie inside $\mathcal A$ when $|z|$ is sufficiently large, as long as $|\mathcal f-\tilde{\mathcal f}|$ is small and the initial data is small.''

\subsection{The proof of \texorpdfstring{\cref{main-lemma}}{Lemma \ref{main-lemma}} using \texorpdfstring{\cref{t-focusing}}{Theorem \ref{t-focusing}}}\label{ss-main-lemma-proof}
Let $\Psi'(\theta,\varphi)=\theta^2+\varphi^2$ (the choice of $\Psi$ is irrelevant since $A=\infty$) and fix $A'\in(0,\frac1{10}]$ small enough for \cref{t-focusing} to apply. By \cref{t-focusing}, there are constants $\eps_*,\kappa>0$, $R_*\in(0,1]$ and $R'_*\geq3n$ such that the conclusion of \cref{t-focusing} holds with $R=R_*$, for every $R'\geq R'_*$ and every $\eps\leq\eps_*$. Let $\tilde{\mathcal f}\in F_\eps$ for some $\eps\leq1$. Since $\deg\mathcal p<n$ and the coefficients of $\mathcal p$ have modulus $\leq1$, for $|z|\geq R'_*\geq3n$ we have $|\mathcal p(t;z)|\leq n|z|^{n-1}\leq\frac13|z|^n$, and hence, uniformly in $t$,
    \[\label{e-main-lemma-decay}
        \Re\p{z^n-\mathcal p(t;z)}\geq\p{\cos\tfrac1{10}-\tfrac13}|z|^n\geq\tfrac12|z|^n\qquad\text{whenever }|z|\geq R'_*\text{ and }|\arg(z^n)|\leq\tfrac1{10}.
        \]
Next, for $|z|\leq1$ we have $|{-z^n}+\mathcal p(t;z)|\leq|z|^n+n\eps\leq n+1$ and $|e^{-z^n}|\leq e$. Using $|e^u-1|\leq|u|e^{|u|}$ and \cref{e-energy-drift}, we obtain for $|z|\leq1$, uniformly in $t$,
    \[\label{e-main-lemma-small-z}
        |\tilde{\mathcal f}(t;z)|\leq e^{n+1}(|z|^n+n\eps),\qquad|\tilde{\mathcal f}(t;z)-\mathcal f(z)|\leq e^{n+1}n\eps,\qquad|\tilde{\mathcal f}_t(t;z)|\leq e^{n+1}\sum_{j<n}|\dot a_j(t)|.
        \]
(The last bound follows from the formula for $\tilde{\mathcal f}_t$ in \cref{e-energy-drift}.) We now choose the constants $R$ and $\eps$ of \cref{main-lemma} as follows: $R<R_*$ and $\eps\leq\min(\eps_*,1)$ are so small that
    \[\label{e-main-lemma-energy}
        \tfrac12R^n+e^{n+1}\p{R^n+(2n+1)\eps}\leq\kappa.
        \]

Let $\tilde{\mathcal f}\in F_\eps$ with $V(\mathcal p)\leq\eps$, let $t_*\in\R$, and let $(\bm z,\bm w)$ denote any solution to \cref{z-w-tilde-f-ode} with $|\bm z(t_*)|\leq R$ and $|\dot{\bm z}(t_*)|^2=2|\bm w(t_*)|^2\leq R^n$. Since $(\bm z(-t),-\bm w(-t))$ is a solution for the function $\tilde{\mathcal f}(-t;z)$, which also belongs to $F_\eps$ and has the same total variation, and since it has the same properties at the time $-t_*$, it suffices to show that the solution exists for all $t\geq t_*$. Let $[t_*,T_{\max})$ be its maximal forward interval of existence, and let $\bm\kappa(t):=|\bm w(t)|^2-\Re\p{\tilde{\mathcal f}(t;\bm z(t))}$ denote its energy. By \cref{e-main-lemma-small-z}, we have $|\bm\kappa(t_*)|\leq\frac12R^n+e^{n+1}(R^n+n\eps)$. By \cref{e-energy-drift,e-main-lemma-small-z} and the hypothesis $V(\mathcal p)\leq\eps$, if $|\bm z|\leq R_*\leq1$ on $[t_*,\tau]$, then
    \[\label{e-main-lemma-energy-drift}
        |\bm\kappa(t)|\leq|\bm\kappa(t_*)|+e^{n+1}V(\mathcal p)\leq\tfrac12R^n+e^{n+1}\p{R^n+(n+1)\eps}\qquad\text{for all }t\in[t_*,\tau].
        \]
In particular, $|\bm w|^2=\bm\kappa+\Re\p{\tilde{\mathcal f}(t;\bm z)}$ is bounded as long as $|\bm z|\leq R_*$. So, if $|\bm z(t)|<R_*$ for all $t\in[t_*,T_{\max})$, then $T_{\max}=\infty$ and we are done. Otherwise, let $t_0>t_*$ be the first time with $|\bm z(t_0)|=R_*$. Since $|\bm z|<R_*$ on $[t_*,t_0)$, we have $\frac d{dt}|\bm z|^2(t_0)=2\sqrt2\Re\p{\bar{\bm z}\bm w}(t_0)\geq0$. Moreover, by \cref{e-main-lemma-energy-drift,e-main-lemma-small-z,e-main-lemma-energy},
    \[\abs{|\bm w(t_0)|^2-\Re\p{\mathcal f(\bm z(t_0))}}\leq|\bm\kappa(t_0)|+\abs{\tilde{\mathcal f}(t_0;\bm z(t_0))-\mathcal f(\bm z(t_0))}\leq\tfrac12R^n+e^{n+1}\p{R^n+(2n+1)\eps}\leq\kappa.\]
This says precisely that $(\bm z(t_0),\bm w(t_0))\in\tilde{\mathcal S}_\Psi(R_*,\infty,\kappa)$; recall that for $A=\infty$, \cref{d-S-region} imposes no condition on the angles. By \cref{t-focusing} with $R'=R'_*$, there is a first time $t_1>t_0$ with $|\bm z(t_1)|=R'_*$, and for all $t\in[t_1,T_{\max})$ we have $|\bm z(t)|\geq R'_*$ and $|\bm\theta(t)|\leq A'\leq\frac1{10}$. By \cref{e-main-lemma-decay},
    \[
        |\tilde{\mathcal f}'(t;\bm z(t))|=\abs{n\bm z^{n-1}-\mathcal p'(t;\bm z)}e^{-\Re\p{\bm z^n-\mathcal p(t;\bm z)}}\lesssim|\bm z|^{n-1}e^{-|\bm z|^n/2}\lesssim1\qquad\text{for }t\in[t_1,T_{\max}),
        \]
while on $[t_*,t_1]$ the trajectory stays in the compact set $\{|z|\leq R'_*\}$, on which $\tilde{\mathcal f}'$ is bounded uniformly in $t$ (because the coefficients of $\mathcal p$ are bounded by $1$). Therefore $|\ddot{\bm z}|=|\tilde{\mathcal f}'(t;\bm z)|$ is uniformly bounded on $[t_*,T_{\max})$; this implies that $|\bm z(t)|$ grows at most quadratically in $t$, and hence cannot blow up in finite time. So $T_{\max}=\infty$, which proves \cref{main-lemma}.

\subsection{Proof of the focusing theorem using the focusing lemmas.}\label{ss-focusing-theorem-logic}

In order to prove \cref{t-focusing}, we will need to be able to understand the $\mathcal f$ and $\tilde{\mathcal f}$ at many different scales. Sometimes, the choice of one scale will depend on others, in a complicated way (e.g. \cref{chain0}). We have broken up the problem into 5 lemmas, which will help us understand $\mathcal f$ or $\tilde{\mathcal f}$ in a different way at different scales.

More specifically, we will let $R_1=R$, $A_1=A$, $A_6=A'$, and $R_6=R'$, and let $R_2,\cdots,R_5$ and $A_2,\cdots,A_5$ be defined later. Suppose that we would like to prove something akin to the following:
    \[\mathcal S_\Psi(R_1,A_1)\prec\dots\prec\mathcal S_\Psi(R_6,A_6)\]
This would be sufficient to show that $\mathcal S_\Psi(R_1,A_1)\prec\mathcal S_\Psi(R_6,A_6)$. However, in order to prove anything involving $\tilde{\mathcal S}$, we will need the following lemma (or at least something similar):

\begin{lemma}[Stable focusing]\label{l-stable-focusing}
    Let $\Psi,\Psi'$ be aperture functions, let $R<R'$ and let $\kappa'>0$. Suppose that $0\leq A'<\tilde A'\ll1$, and that either $0\leq A\ll1$ or $A=\infty$. If
        \[\label{e-stable-focusing}
            \eps,\kappa,\alpha\lll A,A',\tilde A',R,R',\kappa'
            \]
    (where the values $\eps=0$, $\kappa=0$ and $\alpha=0$ are allowed), then
        \[\mathcal S_\Psi(R,A)\prec\mathcal S_{\Psi'}(R',A')\]
    implies
        \[\tilde{\mathcal S}_\Psi(R,A+\alpha,\kappa)\prec\tilde{\mathcal S}_{\Psi'}(R',\tilde A',\kappa')\]
\end{lemma}

Throughout this section, we often use the letters $R$, $A$, $R'$, and $A'$ to reference the radii and apertures for the lemma at hand. Thus, the meaning of $R,R',A,A'$ in \cref{l-stable-focusing} are not the same as they are in \cref{t-focusing}.

\Cref{l-stable-focusing} shows us how we can reduce to the case of zero-energy solutions to \cref{z-w-ode} without much loss (we may take $\tilde A'=2A'$). The parameter $\alpha$, which allows us to enlarge the aperture of the initial data slightly, will only be used once, in the proof of \cref{l-medium-scale-focusing}; everywhere else, $\alpha=0$.

We will also need the following lemma for small scales:

\begin{lemma}[Small-scale focusing]\label{l-small-scale-focusing}
    For any $\Psi$, if
        \[\label{e-small-scale-requirement}
            R\lll R'\lll A'<\infty=A
            \]
    then
        \[\mathcal S_\Psi(R,A)\prec\mathcal S_\Psi(R',A')\]
\end{lemma}

\cref{l-small-scale-focusing} gives us the most possible freedom over the amount of focusing, and the least amount of freedom in terms of the scales $R,R'$; this is manifested in the fact that we may choose $A$ as large as we want: that is, there are no focusing requirements for the initial data. In spite of this, we may prove that $(\bm\theta,\bm\varphi)$ will become ``arbitrarily well-focused'' for sufficiently small $R$. \Cref{l-small-scale-focusing} is in fact the only one of these 5 lemmas that will allow $A'<A$, so all other lemmas result in no focusing gain.

The following is the lemma in the opposite regime, i.e. for $R,R'$ large:
\begin{lemma}[Large scale focusing]\label{l-large-scale-focusing}
    There exist an aperture function $\PsiL$ and a constant $\kappa_0\in(0,1]$ (depending only on $n$) so that, if
        \[\label{e-large-scale-requirement}
            {R'}\inv\leq R\inv\lll A\lll A'\ll1,\qquad\eps\leq\kappa_0\qquad\text{and}\qquad\kappa\leq\kappa_0,
            \]
    then
        \[\tilde{\mathcal S}_{\PsiL}(R,A,\kappa)\prec\tilde{\mathcal S}_{\PsiL}(R',A',1)\]
    is true. Moreover, if $\tilde{\mathcal f}\in F_\eps$, if $(\tilde{\bm z},\tilde{\bm w})$ is any solution to \cref{z-w-tilde-f-ode} with $(\tilde{\bm z}(t_0),\tilde{\bm w}(t_0))\in\tilde{\mathcal S}_\PsiL(R,A,\kappa)$, and if $[t_0,T_{\max})$ is its maximal forward interval of existence, then $|\tilde{\bm z}|$ is strictly increasing on $[t_0,T_{\max})$ and $(\tilde{\bm\theta}(t),\tilde{\bm\varphi}(t))\in\mathcal A(\PsiL,A')$ for all $t\in[t_0,T_{\max})$.
\end{lemma}

Superficially, it appears that these two lemmas will suffice, when combined with the continuity of the flow:

\begin{lemma}[Medium-scale focusing]\label{l-medium-scale-focusing}
    For any aperture function $\Psi$, if
        \[\label{e-medium-scale-requirement}
            A\lll A',R,R'\quad\text{and}\quad R\leq R'
            \]
    then
        \[\label{e-medium-scale-conclusion}
            \mathcal S_\Psi(R,A)\prec\mathcal S_\Psi(R',A')
            \]
\end{lemma}
These four lemmas are not sufficient to prove \cref{t-focusing}. The reason why is because \cref{l-small-scale-focusing} requires the radii to be chosen last, and \cref{l-medium-scale-focusing} requires the apertures to be chosen after the radii. Therefore, we will need the following ``medium small'' scale focusing lemma:
\begin{lemma}[Medium-small scale focusing]\label{l-medium-small-scale-focusing}
    There exists some aperture function $\PsiS$ so that, if
        \[\label{e-medium-small-scale-requirement}
            R\leq R'\ll1\quad\text{and}\quad A=A'\ll1
            \]
    then
        \[\label{e-medium-small-scale-conclusion}
            \mathcal S_\PsiS(R,A)\prec\mathcal S_\PsiS(R',A')
            \]
\end{lemma}

We are now ready to prove \cref{t-focusing} using the focusing lemmas above.

Let $\kappa_0$ be the constant of \cref{l-large-scale-focusing}, which depends only on $n$, and let $\eps$, $\kappa$, the radii $R_1,\cdots,R_6$ and the aperture sizes $A_1,\cdots,A_6$ and $\tilde A_5$ be chosen as in \cref{chain0} below (the aperture functions will be $\PsiS$ and $\PsiL$, as indicated in \cref{e-S-chain}; the thresholds implicit in \cref{chain0} are allowed to depend on $\kappa_0$). The first row is meant to serve as a definition, with each row after that being a property of some subset of the quantities at hand; it is easier to verify the hypotheses of the focusing lemmas if we can visually separate which quantities we should focus on.

        \begin{center}
    \def\nl#1{&#1&}
    \newcommand{\therowcntr}{\thesection.\arabic{equation}}
    \newcolumntype{L}{>{\hspace{0.7cm}(\refstepcounter{equation}\therowcntr)}r}
    \newcolumntype{N}{>{\refstepcounter{equation}\therowcntr}c}
    \newcolumntype{C}{>{$\hspace{-4pt}}l<{\hspace{-4pt}$}}
    \vspace{1em}
    \resizebox{\textwidth}{!}{\begin{tabular}{%
        C   CC  CC  CC  CC  CC  CC  CC  CC  CC  CC  CC  CC  CC  CC  CC  CC CC  L}
        \eps
            \nl\lll \kappa
            \nl\lll R_1
            \nl\lll R_2
            \nl\lll A_2
            \nl= A_3
            \nl\lll A_4
            \nl\lll R_3
            \nl= R_4
            \nl\lll R_5\inv
            \nl\lll A_5
            \nl\lll \tilde A_5
            \nl\lll A_6
            \nl\ll 1
            \nl\ll R_5
            \nl\leq R_6
            \nl< \infty
            \nl= A_1
                &\label{chain0}
                \\\\
        
            \nl~
            \nl~ R_1
            \nl\lll R_2
            \nl\lll A_2
            \nl~
            \nl~
            \nl~
            \nl~
            \nl~
            \nl~
            \nl~
            \nl~
            \nl\ll 1
            \nl~
            \nl~
            \nl~ \infty
            \nl= A_1
                &\label{chain1}
                \\\\
        
            \nl~
            \nl~
            \nl~ R_2
            \nl\leq A_2
            \nl= A_3
            \nl~
            \nl\leq R_3
            \nl~
            \nl~
            \nl~
            \nl~
            \nl~
            \nl\ll 1
            \nl~
            \nl~
            \nl~
            \nl~
                &\label{chain2}
                \\\\
        
            \nl~
            \nl~
            \nl~
            \nl~
            \nl~ A_3
            \nl\ll A_4
            \nl\leq R_3
            \nl= R_4
            \nl~
            \nl~
            \nl~
            \nl~
            \nl~
            \nl~
            \nl~
            \nl~
            \nl~
                &\label{chain3}
                \\\\
        
            \nl~
            \nl~
            \nl~
            \nl~
            \nl~
            \nl~ A_4
            \nl~
            \nl\lll R_4
            \nl\leq R_5\inv
            \nl\leq A_5
            \nl~
            \nl~
            \nl\leq 1
            \nl\leq R_5
            \nl~
            \nl~
            \nl~
                &\label{chain4}
                \\\\
        \eps
            \nl\lll \kappa
            \nl\lll R_1
            \nl~
            \nl~
            \nl~
            \nl~
            \nl~
            \nl~
            \nl\lll R_5\inv
            \nl\lll A_5
            \nl\lll \tilde A_5
            \nl~
            \nl~
            \nl~
            \nl~
            \nl~ \infty
            \nl= A_1
                &\label{chain5}
                \\\\
        \eps
            \nl\lll \kappa
            \nl~
            \nl~
            \nl~
            \nl~
            \nl~
            \nl~
            \nl~
            \nl\lll R_5\inv
            \nl~
            \nl\lll \tilde A_5
            \nl\lll A_6
            \nl\ll 1
            \nl\ll R_5
            \nl\leq R_6
            \nl~
            \nl~
                &\label{chain6}
                \\
    \end{tabular}}
    \vspace{1em}
    \end{center}

Note that, according to \cref{chain0} and the conventions of \cref{ss-notation}, the constants are chosen in the order $A_6$, $\tilde A_5$, $A_5$, $R_5$, $R_4=R_3$, $A_4$, $A_3=A_2$, $R_2$, $R_1$, $\kappa$, $\eps$. In particular, none of them depends on $R_6=R'$: the only requirement on $R'$ is $R'\geq R_5$, where $R_5$ depends only on $A_6=A'$. This is the hypothesis ${R'}\inv\lll A'$ of \cref{t-focusing}, and the first row \cref{chain0} encodes the hypothesis $\eps\lll\kappa\lll R\lll A'\ll1$. A schematic picture of the resulting aperture sizes is given in \cref{fig-apertures}.

\begin{figure}[htbp]
    \centering
    \includegraphics[width=\textwidth]{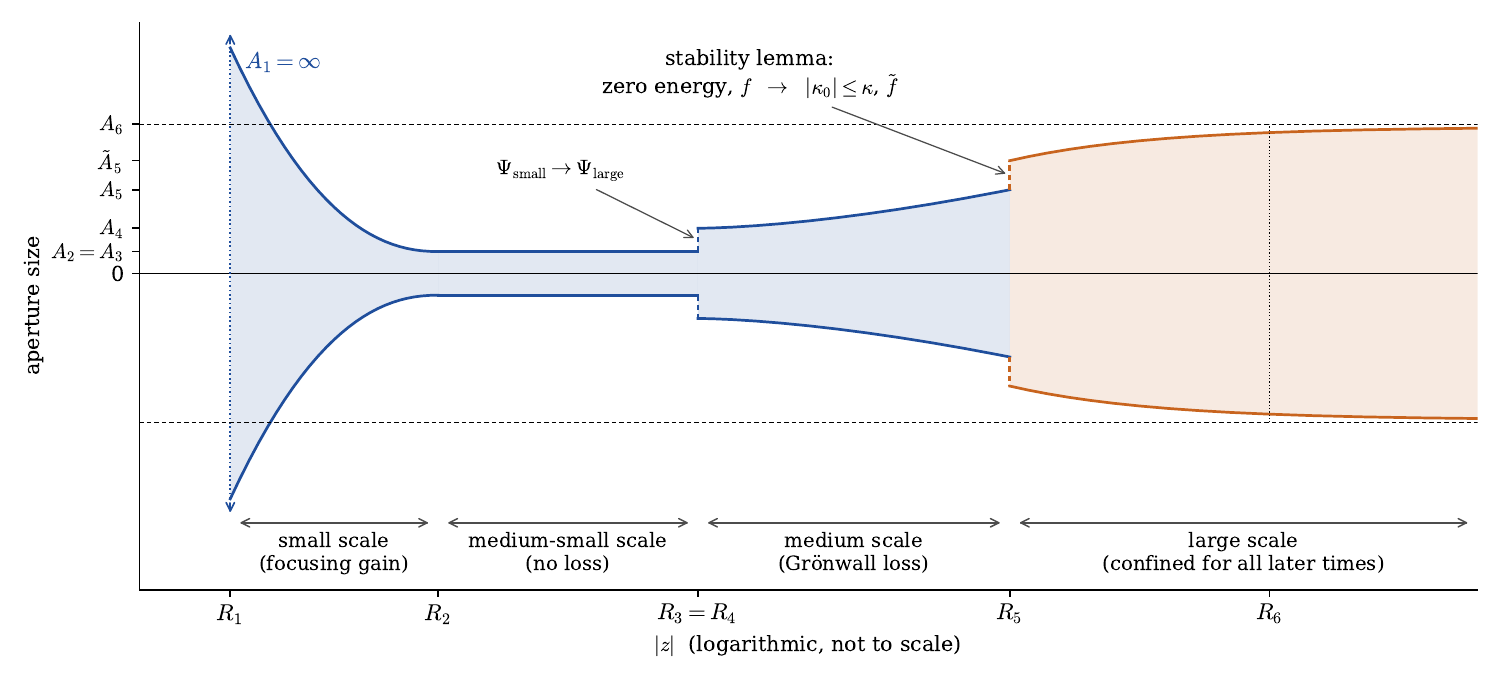}
    \caption{A schematic picture (not to scale) of the aperture sizes in \cref{chain0,e-S-chain}. Focusing is only gained at small scales (\cref{l-small-scale-focusing}); it is preserved at medium-small scales (\cref{l-medium-small-scale-focusing}), and at medium scales \cref{l-medium-scale-focusing} only controls the loss. At $R_3=R_4$ the aperture function changes from $\PsiS$ to $\PsiL$, so the jump from $A_3$ to $A_4$ is in part a change of units. Blue: zero-energy solutions to \cref{z-w-ode}. Orange: all solutions to \cref{z-w-tilde-f-ode}, for all $\tilde{\mathcal f}\in F_\eps$, whose energy at the radius $R_1$ is at most $\kappa$ in modulus. By \cref{l-stable-focusing}, the statement represented by the whole blue part also holds for \cref{z-w-tilde-f-ode}, with $A_5$ replaced by $\tilde A_5$. By \cref{l-large-scale-focusing}, the trajectory then stays inside the aperture of size $A_6$ for all larger radii, not only up to $R_6$.}
    \label{fig-apertures}
\end{figure}

The other rows of the table are there for reference only and they will help us prove the following chain, in which the two arrows $\Downarrow$ indicate that the bottom left relation is deduced from the (whole) top row by means of \cref{l-stable-focusing}:

\begin{center}
    \def\nl#1{&#1&}
    \newcounter{rowcntr}[table]
    \renewcommand{\therowcntr}{\thesection.\arabic{equation}}
    \newcolumntype{N}{>{\refstepcounter{equation}\therowcntr}c}
    \vspace{1em}
    \newcolumntype{A}{>{$\hspace{-4pt}}c<{\hspace{-4pt}$}}
    \newcommand{\precn}{\Downarrow}
    \def\SS{\mathcal S_{\PsiS}}
    \def\SL{\mathcal S_{\PsiL}}
    \def\tprec{&\prec&}
    \resizebox{\textwidth}{!}{\begin{tabular}{AAAAAAAAAAAr}
        \SS(R_1,A_1)\tprec\SS(R_2,A_2)\tprec\SS(R_3,A_3)\tprec\SL(R_4,A_4)\tprec\SL(R_5,A_5)&&&\\
        \precn&&&&&&&&\precn&&&(\refstepcounter{equation}\therowcntr)\label{e-S-chain}\\
        \tilde{\mathcal S}_{\PsiS}(R_1,A_1,\kappa)&&&&\prec&&&&\tilde{\mathcal S}_\PsiL(R_5,\tilde A_5,\kappa_0)\tprec\tilde{\mathcal S}_\PsiL(R_6,A_6,1)&
        \end{tabular}}
    \vspace{1em}
    \end{center}
    
From \cref{chain0,chain1,chain2,chain3,chain4,chain5,chain6}, one can simply prove each part of \cref{e-S-chain} by reading through the appropriate row among \cref{chain1,chain2,chain3,chain4,chain5,chain6}, and using the appropriate focusing lemma. We perform this below, showing each $\prec$ from left to right and top to bottom:
    \begin{enumerate}
        \item To show $\mathcal S_{\PsiS}(R_1,A_1)\prec\mathcal S_{\PsiS}(R_2,A_2)$, apply \cref{l-small-scale-focusing} with $R=R_1$, $R'=R_2$, $A=A_1$, and $A'=A_2$. The hypothesis \cref{e-small-scale-requirement} is precisely \cref{chain1}.
        \item To show $\mathcal S_{\PsiS}(R_2,A_2)\prec\mathcal S_{\PsiS}(R_3,A_3)$, apply \cref{l-medium-small-scale-focusing} with $R=R_2$, $R'=R_3$, $A=A_2$, and $A'=A_3$. The hypothesis \cref{e-medium-small-scale-requirement} is precisely \cref{chain2}.
        \item To show $\mathcal S_{\PsiS}(R_3,A_3)\prec\mathcal S_{\PsiL}(R_4,A_4)$, note that by \cref{chain3} we have $R_3=R_4$, so by \cref{defn-S-statement} we need to show that $\mathcal S_{\PsiS}(R_3,A_3)\subset\mathcal S_{\PsiL}(R_3,A_4)$, i.e. that $\PsiS(\theta,\varphi)\leq A_3^2$ implies $\PsiL(\theta,\varphi)\leq A_4^2$. This follows from \cref{d-aperture} and $A_3\ll A_4$.
        \item To show $\mathcal S_{\PsiL}(R_4,A_4)\prec\mathcal S_{\PsiL}(R_5,A_5)$, apply \cref{l-medium-scale-focusing} with $R=R_4$, $R'=R_5$, $A=A_4$, and $A'=A_5$. The hypotheses \cref{e-medium-scale-requirement} follow from \cref{chain4}.
        \item To show $\tilde{\mathcal S}_\PsiS(R_1,A_1,\kappa)\prec\tilde{\mathcal S}_\PsiL(R_5,\tilde A_5,\kappa_0)$, apply \cref{l-stable-focusing} with $\Psi=\PsiS$, $\Psi'=\PsiL$, $A=A_1=\infty$, $\alpha=0$, $R=R_1$, $A'=A_5$, $R'=R_5$, $\tilde A'=\tilde A_5$ and $\kappa'=\kappa_0$. Its main hypothesis, $\mathcal S_\PsiS(R_1,A_1)\prec\mathcal S_\PsiL(R_5,A_5)$, is the top row of \cref{e-S-chain}, which we have already shown (recall that $\prec$ is transitive). The hypotheses \cref{e-stable-focusing} follow from \cref{chain5}, and $A_5<\tilde A_5\ll1$ follows from \cref{chain0}.
        \item To show $\tilde{\mathcal S}_{\PsiL}(R_5,\tilde A_5,\kappa_0)\prec\tilde{\mathcal S}_{\PsiL}(R_6,A_6,1)$, apply \cref{l-large-scale-focusing} with $R=R_5$, $R'=R_6$, $A=\tilde A_5$, $A'=A_6$ and $\kappa=\kappa_0$. The hypotheses \cref{e-large-scale-requirement} follow from \cref{chain6} (which in particular gives $\eps\leq\kappa_0$).
    \end{enumerate}

Therefore, we conclude
    \[\label{e-prec-finally}
        \tilde{\mathcal S}_\PsiS(R_1,A_1,\kappa)\prec\tilde{\mathcal S}_\PsiL(R_6,A_6,1)\]
Moreover, by the last assertion of \cref{l-large-scale-focusing} (applied as in (6)), from the first time at which $|\tilde{\bm z}|=R_5$ onwards, $|\tilde{\bm z}|$ is strictly increasing and $(\tilde{\bm\theta},\tilde{\bm\varphi})$ stays in $\mathcal A(\PsiL,A_6)$ for as long as the solution exists; since $R_6\geq R_5$, this holds in particular from the first time at which $|\tilde{\bm z}|=R_6$ onwards.

Finally, since any aperture function $\Psi'$ satisfies $\Psi'\lesssim\PsiL$ by \cref{d-aperture}, we have $\mathcal A(\PsiL,A_6)\subset\mathcal A(\Psi',CA_6)$ and hence $\tilde{\mathcal S}_\PsiL(R_6,A_6,1)\subset\tilde{\mathcal S}_{\Psi'}(R_6,CA_6,1)$ for some constant $C$ depending on $\Psi'$; replacing $A'$ by $A'/C$ throughout gives the statement of \cref{t-focusing} for a general $\Psi'$, including its last assertion. This completes the proof.
\section{Proofs of the focusing lemmas}\label{s-focusing-lemmas}
In this section, we seek to prove the 5 ``focusing'' lemmas of the previous section. Three of them (\cref{l-stable-focusing,l-small-scale-focusing,l-medium-scale-focusing}) are soft: they follow from the continuity of the flow and of the first crossing (\cref{l-continuity}), together with the model case \cref{c-model} and the scaling symmetry of the equation. To prove the other two, we will use the change of variables of \cref{l-cov} (with $\mathcal f$ independent of $t$, so that $\bm f_t=0$ and $\bm\kappa$ is constant, except at the end of \cref{ss-large-scale}). Throughout, $\bm\theta=\arg(\bm z^n)$ and $\bm\varphi=\arg\bm w-\arg\bm z\in(-\tfrac\pi2,\tfrac\pi2)$ are as in \cref{l-cov}; recall that condition (4) of \cref{d-S-region} forces $|\bm\theta|,|\bm\varphi|<\frac\pi2$ (in particular $\bm w\neq0$) for data in $\mathcal S_\Psi(R,A)$ or $\tilde{\mathcal S}_\Psi(R,A,\delta)$ when $A<\infty$, whereas for $A=\infty$ only conditions (1)--(3) of \cref{d-S-region} are imposed (so that, for instance, $\bm w=0$ is allowed). The case $A=\infty$ only occurs in \cref{ss-stability,ss-small-scale}. We will also use repeatedly that, for $\mathcal f(z)=1-e^{-z^n}$,
    \[\label{e-Re-f-positive}
        \Re\p{\mathcal f(z)}=1-e^{-\Re(z^n)}\cos\p{\Im(z^n)}>0\qquad\text{whenever }\Re(z^n)>0,
        \]
since then $e^{-\Re(z^n)}<1$.

\subsection{Stability (\texorpdfstring{\cref{l-stable-focusing}}{Lemma \ref{l-stable-focusing}})}\label{ss-stability}
First, recall the statement of \cref{l-stable-focusing}:

\begin{restatedlemma}{l-stable-focusing}[Stable focusing]
    Let $\Psi,\Psi'$ be aperture functions, let $R<R'$ and let $\kappa'>0$. Suppose that $0\leq A'<\tilde A'\ll1$, and that either $0\leq A\ll1$ or $A=\infty$. If
        \[
            \eps,\kappa,\alpha\lll A,A',\tilde A',R,R',\kappa'
            \]
    (where the values $\eps=0$, $\kappa=0$ and $\alpha=0$ are allowed), then
        \[\label{e-stable-focusing-h}
            \mathcal S_\Psi(R,A)\prec\mathcal S_{\Psi'}(R',A')
            \]
    implies
        \[\label{e-stable-conclusion}
            \tilde{\mathcal S}_\Psi(R,A+\alpha,\kappa)\prec\tilde{\mathcal S}_{\Psi'}(R',\tilde A',\kappa')
            \]
    \end{restatedlemma}

(Since the hypothesis \cref{e-stable-focusing-h} concerns fixed values of $R,R',A,A'$, the condition on $\eps,\kappa,\alpha$ is to be read as follows: they are smaller than a positive threshold which depends on $R,R',A,A',\tilde A',\kappa',\Psi,\Psi'$. If $A=\infty$, then $A+\alpha=\infty$.)
\begin{proof}
Let $M=\{(z,w)\in\C^2:|z|=R'\}$, and let $V\subset M$ be a relatively open set with $\mathcal S_{\Psi'}(R',A')\subset V\subset\tilde{\mathcal S}_{\Psi'}(R',\tilde A',\kappa')$, which exists since $A'<\tilde A'$ and $\kappa'>0$. By \cref{e-stable-focusing-h}, every solution for $\mathcal f$ with data in $\mathcal S_\Psi(R,A)$ reaches $M$ before the time $T$ of \cref{defn-S-statement}, at a point of $\mathcal S_{\Psi'}(R',A')$, where $\Re\p{\bar zw}>0$ since $A'<\infty$; so the crossing is transversal. In other words, the set $\mathcal W(M,V,T)$ of \cref{e-crossing-open}, which is open by \cref{l-continuity}, contains $\{\mathcal f\}\times\mathcal S_\Psi(R,A)$. Since $\mathcal S_\Psi(R,A)$ is compact and the compact sets $\tilde{\mathcal S}_\Psi(R,A+\delta,\delta)$ decrease to it as $\delta\to0$ (\cref{e-S-intersection}), $\mathcal W(M,V,T)$ contains $\mathcal N\times\tilde{\mathcal S}_\Psi(R,A+\delta,\delta)$ for some neighborhood $\mathcal N$ of $\mathcal f$ in $\Hsp$ and some $\delta>0$. Finally, $F_\eps\subset\mathcal N$ for $\eps\lll\mathcal N$, since $\mathcal N$ only constrains $\tilde{\mathcal f}$ on a compact subset of $\R\times\C$; as $F_\eps$ is invariant under time translations, this gives \cref{e-stable-conclusion} for $\kappa,\alpha\leq\delta$.
\end{proof}

\subsection{Small scale (\texorpdfstring{\cref{l-small-scale-focusing}}{Lemma \ref{l-small-scale-focusing}})}\label{ss-small-scale}
Recall the statement of \cref{l-small-scale-focusing}: if $R\lll R'\lll A'<\infty=A$, then $\mathcal S_\Psi(R,\infty)\prec\mathcal S_\Psi(R',A')$. After rescaling by $R'$, the function $\mathcal f(z)=1-e^{-z^n}$ becomes a small perturbation of the model function $z^n$, for which \cref{c-model} gives focusing at a definite rate; since the rescaling is an exact symmetry of the problem, no estimates are needed beyond the continuity of the first crossing (\cref{l-continuity}).

Since any two aperture functions are comparable (\cref{d-aperture}), we may assume without loss of generality that $\Psi=\Phi$ as in \cref{c-model} ($\Phi$ is an aperture function by \cref{c-model}(3c)). It suffices to prove that if $R'\lll s\lll A'$, then $\mathcal S_\Phi(sR',\infty)\prec\mathcal S_\Phi(R',A')$. Indeed, given $R\lll R'\lll A'$, there is an integer $k\geq1$ for which $s:=(R/R')^{1/k}$ satisfies $R'\lll s\lll A'$ (take the smallest $k$ for which $(R/R')^{1/k}$ is at least the square of the threshold in $s\lll A'$; then $s$ is at most that threshold, using $R\lll R',A'$ when $k=1$, and $R'\lll s$ follows from $R'\lll A'$). Since $s^jR'\leq R'$ for $j\geq0$, the radii $R^{(j)}:=s^{k-j}R'$, $0\leq j\leq k$, all satisfy $R^{(j)}\lll s$, so $\mathcal S_\Phi(R^{(j-1)},\infty)\prec\mathcal S_\Phi(R^{(j)},A')\subset\mathcal S_\Phi(R^{(j)},\infty)$ for each $j$, and \cref{l-small-scale-focusing} follows from the transitivity of $\prec$. For any $(\bm z,\bm w)$ that solves \cref{z-w-ode} and any $\rho>0$, the functions
    \[\label{e-small-scale-rescaling}
        \bm z_\rho(\tau):=\rho\inv\bm z\p{\rho\ief{n-2}2\tau},\qquad\bm w_\rho(\tau):=\rho\ief n2\bm w\p{\rho\ief{n-2}2\tau}\]
solve \cref{z-w-ode} with $\mathcal f$ replaced by $\mathcal f_\rho(z):=\rho\ie n\mathcal f(\rho z)$, and $\lim_{\rho\to0}\mathcal f_\rho=\mathcal f_0$ in $\Hsp$, where $\mathcal f_0(z):=z^n$. Denoting by $\mathcal S_\Phi(\rho;R,A)$ the corresponding region for $\mathcal f_\rho$ as in \cref{d-S-region}, and using $\prec$ for these regions as well, the problem reduces to showing $\mathcal S_\Phi(\rho;s,\infty)\prec\mathcal S_\Phi(\rho;1,A')$ whenever $\rho\lll s\lll A'$. Note that $\mathcal S_\Phi(0;s,\infty)$ is precisely the set of data \cref{e-model-hypotheses} of \cref{c-model} with $|z|=s$. By \cref{c-model}(1), every solution for $\mathcal f_0$ with such data reaches $M:=\{(z,w)\in\C^2:|z|=1\}$ transversely, with $w\neq0$, $\Re(z^n)>0$ and $\Re\p{\bar zw}>0$, and by \cref{e-Phi-decay} and $s\lll A'$, at this point
    \[\label{Phi-bound-R-Ri0}
        \Phi(\theta,\varphi)=s^{\beta n}\Phi\p{\bm\theta(0),\bm\varphi(0)}\leq s^{\beta n}\sup_{\overline\Qsq}\Phi<{A'}^2.\]
Thus, as in \cref{ss-stability}, there is an open set $W\subset\Hsp\times\C^2$ containing $\{\mathcal f_0\}\times\mathcal S_\Phi(0;s,\infty)$ so that every solution with data in $W$ reaches $M$ before a fixed time, at a point of $\{w\neq0,\ \Re(z^n)>0,\ \Re(\bar zw)>0,\ \Phi(\theta,\varphi)<{A'}^2\}$; by conservation of energy, this point lies in $\mathcal S_\Phi(\rho;1,A')$ for a solution for $\mathcal f_\rho$. This completes the proof, since for $\rho\lll s,A'$ we have $\{\mathcal f_\rho\}\times\mathcal S_\Phi(\rho;s,\infty)\subset W$: indeed, $\mathcal S_\Phi(\rho;s,\infty)$ is contained in the compact set $\{|z|=s,\ \Re(\bar zw)\geq0,\ \abs{|w|^2-\Re(z^n)}\leq\sup_{D_1}|\mathcal f_\rho-\mathcal f_0|\}$, and these sets decrease to $\mathcal S_\Phi(0;s,\infty)$ as $\rho\to0$.

\subsection{Medium-small scale (\texorpdfstring{\cref{l-medium-small-scale-focusing}}{Lemma \ref{l-medium-small-scale-focusing}})}\label{ss-medium-small-scale}
        We set
            \[\label{def-Psi-2}
                \PsiS(\theta,\varphi):=\tfrac12\theta^2+\beta\theta\varphi+\varphi^2
                \]
        where, as in \cref{ss-coordinate-transformation}, $\beta=\tfrac12+\tfrac1n$; this is the quadratic form $\tilde\Phi$ of \cref{c-model}(3). Let $(\bm z(t),\bm w(t))$ denote some solution to \cref{z-w-ode} with $(\bm z(t_0),\bm w(t_0))\in\mathcal S_\PsiS(R,A)$, and let $\bm r,\bm\theta,\bm\varphi$ be as in \cref{l-cov}. The hypotheses \cref{e-hyp} hold at $t_0$: $\bm z(t_0)\neq0$, $|\bm\varphi(t_0)|<\frac\pi2$ gives $\frac d{dt}|\bm z|(t_0)>0$, and $|\bm\theta(t_0)|<\frac\pi2$ gives $\Re\p{\mathcal f(\bm z(t_0))}>0$ by \cref{e-Re-f-positive}; they persist as long as $|\bm\theta|,|\bm\varphi|<\frac\pi2$, which will be guaranteed by the estimate $\PsiS\ll1$ that we prove below (by a continuity argument). Since the data has zero energy, $\bm\kappa\equiv0$ and $\bm c\equiv0$. Writing $\bm\omega:=\bm z^n=e^{\bm r+i\bm\theta}$, we have for $\mathcal f(z)=1-e^{-z^n}$
            \[
                \bm G=\tfrac1n\bm z\mathcal f'(\bm z)=\bm\omega e^{-\bm\omega},\qquad \Re\p{\bm f}=1-\Re\p{e^{-\bm\omega}},
                \]
        so that \cref{e-theta-prime,e-varphi-prime} read
            \begin{align}\label{theta-varphi-ode-u}
                \bm\theta'&=\tan\bm\varphi\\
                \bm\varphi'&=-\tfrac1n\tan\bm\varphi-\tfrac12\p{\frac{\Im\p{\bm\omega e^{-\bm\omega}}}{1-\Re\p{e^{-\bm\omega}}}+\frac{\Re\p{\bm\omega e^{-\bm\omega}}}{1-\Re\p{e^{-\bm\omega}}}\tan\bm\varphi}\nonumber
                \end{align}
        We would like to show that
            \[\PsiS(\bm\theta(r),\bm\varphi(r))\ll1,\quad\text{and}\quad r\ll0\]
        implies
            \[\label{e-PsiS-derivative}
                \frac{d}{dr}\PsiS(\bm\theta(r),\bm\varphi(r))\leq0
                \]
        This would guarantee $\PsiS(\bm\theta(r),\bm\varphi(r))\leq A^2$ along the whole trajectory, since $R'\ll1$ implies $r\ll0$. To see why \cref{e-PsiS-derivative} is true, we expand the coefficients in \cref{theta-varphi-ode-u} for $|\bm\omega|=e^{r}\ll1$. Since $\bm\omega e^{-\bm\omega}=\sum_{k\geq0}\frac{(-1)^k}{k!}\bm\omega^{k+1}$, $\Im\p{\bm\omega^{k+1}}=e^{(k+1)r}\sin((k+1)\bm\theta)$ and $|\sin((k+1)\bm\theta)|\leq(k+1)|\sin\bm\theta|$, we have $\Im\p{\bm\omega e^{-\bm\omega}}=e^{r}\sin\bm\theta\,(1+O(e^r))$ and, similarly, $\Re\p{\bm\omega e^{-\bm\omega}}=e^r\cos\bm\theta+O(e^{2r})$ and $1-\Re\p{e^{-\bm\omega}}=e^r\cos\bm\theta+O(e^{2r})$. Hence, when $|\bm\theta|\ll1$ (so that $\cos\bm\theta\sim1$),
            \begin{align}\label{e-Psi-derivative-computing}
                \bm\varphi'
                    =-&\tfrac1n\tan(\bm\varphi)-\tfrac12\p{\frac{e^r\sin\bm\theta\,(1+O(e^r))}{e^r\cos\bm\theta\,(1+O(e^r))}+\frac{e^r\cos\bm\theta\,(1+O(e^r))}{e^r\cos\bm\theta\,(1+O(e^r))}\tan(\bm\varphi)}\\
                    =-&\tfrac12\tan(\bm\theta)-\beta\tan(\bm\varphi)+O(e^r\tan(\bm\theta))+O(e^r\tan(\bm\varphi))\\
                    =-&\tfrac12\bm\theta-\beta\bm\varphi+O((\PsiS(\bm\theta,\bm\varphi))\ef32)+O(e^r(\PsiS(\bm\theta,\bm\varphi))\ef12)\\
                \bm\theta'
                    =&\tan(\bm\varphi)\\
                    =&\bm\varphi+O((\PsiS(\bm\theta,\bm\varphi))\ef32)
                \end{align}
        which holds whenever $r\ll0$ and $|\bm\theta|,|\bm\varphi|\ll1$, i.e. whenever $r\ll0$ and $\PsiS(\bm\theta,\bm\varphi)\ll1$. In other words, \cref{theta-varphi-ode-u} is the model system \cref{e-model} up to errors of relative size $O(e^r)$. Therefore, using $\nabla\PsiS\cdot(\varphi,-\tfrac12\theta-\beta\varphi)=-\beta\PsiS$ and $|\nabla\PsiS|\lesssim\PsiS\ef12$, we have
            \[\aline\label{e-still-comuting-d-dr-PsiS}
                \frac d{dr}(\PsiS(\bm\theta(r),\bm\varphi(r)))
                    =&-\beta\PsiS+O(\PsiS\ef12)\p{O(\PsiS\ef32)+O(e^r\PsiS\ef12)}\\
                    =&(-\beta+O(\PsiS(\bm\theta(r),\bm\varphi(r))+e^r))\PsiS(\bm\theta(r),\bm\varphi(r))\\
                    \leq&0
                \]
        which holds whenever $r\ll0$ and $\PsiS(\bm\theta(r),\bm\varphi(r))\ll1$. Since $A\ll1$, this shows that $\PsiS(\bm\theta(r),\bm\varphi(r))\leq A^2$ for all $r\in[n\log R,n\log R']$, and in particular that the trajectory reaches $|\bm z|=R'$ with $(\bm\theta,\bm\varphi)\in\mathcal A(\PsiS,A')$. To show that $|\bm z(t_1)|=R'$ for some time $t_1\leq t_0+T$, for some constant time $T$, note that $\PsiS(\bm\theta,\bm\varphi)\ll1$ implies that $|\bm\varphi|\ll1$ and hence $\frac d{dt}|\bm z|=\sqrt2|\bm w|\cos\bm\varphi\sim|\bm w|$ (see \cref{e-zbarw}). Since $|\bm\theta|\ll1$, $|\bm z|\leq R'\ll1$ and $|\bm z|$ is increasing, we have $|\bm w|^2=\Re\p{\mathcal f(\bm z)}\gtrsim|\bm z|^n\geq|\bm z(t_0)|^n=R^n$, i.e. $|\bm w|\gtrsim R\ef n2$. We conclude that $|\bm z(t)|-R\gtrsim(t-t_0)R\ef n2$, so we may take $T\sim R'R\ief n2$, for example.
\subsection{Medium scale (\texorpdfstring{\cref{l-medium-scale-focusing}}{Lemma \ref{l-medium-scale-focusing}})}\label{ss-medium-scale}
Recall the statement of \cref{l-medium-scale-focusing}: for any aperture function $\Psi$, if $A\lll A',R,R'$ and $R\leq R'$, then $\mathcal S_\Psi(R,A)\prec\mathcal S_\Psi(R',A')$.
\begin{proof}[Proof of \cref{l-medium-scale-focusing}]
For $R=R'$ this is the inclusion $\mathcal S_\Psi(R,A)\subset\mathcal S_\Psi(R,A')$. For $R<R'$, it follows from $\mathcal S_\Psi(R,0)\prec\mathcal S_\Psi(R',0)$ (the radial solutions remain radial, and $|\bm z|$ increases along them since $\mathcal f>0$ on $(0,\infty)$), along with \cref{l-stable-focusing} with $A=A'=0$.
\end{proof}

\subsection{Large scale (\texorpdfstring{\cref{l-large-scale-focusing}}{Lemma \ref{l-large-scale-focusing}})}\label{ss-large-scale}
At large scales and inside a small aperture, $\mathcal f'$ is negligible, so the solutions are almost straight lines traversed at constant speed; in other words, we are comparing with the equation for the zero function. For a straight line, the direction $\arg\bm w=\arg\bm z+\bm\varphi$ is constant, so that $\bm\theta+n\bm\varphi=n\arg\bm w$ (mod $2\pi$) is constant, while $|\bm\varphi|$ decreases. This explains the choice of the aperture function $\PsiL$ below: $|\bm\theta|+n|\bm\varphi|$ does not increase along straight lines.
We define
    \[\label{def-Psi-3}
        \PsiL(\theta,\varphi)=(|\theta|+n|\varphi|)^2
        \]
Let $(\bm z(t),\bm w(t))$ denote any solution to \cref{z-w-ode} with $(\bm z(t_0),\bm w(t_0))\in\mathcal S_\PsiL(R,A)$. Let $(\bm r,{\bm\theta},{\bm\varphi})$ be as defined in \cref{l-cov}. As in \cref{ss-medium-small-scale}, the hypotheses \cref{e-hyp} hold at $t_0$ and persist as long as $|\bm\theta|,|\bm\varphi|<\frac\pi2$, $\bm c\equiv0$, and $(\bm\theta,\bm\varphi)$ satisfy \cref{theta-varphi-ode-u}, i.e.
    \begin{align}\label{theta-varphi-ode-large}
        \bm\theta'&=\tan\bm\varphi\\
        \bm\varphi'&=-\tfrac1n\tan\bm\varphi-\tfrac12\p{\frac{\Im\p{\bm G}}{\Re\p{\bm f}}+\frac{\Re\p{\bm G}}{\Re\p{\bm f}}\tan\bm\varphi}\nonumber
        \end{align}
with $\bm G=\bm\omega e^{-\bm\omega}$ and $\Re\p{\bm f}=1-\Re\p{e^{-\bm\omega}}$, where $\bm\omega=\bm z^n=e^{\bm r+i\bm\theta}$. If $r_0=\bm r(t_0)=n\log(R)$, then we have
    \[\PsiL({\bm\theta}(r_0),{\bm\varphi}(r_0))\leq A^2\]
It suffices to show that $(\bm\theta,\bm\varphi)$ can be extended to $r=r'=n\log(R')$ and satisfies
    \[
        \PsiL({\bm\theta}(r'),\bm\varphi(r'))\leq {A'}^2
        \]
Note that $r\geq r_0\gg0$. If $\PsiL(\bm\theta,\bm\varphi)\ll1$, i.e. $|\bm\theta|,|\bm\varphi|\ll1$, then $\Re\p{\bm\omega}=e^{r}\cos\bm\theta\geq\tfrac12e^r$, so $|e^{-\bm\omega}|\leq e^{-e^r/2}$, and hence
    \[\label{e-large-scale-coefficients}
        \abs{\frac{\Im\p{\bm G}}{\Re\p{\bm f}}}+\abs{\frac{\Re\p{\bm G}}{\Re\p{\bm f}}}\lesssim\frac{|\bm\omega e^{-\bm\omega}|}{1-|e^{-\bm\omega}|}\lesssim e^{r-e^r/2}\leq e^{-r}
        \]
By \cref{def-Psi-3} and \cref{theta-varphi-ode-large}, we find that (using one-sided derivatives where $\bm\theta=0$ or $\bm\varphi=0$)
    \[
        \frac{d}{dr}\Big[\sqrt{\PsiL({\bm\theta}(r),{\bm\varphi}(r))}\Big]=\frac{d}{dr}\Big[|\bm\theta(r)|+n|\bm\varphi(r)|\Big]\leq\operatorname{sgn}(\bm\theta)\tan\bm\varphi-\operatorname{sgn}(\bm\varphi)\tan\bm\varphi+O(e^{-r})\lesssim e^{-r}
        \]
since $\operatorname{sgn}(\bm\theta)\tan\bm\varphi\leq|\tan\bm\varphi|=\operatorname{sgn}(\bm\varphi)\tan\bm\varphi$. Integrating from $r_0$, we conclude that
    \[\label{e-large-scale-conclusion}
        \sqrt{\PsiL({\bm\theta}(r),{\bm\varphi}(r))}\lesssim A+e^{-r_0}\lll A'\qquad\text{for all }r\geq r_0,
        \]
which is in particular true at $r=r'$; the hypothesis $\PsiL\ll1$ used above is justified by a continuity argument, since $A+e^{-r_0}\ll1$. Note that \cref{e-large-scale-conclusion} holds for \textit{every} $r\geq r_0$ for which the solution exists, not just at $r=r'$: the hypotheses \cref{e-hyp} persist on the whole forward interval of existence, so $|\bm z|$ is strictly increasing there and the trajectory stays inside the aperture $\mathcal A(\PsiL,A')$ for all later times; in particular $|\bm\theta|$ stays small. This is the last assertion of \cref{l-large-scale-focusing}. To show that $|\bm z(t_1)|=R'$ for some time $t_1\leq t_0+T$, where $T$ is some constant, we may apply the same argument as the one at the end of \cref{ss-medium-small-scale}: now $\frac d{dt}|\bm z|=\sqrt2|\bm w|\cos\bm\varphi\sim|\bm w|$ and $|\bm w|^2=\Re\p{\mathcal f(\bm z)}\sim1$, so we may take $T\sim R'$.

To see why the same is true for solutions to \cref{z-w-tilde-f-ode} with $\tilde{\mathcal f}\in F_\eps$ and with data in $\tilde{\mathcal S}_\PsiL(R,A,\kappa)$ at some time $t_0$, note that \cref{l-cov} applies verbatim to $\tilde{\mathcal f}$, which is of class $C^1$ (with $\tilde{\bm G}=\tfrac1n\tilde{\bm z}\tilde{\mathcal f}'(t;\tilde{\bm z})$; note that \cref{e-theta-prime,e-varphi-prime} do not involve $\bm f_t$), and that the only properties of ${\mathcal f}$ that we have used are as follows:
    \begin{enumerate}
        \item $\Re\p{\mathcal f(z)}\sim1$ when $|\arg(z^n)|\leq\frac1{10}$ and $|z|\gg0$.
        \item $|\tfrac1nz\mathcal f'(z)|\lesssim e^{-|z|^n/3}$ (and hence $\lesssim|z|^{-n}=e^{-r}$, which is what was used in \cref{e-large-scale-coefficients}) when $|\arg(z^n)|\leq\frac1{10}$ and $|z|\gg0$.
        \end{enumerate}
Both of these properties are also met by $\tilde {\mathcal f}(t;z)=1-e^{-z^n+\mathcal p(t;z)}$, uniformly in $t$, since $\Re(z^n-\mathcal p(t;z))\geq\tfrac12|z|^n$ when $|\arg(z^n)|\leq\frac1{10}$ and $|z|\gg0$ (recall that $\deg\mathcal p<n$ and the coefficients of $\mathcal p$ have modulus $\leq\eps\leq1$; see \cref{e-main-lemma-decay}), so $|e^{-z^n+\mathcal p(t;z)}|\leq e^{-|z|^n/2}$ and $\tfrac1nz\tilde{\mathcal f}'(t;z)=(z^n-\tfrac1nz\mathcal p'(t;z))e^{-z^n+\mathcal p(t;z)}$ has modulus $\lesssim|z|^ne^{-|z|^n/2}\leq e^{-|z|^n/3}$. In the same region, we also have
    \begin{enumerate}
        \item[(3)] $|\tilde{\mathcal f}(t;z)-\mathcal f(z)|\leq2e^{-|z|^n/2}$ and $|\tilde{\mathcal f}'(t;z)-\mathcal f'(z)|\lesssim e^{-|z|^n/3}$,
        \end{enumerate}
where the second bound follows from (2) for $\mathcal f$ and for $\tilde{\mathcal f}$, since $|z|\geq1$. The hypotheses \cref{e-hyp} hold at $t_0$ for the same reasons as before, since data in $\tilde{\mathcal S}_\PsiL(R,A,\kappa)$ with $A\ll1$ have $\tilde{\bm w}\neq0$ and $|\tilde{\bm\varphi}|<\frac\pi2$, and $\Re\p{\tilde{\mathcal f}(t;\tilde{\bm z})}\sim1$ by (1).

The energy $\bm\kappa(t)=|\tilde{\bm w}|^2-\Re\p{\tilde{\mathcal f}(t;\tilde{\bm z})}$ is no longer constant, but it stays small. To see this, we use the energy with respect to $\mathcal f$, namely $\bm\kappa_{\mathcal f}(t):=|\tilde{\bm w}|^2-\Re\p{\mathcal f(\tilde{\bm z})}$; by (3), $|\bm\kappa-\bm\kappa_{\mathcal f}|\leq2e^{-|\tilde{\bm z}|^n/2}$, and $|\bm\kappa_{\mathcal f}(t_0)|\leq\kappa$ by condition (3) of \cref{d-S-region}. Suppose that, on $[t_0,t]$, we have $|\tilde{\bm z}|\geq R$, $|\tilde{\bm\theta}|\leq\frac1{10}$, $\tilde{\bm w}\neq0$ and $|\tilde{\bm\varphi}|\leq1$. Then $\frac d{dt}|\tilde{\bm z}|=|\dot{\tilde{\bm z}}|\cos\tilde{\bm\varphi}\geq\frac12|\dot{\tilde{\bm z}}|>0$ by \cref{e-zbarw}, so we may use $\varrho=|\tilde{\bm z}|$ as the variable of integration in \cref{e-energy-drift-f}, and (3) gives
    \[|\bm\kappa_{\mathcal f}(t)-\bm\kappa_{\mathcal f}(t_0)|\leq\int_{t_0}^t|\dot{\tilde{\bm z}}|\,\abs{\tilde{\mathcal f}'(t';\tilde{\bm z})-\mathcal f'(\tilde{\bm z})}dt'\lesssim\int_R^\infty e^{-\varrho^n/3}d\varrho\leq3e^{-R^n/3},\]
because $\varrho^n\geq R^{n-1}\varrho\geq\varrho$ for $\varrho\geq R\geq1$. (Here it is essential that the trajectory moves outwards: the bound is in terms of the length of the trajectory, not of the elapsed time.) Therefore
    \[\label{e-large-scale-energy}
        |\bm\kappa_{\mathcal f}(t)|\leq\kappa+O(e^{-R^n/3})\qquad\text{and}\qquad|\bm\kappa(t)|\leq\kappa+O(e^{-R^n/3}).
        \]
Hence the additional factor $\frac1{1+\bm c}$ in \cref{e-varphi-prime} satisfies $\bm c=\bm\kappa/\Re\p{\tilde{\bm f}}=O(\kappa+e^{-R^n/3})$, which is $\ll1$ if $\kappa_0\ll1$ and $R\gg0$ (this is where the hypothesis $\kappa\leq\kappa_0$ of \cref{l-large-scale-focusing} is used; the hypothesis $\eps\leq\kappa_0\leq1$ was used for (1)--(3); they also give $|\tilde{\bm w}|^2=(1+\bm c)\Re\p{\tilde{\bm f}}\sim1$). This factor only multiplies the terms bounded in \cref{e-large-scale-coefficients}, so the argument above goes through unchanged: by the same continuity argument, $|\tilde{\bm z}|$ is strictly increasing, the bound \cref{e-large-scale-conclusion} holds (so that, in particular, the solution stays in the region $|\tilde{\bm z}|\geq R$, $|\tilde{\bm\theta}|\leq\frac1{10}$, $\tilde{\bm w}\neq0$, $|\tilde{\bm\varphi}|\leq1$ in which (1)--(3) and \cref{e-large-scale-energy} are valid) for as long as the solution exists, and $|\tilde{\bm z}|$ reaches $R'$ at some time $t_1\leq t_0+T$ with $T\sim R'$. Finally, at the time $t_1$, by \cref{e-large-scale-energy} and (3),
    \[\abs{|\tilde{\bm w}(t_1)|^2-\Re\p{\mathcal f(\tilde{\bm z}(t_1))}}=|\bm\kappa_{\mathcal f}(t_1)|\leq\kappa+O(e^{-R^n/3})\leq1\]
(if $\kappa_0\leq\frac14$ and $R\gg0$), which is condition (3) of \cref{d-S-region} for $\tilde{\mathcal S}_\PsiL(R',A',1)$. This completes the proof of \cref{l-large-scale-focusing}.

\section{A proof of the polynomial case of the conjecture}\label{s-polynomial-conjecture}
In this section, we recover the primary result of \cite{Flores_2020}.%
    \footnote{%
        The result in \cite{Flores_2020} is stated for $C^1$ harmonic polynomials (\cite[Theorem 1.1]{Flores_2020}); our result is for continuous harmonic polynomials. The specific incomplete trajectories constructed in \cite{Flores_2020} also differ from ours.
        }%
That is, we show (\cref{t-polynomial}) that, if $\mathcal f(t;z)$ is a polynomial in $z$ of degree $n\geq3$ whose coefficients are continuous functions of $t$, then there is at least one incomplete solution to $\ddot{\bm z}(t)=\overline{\mathcal f'(t;\bm z(t))}$; in terms of the metric, $H(u,\cdot)=-2\Re\p{\mathcal f(u;\cdot)}$ is a harmonic polynomial of degree $\geq3$ (see \cref{ss-reduction}). This does not require any substantial new insights. The point is that the argument becomes short with the tools of the preceding sections.

\begin{theorem}\label{t-polynomial}
    Let $n\geq3$ and let $\mathcal f:(t_-,t_+)\times\C\to\C$ be a function of the form $\mathcal f(t;z)=\sum_{j=0}^nc_j(t)z^j$, where each $c_j:(t_-,t_+)\to\C$ is continuous, and let $t_*$ be so that $c_n(t_*)\neq0$. Then there exists $\delta>0$ so that for $|t_0-t_*|\ll1$ and $R\gg1$, the maximal forward solution $(\bm z(t),\bm w(t))$ to \cref{e-ode} whose initial data $(\bm z(t_0),\bm w(t_0))$ lies in $\mathcal K_R(\delta)$ ceases to exist within time $\lesssim R\ief{n-2}2$, before $t_+$. Here, $\mathcal K_R(\delta)$ denotes the set of all $(z,w)\in\C^2$ which satisfy the following:
        \begin{enumerate}
            \item $|z|=R$
            \item $\Re(\bar zw)\geq0$
            \item $\abs{|w|^2-\Re\p{c_n(t_*)z^n}}\leq\delta|z|^n$.
            \end{enumerate}
\end{theorem}

\begin{proof}
Assume without loss of generality that $t_*=0$ and $c_n(0)=1$ and let $M=\{(z,w)\in\C^2:|z|=2\}$. We claim that for some $\delta>0$, $R\gg0$, and $T\ll1$, the following is true. Let $(\bm z,\bm w)$ be any solution to \cref{e-ode} with $(\bm z(t_0),\bm w(t_0))\in\mathcal K_R(\delta)$, for some $|t_0|\leq T$. Then, for some time $t_1$ with $t_1-t_0\lesssim R\ief{n-2}2$, we have $(\bm z(t_1),\bm w(t_1))\in\mathcal K_{2R}(\delta)$. This suffices to prove \cref{t-polynomial}; if the doubling times decay exponentially, they are summable and the solution cannot be extended to all time. To see why this is true, we will invoke a symmetry of the flow $\mathcal F$:
    \[\mathcal F(\mathcal f_{1,t_0},z,w)=\sigma_R\circ\tau_R\circ\mathcal F(\mathcal f_{R,t_0},\sigma_R\inv(z,w))\]
Here, $\mathcal f_{\lambda,t_0}(t;z):=\lambda\ie n\mathcal f(\lambda\ief{n-2}2t+t_0;\lambda z)$, $\tau_\lambda(\bm z,\bm w)(t)=(\bm z(\lambda\ef{n-2}2t),\bm w(\lambda\ef{n-2}2t))$, and $\sigma_\lambda(\bm z,\bm w)=(\lambda\bm z,\lambda\ef n2\bm w)$. Since $\sigma_R(\mathcal K_1(\delta))=\mathcal K_R(\delta)$, it suffices to show that $(t_M,p_M)\circ\mathcal F(\mathcal f_{R,t_0},\mathcal K_1(\delta))$ is well-defined and lies in a compact subset of $\R\times\mathcal K_2(\delta)$ (independent of $R,t_0$). Note that $\lim_{R\to\infty,t_0\to 0}\mathcal f_{R,t_0}=\mathcal f_\infty$, where $\mathcal f_\infty(t;z):=z^n$. By \cref{c-model}, $p_M\circ\mathcal F(\mathcal f_\infty,\mathcal K_1(0))$ is well-defined and lies in an open set $U\subset\mathcal K_2(\infty)$. Therefore, $(p_M\circ\mathcal F)\inv(U)=:V$ is an open set containing $\{\mathcal f_\infty\}\times\mathcal K_1(0)$. Let $\delta>0$ be small enough so that $\{\mathcal f_\infty\}\times\mathcal K_1(\delta)\subset V$, and note that $p_M\circ\mathcal F(\mathcal f_\infty,\mathcal K_1(\delta))$ is well defined and lies in an open subset $U'\subset\mathcal K_2(\delta)$, by conservation of energy. Finally, note that $(p_M\circ\mathcal F)\inv(U')$ is an open set containing $\{\mathcal f_\infty\}\times\mathcal K_1(\delta)$; since $\lim_{R\to\infty,t_0\to0}\mathcal f_{R,t_0}=\mathcal f_\infty$, it follows that for sufficiently large $R$ and small $|t_0|$ we have $\{\mathcal f_{R,t_0}\}\times\mathcal K_1(\delta)\subset K\subset(p_M\circ\mathcal F)\inv(U')$ for some compact $K$. This completes the proof.
\end{proof}

\printbibliography
\end{document}